\PassOptionsToPackage{unicode}{hyperref}
\PassOptionsToPackage{hyphens}{url}
\PassOptionsToPackage{dvipsnames,svgnames,x11names}{xcolor}
\documentclass[11pt,a4paper]{article}
\usepackage{amsmath,amssymb}
  \usepackage[T1]{fontenc}
  \usepackage[utf8]{inputenc}
  \usepackage{textcomp} 
\usepackage{lmodern}

\usepackage{xcolor}
\makeatletter
\ifx\paragraph\undefined\else
  \let\oldparagraph\paragraph
  \renewcommand{\paragraph}{
    \@ifstar
      \xxxParagraphStar
      \xxxParagraphNoStar
  }
  \newcommand{\xxxParagraphStar}[1]{\oldparagraph*{#1}\mbox{}}
  \newcommand{\xxxParagraphNoStar}[1]{\oldparagraph{#1}\mbox{}}
\fi
\ifx\subparagraph\undefined\else
  \let\oldsubparagraph\subparagraph
  \renewcommand{\subparagraph}{
    \@ifstar
      \xxxSubParagraphStar
      \xxxSubParagraphNoStar
  }
  \newcommand{\xxxSubParagraphStar}[1]{\oldsubparagraph*{#1}\mbox{}}
  \newcommand{\xxxSubParagraphNoStar}[1]{\oldsubparagraph{#1}\mbox{}}
\fi
\makeatother

\usepackage{longtable,booktabs,array}
\usepackage{calc} 
\usepackage{etoolbox}
\makeatletter
\patchcmd\longtable{\par}{\if@noskipsec\mbox{}\fi\par}{}{}
\makeatother
\usepackage{footnote}
\makesavenoteenv{longtable}
\usepackage{graphicx}
\makeatletter
\def\maxwidth{\ifdim\Gin@nat@width>\linewidth\linewidth\else\Gin@nat@width\fi}
\def\maxheight{\ifdim\Gin@nat@height>\textheight\textheight\else\Gin@nat@height\fi}
\makeatother
\setkeys{Gin}{width=\maxwidth,height=\maxheight,keepaspectratio}
\makeatletter
\def\fps@figure{htbp}
\makeatother
\makeatletter
\@ifpackageloaded{caption}{}{\usepackage{caption}}
\AtBeginDocument{%
\ifdefined\contentsname
  \renewcommand*\contentsname{Table of contents}
\else
  \newcommand\contentsname{Table of contents}
\fi
\ifdefined\listfigurename
  \renewcommand*\listfigurename{List of Figures}
\else
  \newcommand\listfigurename{List of Figures}
\fi
\ifdefined\listtablename
  \renewcommand*\listtablename{List of Tables}
\else
  \newcommand\listtablename{List of Tables}
\fi
\ifdefined\figurename
  \renewcommand*\figurename{Figure}
\else
  \newcommand\figurename{Figure}
\fi
\ifdefined\tablename
  \renewcommand*\tablename{Table}
\else
  \newcommand\tablename{Table}
\fi
}
\@ifpackageloaded{float}{}{\usepackage{float}}
\floatstyle{ruled}
\@ifundefined{c@chapter}{\newfloat{codelisting}{h}{lop}}{\newfloat{codelisting}{h}{lop}[chapter]}
\floatname{codelisting}{Listing}

\makeatother
\makeatletter
\@ifpackageloaded{caption}{}{\usepackage{caption}}
\@ifpackageloaded{subcaption}{}{\usepackage{subcaption}}
\makeatother

\usepackage[]{natbib}
\usepackage{bookmark}
\hypersetup{
  pdftitle={MMP},
  pdfauthor={Author 1; Author 2},
  pdfkeywords={3 to 6 keywords, that do not appear in the title},
  colorlinks=true,
  linkcolor={blue},
  filecolor={Maroon},
  citecolor={Blue},
  urlcolor={Blue},
  pdfcreator={LaTeX via pandoc}}

\newcommand{\anon}{1}

\usepackage{setspace}
\usepackage{amsmath}
\usepackage[dvipsnames]{xcolor}
\usepackage{subcaption}
\usepackage{bbm}
\newcommand{\iid}{\overset{\mathrm{iid}}{\sim}}
\DeclareMathOperator*{\argmax}{arg\,max}

\usepackage{graphicx}
\graphicspath{ {./images/} }
\usepackage{algorithm}
\usepackage{algpseudocode}
\usepackage{amsmath, amsthm, amssymb}
\usepackage{url}
\theoremstyle{plain}

\newtheorem{theorem}{Theorem}
\newtheorem{lemma}{Lemma}
\newtheorem{prop}{Proposition}
\newtheorem{example}{Example}
\newtheorem{corollary}{Corollary}
\theoremstyle{remark}

\usepackage{amssymb}
\theoremstyle{remark}
\usepackage[font=small,labelfont=bf]{caption}

\usepackage{xr-hyper}
\usepackage{hyperref}

\makeatletter
\newcommand*{\addFileDependency}[1]{
\typeout{(#1)}
\@addtofilelist{#1}
\IfFileExists{#1}{}{\typeout{No file #1.}}
}\makeatother

\begin{document}

\def\spacingset#1{\renewcommand{\baselinestretch}%
{#1}\small\normalsize} \spacingset{1}


\if1\anon
{
  \title{Moment Martingale Posteriors \\ for Semiparametric Predictive Bayes}
  \author{Yiu Yin Yung$^1$, Stephen M.S. Lee$^1$ and Edwin Fong$^1$\\
  \\
  \begin{small}
  $^1$Department of Statistics and Actuarial Science, University of Hong Kong
    \end{small}
}
  \date{}
  \maketitle
} \fi

\if0\anon
{
  \bigskip
  \begin{center}
    {\LARGE\bf Moment Martingale Posteriors \\\vspace{1mm} for Semiparametric Predictive Bayes}
\end{center}
  \medskip
} \fi

\begin{abstract}
The predictive Bayesian view involves eliciting a sequence of one-step-ahead predictive distributions in lieu of specifying a likelihood function and prior distribution. 
Recent methods have leveraged predictive distributions which are either nonparametric or parametric, but not a combination of the two. This paper introduces a semiparametric martingale posterior which utilizes a  predictive distribution that is a mixture of a parametric and nonparametric component. The semiparametric nature of the predictive allows for regularization of the nonparametric component when the sample size is small, and robustness to model misspecification of the parametric component when the sample size is large. We call this approach the \textit{moment martingale posterior}, as the core of our proposed methodology is to utilize the method of moments as the vehicle for tying the nonparametric and parametric components together. In particular, the predictives are constructed so that the moments are martingales, which allows us to verify convergence under predictive resampling. A key contribution of this work is a novel procedure based on the energy score to optimally weigh between the parametric and nonparametric components, which has desirable asymptotic properties. The effectiveness of the proposed approach is demonstrated through simulations and a real world example.\end{abstract}

\section{Introduction}
Bayesian uncertainty can be understood as arising through  missing data, where given an i.i.d. sample $y_{1:n}$ from an infinite population, the uncertainty in parameters of interest comes from the unobserved data $y_{n+1:\infty}$. To deal with such uncertainty, a natural way is to impute the missing data \citep{fortini2020quasi,fong2023martingale}, which is also known as \textit{predictive resampling}. Under such a view, the  Bayesian model can be viewed as an approach for predictive resampling through the chain rule,
$$p(y_{n+1:\infty}\mid y_{1:n})=\prod_{i=n+1}^\infty p(y_i\mid y_{1:i-1}),$$
where $p(y_i\mid y_{1:i-1})$ is the posterior predictive constructed from the prior and likelihood. Posterior uncertainty on parameters of interest is then obtained by computing functionals of the imputed population, i.e. through $\theta(y_{n+1:\infty})$. 

This predictive view of Bayesian inference has a rich history dating back to de Finetti \citep{de1937prevision}, but was recently invigorated by recent works that leveraged the predictive view to develop new methodology, e.g. \cite{fortini2020quasi, berti2023bayesian,fong2023martingale}. An excellent review of the foundations of predictive Bayesian inference is given by \cite{fortini2025exchangeability}. A key advantage of the predictive approach over the traditional Bayesian approach
is the ability to avoid specification of the likelihood and prior, and predictive resampling can often be much more expedient than Markov Chain Monte Carlo (MCMC) and quantify uncertainty more accurately than variational Bayes.

A canonical example of the predictive view is the Bayesian bootstrap (BB) \citep{rubin1981bayesian}, which can be interpreted as using the empirical distribution for predictive resampling in connection to P\'{o}lya urns \citep{blackwell1973ferguson}. The BB posterior assigns Dirichlet-distributed weights $w \sim \text{Dirichlet}(1,\ldots,1)$ to the observations, and  
this time-honored method is recognized for its speed and robustness to model misspecification \citep{Lyddon2018, Fong2019}. However, it is also known that it performs poorly in small sample scenarios and it is unable to generate samples beyond the convex hull of the observed data. Furthermore, the Dirichlet weights are highly variable when the sample size is small, leading to erratic posterior distributions which fluctuates significantly across different bootstrap replicates.  These issues are for example problematic in the estimation of quantiles, as the BB posterior of a quantile will place point masses on a subset of the data points, which we demonstrate in the following example.
\begin{example}
   Suppose we have observed $n = 7$ samples from a skew normal distribution with true 95th percentile and median equal to 10.8 and 3.8 respectively. We conduct the BB to generate 5000 posterior samples. The  posterior of the 95th percentile is concentrated at a single datum near 10, with minimal mass distributed among other values (Fig. \ref{fig:i} (left)). Similarly, the posterior of the median has most of its mass on 3 data points  (Fig. \ref{fig:i} (right)).\vspace{-5mm}
   \begin{figure}[ht]
    \centering
    \begin{subfigure}[t]{0.4\textwidth}
        \centering
\includegraphics[width=\linewidth]{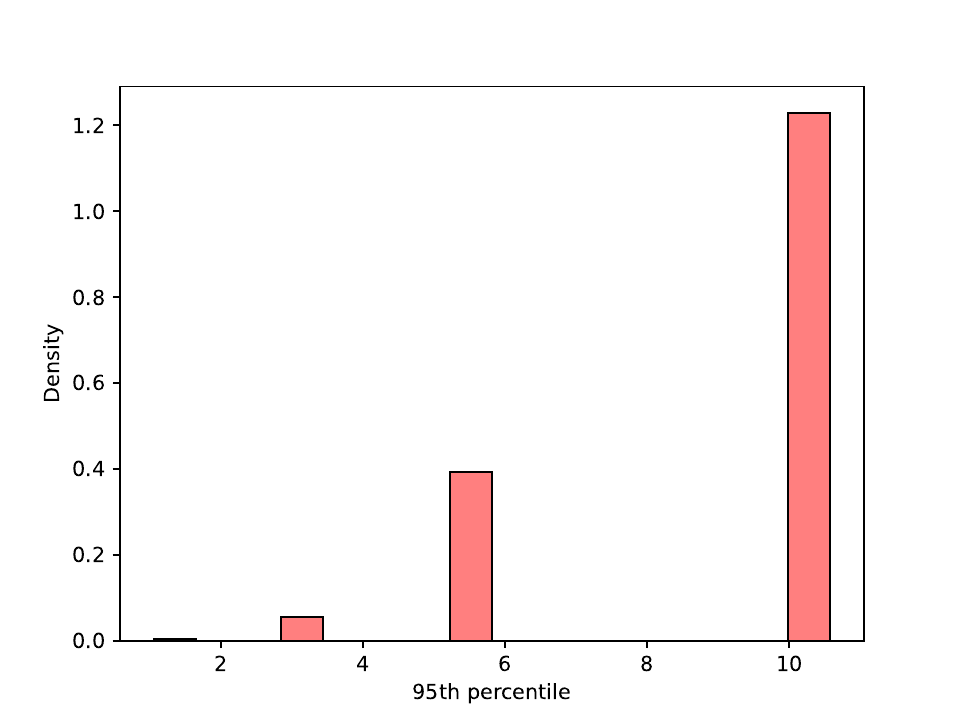}
    \end{subfigure}
    \begin{subfigure}[t]{0.4\textwidth}
        \centering
\includegraphics[width=\linewidth]{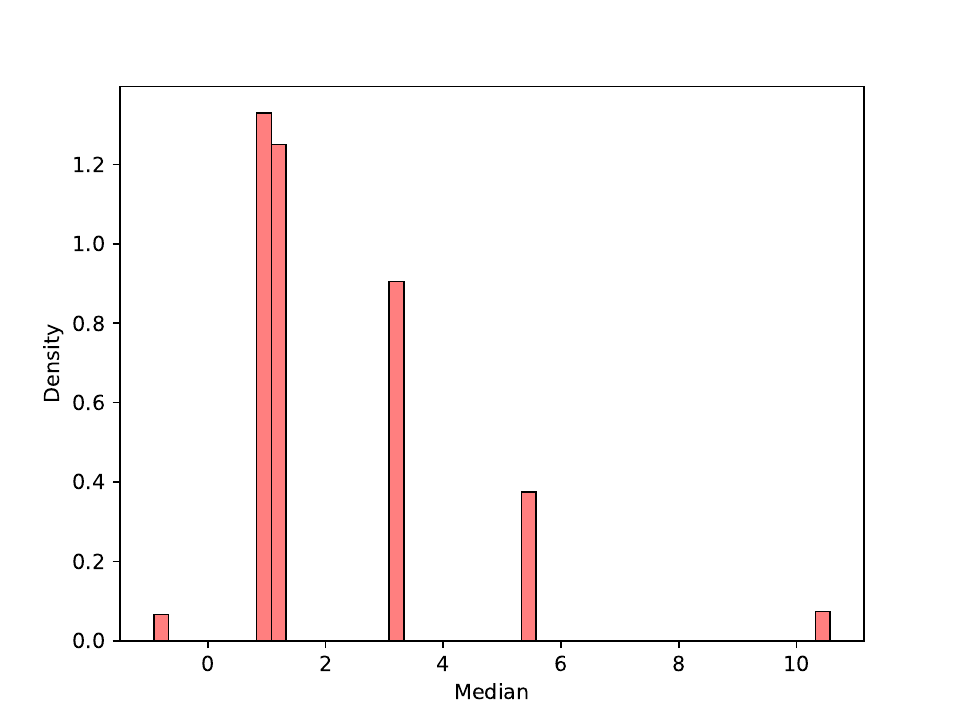}
    \end{subfigure}
    \vspace{-2mm}
    \caption{Posterior distribution obtained by the BB for the (Left) 95th percentile and (Right) median.} 
        \label{fig:i}
\end{figure}
\label{exp1}\vspace{-2mm}
\end{example} 
 A natural extension to the BB is to utilize a recursive estimate of a nonparametric predictive distribution, e.g. using copulas \citep{hahn2018recursive, fong2023martingale} or predictive recursion \citep{fortini2020quasi}. We refer to this approach as the nonparametric martingale posterior (MP) as coined in \cite{fong2023martingale}. The sequence of predictive distributions is required to satisfy a martingale condition $E[p_i(y)\mid y_{1:i-1}] = p_{i-1}(y)$, which induces a conditionally identically distributed (c.i.d.) sequence \citep{berti2004limit}. A key advantage is that predictive resampling is computationally efficient as it avoids MCMC. The BB is then a special case of the nonparametric MP using the empirical distribution as the predictive. 
 The MP framework has also been extended to survival analysis \citep{walker2024martingale}, quantile regression \citep{fong2024Bayesian}, time series \citep{Moya2024}, and machine learning \citep{lee2023martingale,falck2024large}.
 A parametric variant of the MP has also been investigated by \citep{holmes2023statistical,garelli2024asymptotics, fortini2025exchangeability}, where the sequence of predictives is a plug-in parametric model updated with stochastic gradient descent. In this case, the model parameter is constructed to be the martingale. 

In this paper, we consider the construction of a semiparametric MP with the goal of alleviating the pathologies of the BB in finite samples, whilst preserving its robustness to model misspecification and computational benefits. As parametric and nonparametric models have their strengths and drawbacks, a semiparametric approach can act as a potentially useful compromise \citep{powell1994estimation}. Drawing inspiration from the mixture of Dirichlet processes (MDP) \citep{Antoniak1974,Lyddon2018}, the main idea is to utilize a \textit{mixture predictive distribution}, which contains a parametric and nonparametric component.  Moment estimation will play a crucial role in tying these two components together, and we henceforth refer to our method as the \textit{moment martingale posterior}. 
 While we focus on using the BB as the nonparametric component in this paper, our methodology can be easily extended to use more general nonparametric components such as those in \cite{fortini2020quasi} and \cite{fong2023martingale}. 
Interestingly, although the parametric and nonparametric MP have been investigated, the semiparametric MP will require a novel approach.

 In Section 2, we review the nonparametric and parametric MP framework, and discuss connections to the MDP. Then, in Section 3, we introduce the moment MP and establish its asymptotic properties. In contrast to the parametric and nonparametric MP, our method depends crucially on a hyperparameter which balances the nonparametric and parametric components of the mixture predictive.  We consider the selection of this hyperparameter using a specific strictly proper scoring rule in Section 4, and show its asymptotic properties. Section 5 then extends the method to regression setting. Section 6 includes the experimental studies and Section 7 concludes the paper with future directions.

\section{Martingale posteriors and related work}

\subsection{Martingale posteriors}
In this section, we review the MP framework. The main idea of the MP is to specify a sequence of predictive distributions $\left(p_i(y)\right)_{i \geq n}$, which can then be utilized to impute the population $y_{n+1:\infty}$ with predictive resampling. When eliciting the sequence of predictives, two criteria are introduced in \cite{fong2023martingale} for \textit{coherency} of the MP: the first is the existence of the limit of predictive resampling, while the second is a guarantee that no bias is induced through predictive resampling. Both properties can be satisfied by requiring a martingale condition and  relying on  Doob's celebrated convergence theorem. 
A general criterion is thus to specify the sequence of predictives so that the object of interest is a martingale. This angle encompasses both the nonparametric and parametric MP, where the predictive distribution and parameter of interest are enforced as martingales respectively.

We now state two key results that are leveraged in most MP constructions. For convenience, we omit technical details regarding the filtration and measurability.
\begin{lemma}[Convergence]\label{lem:1}
Let $(X_i)_{i \geq n}$ be a martingale bounded in $L_1$, that is 
$\sup_{i \geq n} E\left[|X_i| \mid X_{1:n}\right] < \infty$.
Then we have that the limit
$X_\infty:=\lim_i X_i$
exists a.s.
\end{lemma}
\begin{lemma}[Unbiasedness]\label{lem:2}
If the martingale $(X_i)_{i\geq n}$ is uniformly integrable (where being bounded in $L_2$ is sufficient), the limiting random variable $X_\infty$ satisfies
    $E[X_\infty \mid X_{1:n}]=X_n$.
\end{lemma}

The first result ensures the existence of the MP, which is precisely the distribution of $X_\infty$. Under uniform integrability, the second result ensures  the MP distribution of $X_\infty$ is unbiased. We now outline the two canonical choices of the martingale for the nonparametric and parametric MPs, before generalizing to the semiparametric case in Section 3.

\subsubsection{Nonparametric MP}
For the nonparametric MP, we require the predictive cumulative distribution function (CDF) $P_i(y)$ to be a martingale under predictive resampling \cite{fortini2020quasi,fong2023martingale}. For example, in the BB, the following update rule satisfies this:
\begin{align*}
    P_{i+1}(y) = (1-\alpha_{i+1})P_i(y) + \alpha_{i+1} \mathbbm{1}\left(y_{i+1} \leq y\right)
\end{align*}
where $\mathbbm{1}$ is the indicator function. Other update functions for smooth CDFs also exist, e.g. using a bivariate copula. Given the observed data $y_{1:n}$, we estimate $P_n(y)$ and sample $y_{n+1} \sim P_n(y)$, then update to $P_{n+1}$ using the above rule. Iterating the above gives a random sample $Y_{n+1:\infty}$, from which we can compute a sample from the MP through $\theta(Y_{n+1:\infty})$. 
The existence of the MP relies on the existence of the limiting predictive CDF $P_\infty$. As $P_i$ is a CDF bounded by 1, both Lemmas \ref{lem:1} and \ref{lem:2} apply, and the sequence $\{P_{n+1}(y), P_{n+2}(y),...\}$ converges to a random $P_\infty(y)$ a.s. for each $y \in \mathbbm{R}$. This is sufficient for $P_i$ to converge to $P_\infty$ weakly a.s. \citep{berti2004limit}. For the nonparametric MP, letting the martingale be the predictive CDF $P_i(y)$ is thus sufficient.

\subsubsection{Parametric MP}

For the parametric MP \citep{holmes2023statistical, fong2024asymptotics, fortini2025exchangeability}, 
the predictive is a parametric density $p_{\theta}$. The parameters $\theta$ of the parametric model are then updated given an imputed sample $y_i\sim f_{\theta_{i-1}}$ as follows: 
\begin{align}\label{eq:param_MP}
    \theta_{i} =\theta_{i-1} + i^{-1}\mathcal{I}(\theta_{i-1})^{-1} s(y_{i}, \theta_{i-1})
\end{align}
in which $s$ is the score function, $\mathcal{I}$ is the Fisher information matrix, and the $i^{-1}$ is the learning rate that satisfies $\sum_{i=n+1}^\infty i^{-1}=\infty$ and $\sum_{i=n+1}^\infty i^{-2}<\infty$. Repeating the aforementioned predictive resampling process yields the posterior distribution of the parameter of interest. 

 Under predictive resampling, the score function has mean 0, so $\theta_i$ is a martingale, and under regularity conditions, the martingale is bounded in $L^2$ \citep{fong2024asymptotics}. Again both Lemmas \ref{lem:1} and \ref{lem:2} apply, where the martingale is $\theta_i$ for the parametric MP.
 
\subsection{Mixture of Dirichlet processes}
A key source of inspiration of our method is the MDP of \cite{Antoniak1974}. For a parametric model $f_\theta$, the MDP prior can be written as
$F|\theta\sim \text{DP}(c,f_\theta)$ where $\theta\sim\pi(\theta).$
This is a mixture of standard Dirichlet processes (DP) with hyper-prior $\pi(\theta)$ and concentration parameter $c > 0$. The MDP can be concisely written as
$F|\theta\sim \text{MDP}(\pi(\theta),c,f_\theta).$
By the conjugacy of the DP, the posterior distribution  $\pi(F\mid y_{1:n})$ is also an MDP,
\begin{flalign*}
\text{MDP}\Big(\pi(\theta \mid y_{1:n}),c + n,\frac{c}{c+n}f_\theta+\frac{n}{c+n}\mathbb{P}_n\Big),
\end{flalign*}
where $\pi(\theta \mid y_{1:n})$ is the standard parametric Bayesian posterior and $\mathbb{P}_n$ is the empirical measure of the sample $y_{1:n}$.

The key advantage of the MDP is the robustness to model misspecification, whilst still using $f_\theta$ to regularize the nonparametric component \citep{Lyddon2018}.
The magnitude scalar $c$ can be understood as the confidence level on the baseline model $f_\theta$ being accurate, whilst asymptotically the nonparametric component dominates. Using the MDP, \cite{Lyddon2018} suggested a BB-like method to construct a robust posterior distribution on $\theta$. However, the main downside is the requirement to still draw $\theta^{(i)}\sim \pi(\theta\mid y_{1:n})$, which would require MCMC. Furthermore, while $c$ is crucial for the performance of the MDP, \cite{Lyddon2018} did not select the $c$ in a systematic fashion. 

\section{Moment martingale posteriors}

 Inspired by the centering measure of the posterior MDP, we elicit our predictive distribution for the semiparametric case as
\begin{flalign}
p_{i}(y)=\frac{c}{c+i}f_{\theta_i}(y) + \frac{i}{c+i}\mathbbm{P}_i\label{eq:2}
\end{flalign}
where $\theta_i$ is a plug-in estimate to be specified. The above choice is highly intuitive - our predictive distribution is a weighted sum of a regularizing parametric component $f_{\theta}$ with the empirical distribution. As the number of samples $y_{1:i}$ increases, the empirical component dominates so we recover the BB asymptotically. We thus have a predictive that performs well for small $n$ but maintain the robustness to model misspecification. 

Given the above choice of predictive, we are immediately faced with the following question: which quantity should the martingale be, and how do we estimate $\theta_i$ to satisfy this condition? To see that this is not immediately obvious, consider estimating $\theta_{i+1}$ with the update rule (\ref{eq:param_MP}) after drawing $y_{i+1} \sim p_{i}$ given by (\ref{eq:2}). Then we have
\begin{align*}
    E\left[s(y_{i+1}, \theta_{i})\mid y_{1:i}\right] = \frac{i}{c+i}\frac{1}{i} \sum_{j = 1}^i s(y_j, \theta_i) \neq 0
\end{align*}
As a result, due to the nonparametric component, we may have $E[\theta_{i+1} \mid y_{1:i}] \neq \theta_i$, so the parameter is not necessarily a martingale. Furthermore, a similar argument shows that the predictive CDF $P_i(y)$ will not be a martingale either due to the parametric component.

\subsection{Method of moments}
As alluded to in the earlier sections, it turns out that enforcing the predictive \textit{moments} to be a martingale is crucial in our semiparametric setting. This is a surprisingly satisfying answer, as the moments occupy a semiparametric middle ground in between the nonparametric CDF and the parametric $\theta$.
Specifically then, if we allow $\theta_i$ to be the \textit{method of moments} estimator, then a finite subset of the predictive moments will form a martingale.

To solidify the idea, let us assume that $\theta \in \mathbb{R}^p$, and 
define the $k$-th parametric and empirical moments respectively as
$$\mu_\theta^{(k)}=\int y^k \,f_\theta(y)\, dy, \quad \mu_n^{(k)}=\frac{1}{n}\sum_{i=1}^ny_i^k.$$ 
For the method of moments, we match the first $p$ moments, that is we find $\theta_n$ such that for all $k\in\{1,...,p\}$, the following is satisfied:
\begin{align}\mu^{(k)}_{{\theta}_n}=\mu^{(k)}_n.\label{eq:mom_cond}
\end{align}
As the predictive is a mixture of the empirical distribution and the parametric model, the predictive moments will match both the nonparametric and parametric components.  

It is standard to assume that there exists a function $h$ such that $\theta_n = h(\mu_n^{(1)},\ldots, \mu_n^{(p)})$ satisfies the above. Our first result then follows immediately.
\begin{prop}\label{prop:mom_mart} Let our predictive distribution be  (\ref{eq:2}). If $\theta_i$ is estimated by the method of moments, then the first $p$ moments are martingales under predictive resampling, that is $$E\left[\mu^{(k)}_{i+1}\mid y_{1:i}\right]=\mu^{(k)}_i, \quad \forall k\in\{1,...,p\},\quad i\geq n$$
\end{prop}
 We hence term the resulting MP given by our choice of the mixture predictive (\ref{eq:2}) as the \textit{moment martingale posterior}. Moment predictive resampling is then illustrated in Algorithm \ref{alg:ps}, where we recursively update the first $p$ empirical moments and update $\theta$ afterwards.

\begin{algorithm}[!h]
\caption{Moment Predictive Resampling}\label{alg:ps}
\begin{algorithmic}[1]
\State \textbf{Input:} Observed data $y_{1:n}$, parametric model $f_\theta$, and number of posterior samples $B$
\State Calculate $\{\mu^{(k)}_{n}\}_{k = 1,\ldots,p}$ and initial estimate ${\theta}_{n}$ 
\For{$i \gets n$ \textbf{to} $N$}
    \State Sample $y_{i+1}$ by Equation (\ref{eq:2})
    \State Update moments:
    $\mu^{(k)}_{i+1} \gets \frac{i}{i+1}\mu^{(k)}_i + \frac{1}{i+1}y_{i+1}^k, \quad \forall k\in\{1,...,p\}$
    \State Update $\theta_{i+1}$ with $\{\mu^{(k)}_{i+1}\}_{k = 1,\ldots,p}$
\EndFor
\State Repeat Lines 3-7 $B$ times and \textbf{return} $\{\theta^{(1)}_N,...,\theta^{(B)}_N\}$
\end{algorithmic}
\end{algorithm}

 As the moments $\mu_i^{(k)}$ are martingales, we can apply Lemma \ref{lem:1} to show a.s. convergence of the moments. We will however require additional assumptions on $f_\theta$ to ensure the martingale sequence is in $L_1$. To that end, let us define the $k$-th parametric \textit{absolute moment} as 
 \begin{align}
     |\mu_\theta|^{(k)} := \int |y|^k f_{\theta}(y)\, dy.
 \end{align}
 Note that the above is distinct to the absolute value of the moment $|\mu_\theta^{(k)}|$, which is upper bounded by the absolute moment $|\mu_{\theta}|^{(k)}$.
 Our result then follows.
\begin{theorem}\label{thm:1} Let $f_\theta$ be the parametric component of our mixture predictive in (\ref{eq:2}), and assume that the method of moments estimator exists for all values of $\mu^{(k)}$. Assume additionally that at least one of the following conditions is satisfied:
\begin{enumerate}
    \item  $p$ is an even integer;
    \item  $p$ is an odd integer, 
    and there exist some non-negative constants $a$ and $b$ such that 
    $$|\mu_\theta|^{(p)} \leq a|\mu^{(p)}_\theta |+ b.$$
\end{enumerate}\vspace{-2mm}
Then under predictive resampling, the limiting moments exist, that is
\begin{align*}
    \mu_i^{(k)} \to \mu_\infty^{(k)} \quad \forall k \in \{1,\ldots,p\}\quad \textnormal{a.s.}
\end{align*} 
\end{theorem}

We now outline the intuition for the required assumptions above. When $p$ is even, we can rely on the non-negative martingale to ensure that $\mu_n^{(p)}$ is bounded in $L_1$. When $p$ is odd however, the absolute value of the moment $|\mu_n^{(p)}|$ is not a martingale. We thus require the absolute moment of the parametric model to be sufficiently regular in order to control the expectation of $|\mu_n^{(p)}|$. Convergence of lower moments $k < p$ then follows from standard applications of Jensen's and H\"{o}lder's inequalities. Similar assumptions are used in \cite{fong2024asymptotics} for the score function in the parametric MP setting. 
We highlight the key property that the required condition is only a function of the parametric component $f_{\theta}$, so it is easy to verify in practice. We demonstrate verification of the above assumption in Appendix \ref{sec:app_proofth1} for the case where $p$ is odd, and note that it is usually simpler to assume the existence of the $(p+1)$-th moment.

 Although the parameter and predictive CDF are no longer martingales, we have shown that the MP over moments exists. In the subsequent sections, we will additionally show that having the moment martingale and appropriate regularity conditions is sufficient for showing the existence of the posterior distribution on the parameter, CDF and higher-order moments, thereby justifying the semiparametric nature of the moment MP.
 
 For the remainder of the paper,  the term ``parametric model'' will refer to $f_\theta$, and the term ``mixture predictive'' will refer to the $p_i(y)$ as given in (\ref{eq:2}).  Section 3.2 will focus on the convergence of the parametric component, and Section 3.3 will study the convergence of the nonparametric component and higher order moments (i.e. with order higher than $p$).

\subsection{Convergence of the parametric component}

Given a.s. convergence of the first $p$ moments, showing the convergence of $\theta$ is a simple application of the continuous mapping theorem. This is stated formally below.
\begin{corollary} \label{cor:paraconverge}
 Assume the conditions of Theorem \ref{thm:1} are met, and additionally assume that $\theta$ can be expressed as a continuous function of the $p$ moments $\theta= h(\mu^{(1)},...,\mu^{(p)})$ for all values of $\mu^{(k)}$. Then we have $\theta_i \to \theta_\infty$ a.s.
\end{corollary}
The continuity condition in Corollary \ref{cor:paraconverge} is well studied in the literature. The detailed condition for the existence and continuity of the function $h$ is for example stated in \citet[Lemma 4.3]{van2000}. In short, suppose $g$ is a vector valued function that maps parameters to the theoretical moments, that is $g(\theta)=(\mu^{(1)},...,\mu^{(p)}\,)$. If $g$ is continuously differentiable with non-singular derivative, then the $g^{-1}:=h$ is well-defined and continuous. In general, most distributions satisfy this property \citep{van2000}. In conclusion, we have shown that since all $p$ moments converge, the parameters of the parametric component also converge under predictive resampling.

\subsection{Convergence of the nonparametric component and higher order moments}

We now take the natural next step, and show that nonparametric components
of the moment MP also exist, specifically regarding the CDF and higher order moments. We first present the existence of the limiting random probability distribution $P_\infty$ in the following proposition: 
\begin{prop}\label{prop:exist} 
The sequence $P_{n+1}(\cdot),P_{n+2}(\cdot),...$ converges weakly to a random $P_{\infty}(\cdot)$ a.s., where $P_{\infty}(\cdot)$  is a random distribution.
\end{prop}

Proposition \ref{prop:exist} thus confirms that our moment martingale posterior satisfies the convergence condition from \cite{fong2023martingale}. Consequently, $P_\infty$ is defined via the sequence of predictives, allowing us to regard $P_\infty$ directly as the random distribution function without requiring a representation based on Bayes' rule. More importantly, we do not impose any assumptions in Proposition \ref{prop:exist}. Specifically, irrespective of the parametric family $f_\theta$ or the method used to estimate the parameter, our predictive mixture will converge to a random CDF.

For existence of the moment MP on higher order moments (larger than $p$), additional conditions are needed since weak convergence does not imply the convergence of moments (e.g. \cite{van2000}). However, we are able to show that higher order moments are also almost supermartingales under quite weak conditions, with convergence guaranteed by the result of \cite{robbins1971convergence}.
\begin{prop}\label{prop:higher} 
Assume the conditions of Theorem \ref{thm:1} are met. For $k > p$, suppose the $k$-th parametric absolute moment  satisfies
$|\mu_{\theta}|^{(k)} \leq g(\theta)$
where $g(\theta)$ is a continuous function and $g(\theta) < \infty$ for each $\theta$. Then the $k$-th empirical moment satisfies $\mu^{(k)}_i\to \mu^{(k)}_\infty$ a.s.
\end{prop}
We provide an example of verifying the above assumption in Appendix \ref{sec:app_proofhigher}, and once again, if $k$ is odd, it may be simpler to verify the above assumption using the even moment $k+1$ if one wants to avoid absolute moments. Proposition \ref{prop:higher} thus allows us to assess the posterior uncertainty of any functions of higher order moments. For example, if the parametric model follows a Gaussian distribution (Appendix \ref{sec:app_proofhigher}), we can now study the posterior uncertainty on skewness and kurtosis.

\subsection{Comparison to the parametric and nonparametric MP}

Compared to the nonparametric MP and the BB, the moment MP allows the incorporation of the prior information on the distribution function $F$ through the parametric component $f_\theta$. For instance, if we choose a Gaussian parametric component, this provides strong prior information that the kurtosis is 3, so the posterior kurtosis obtained by the moment MP will be noticeably different to that of the BB, which we demonstrate later. Hence, $c$ not only quantifies the level of confidence for the parametric model, but can also be understood as the amount of prior information that we want to incorporate. We highlight however, that our method does not provide strong prior information on $\theta$, as we are still relying on empirical moments to estimate it.

Furthermore, the parametric component can offer regularization for the nonparametric component, which is beneficial when $n$ is small.
As a result, our mixture predictive is able to produce samples that extend beyond the convex hull of the observed data, which produces reasonable posterior distribution even when $n$ is small. On the contrary, we showed in Example \ref{exp1} that the BB posterior is constrained to point masses focused on a small number of observations. The regularization offered by the parametric component greatly alleviates this pathological behavior, which we demonstrate in Section 6.

On the other hand, the nonparametric component of the moment MP also means that we do not rely on the assumption that $f_\theta$ is well-specified, as we recover the BB asymptotically. This robustness to model misspecification is not too surprising given the connections to the MDP \citep{Antoniak1974,Lyddon2018}. Finally, posterior computation is straightforward to compute, especially if an explicit function that maps the moments to $\theta$ is available. This makes posterior sampling potentially significantly faster for the moment MP than traditional Bayes or the MDP which requires MCMC.

Despite its advantages, this moment MP has some limitations. One issue that arises is when the explicit function between the parameters and moments is unavailable. In some cases, such as the discrete uniform distribution, although the explicit functions between the parameters (lower/upper bounds) and the moments exist, the resulting solutions are often not integers, and rounding complicates the proof of martingale convergence. Another obvious problem is that if the moments are undefined (e.g. the Cauchy distribution), it cannot be used as the parametric component of our mixture predictive.

\section{Energy score for model selection}

A critical task is the selection of the hyperparameter $c$ for our mixture predictive (\ref{eq:2}). Setting
 $c =0$ means that we discard the parametric model completely, recovering the BB, whereas for $c =\infty$ we obtain the parametric martingale posterior. \cite{Lyddon2018} set $c$ by tuning the a priori variance of the population mean, but their approach relies on the prior distribution $\pi(\theta)$ and is thus not data driven. 
 Given that the moment MP is primarily a predictive framework, we instead opt to identify a value of $c$ that results in optimal predictions from the mixture predictive $P_n$.

In standard Bayesian practice, the model quality is typically quantified through the marginal likelihood,
$p_\mathcal{M}(y_{1:n})=\int p_n(y_{1:n}\mid \theta)\, \pi(\theta)\, d\theta.$
 Choosing $c$ to maximize the marginal likelihood would then correspond to empirical Bayes. Within the predictive setting in the absence of a prior, the marginal likelihood can be interpreted as a \textit{prequential score} \citep{dawid1984present} based on the factorization
$\log p_\mathcal{M}(y_{1:n}) = \sum_{i = 1}^n \log p(y_i \mid y_{1:i-1})$.
The scoring rule used for measuring the quality of the predictive is the log score \citep{gneiting2007strictly}, which has close connections to the Kullback-Leibler (KL) divergence. The marginal likelihood is also closely connected to cross-validation (CV) \citep{gneiting2007strictly, fong2020marginal}, and a popular alternative to the marginal likelihood is to utilize leave-one-out CV (LOOCV), where 
\begin{align}\label{eq:marginalLikelihoodLOOCV}
    S = \sum_{i=1}^n \log p_{n-1}(y_i \mid y_{-i}) 
\end{align}
 where $y_{-i}= \{y_1,...,y_n\}\setminus \{y_i\}$. LOOCV with the log score is known to alleviate some of the deficiencies of the marginal likelihood \citep{Vehtari2012}, and is thus an immediate contender for selecting $c$.

However, an immediate issue is that our predictive is a mixture of a discrete and continuous component, so such a log score is not well-defined on hold out data due to the discrete nonparametric component. 
This is unsurprising, as it is well known that (\ref{eq:marginalLikelihoodLOOCV}) approximates the negative of the KL divergence between $p_n(\cdot)$ and the unknown data generating distribution, and the KL divergence has strict requirements of support overlap. As a result, for the moment MP, identifying an alternative score is equivalent to identifying an alternative distance measure that allows for mixed predictive distributions. 

In Section 4.1, we will motivate the choice of the {energy distance} \citep{rizzo2016} as an alternative to the KL divergence, which results in the \textit{energy score}.  In Section 4.2, we will then discuss the asymptotic theory related to choosing $c$ using an estimate of the expected energy score. 

\subsection{The energy score}

An alternative popular choice of distances between probability measures is the class of integral pseudo-probability metrics \cite{muller1997integral}, which is defined as: 
$$\text{IPM}(\mathbb{Q},\mathbb{P})=\sup_{f\in\mathcal{F}}\Big|\int_\mathcal{X}f(x)\,d\mathbb{Q}(x)-\int_\mathcal{X}f(x)\,d\mathbb{P}(x)\Big|$$
IPMs can be thought of as comparing a family of generalized moments indexed by the class $\mathcal{F}$. The key is that $\mathbb{Q}$ and $\mathbb{P}$ need not have the same support as we are only comparing expectations, and the connection to moments is particularly elegant in our setting.

There are two commonly used IPMs, specifically the Wasserstein distance \citep{vallender1974calculation} and maximum mean discrepancy (MMD) \citep{smola2006maximum}. We propose the use of the energy distance \citep{rizzo2016}, which is a specific instance of the MMD employing $-|x-y|$ as the kernel. We opt to use this specific choice of kernel as it is tuning-free, and the energy distance is known to have robust properties \citep{gneiting2007strictly,huk2024quasi}. As the predictive score is the negative of the distance, we will use the energy score as the scoring rule for our LOOCV procedure. In the remainder of this section, we present the formal definition of the energy score and provide a rationale for its use. Subsequently, we describe the procedure for estimating the expected energy score using LOOCV.

The energy score between a model $P_n$ and a datum $y_i$ is defined as: 
$$s(y_i,P_n)=  -2E_{X\sim P_n}|y_i-X| + E_{X,X'\iid P_n}|X-X'|.$$
This score is easily interpretable: the first expectation $E_{X\sim P_n}|y_i-X|$ can be understood as the discrepancy between the distribution $P_n$ and the data $y_i$. The negative of this term thus quantifies the model's ability to predict $y_i$. The second term is to guarantee that it is a strictly proper scoring rule \citep{gneiting2007strictly}, which will be discussed shortly. Again, the connection of the energy score to (generalized) moments of $P_n$ is elegant within our moment MP framework.

Assuming $y_i \sim \mathbb{Q}$, the expected energy score between $\mathbb{Q}$ and $P_n$ is 
\begin{flalign}
S(\mathbb{Q},P_n) =E_{Y\sim\mathbb{Q}}[s(Y,P_n)]. \label{eq:theoEnergy}
\end{flalign}
There are several reasons for selecting the expected energy score as our evaluation metric. Firstly, it is based on the expected value of random variables coming from $\mathbb{Q}$ and $P_n$ so unlike the log score, it is perfectly applicable to our predictive which is a mixture of discrete and continuous components. The expected energy score is also a strictly proper scoring rule \citep{gneiting2007strictly}, which means that $$S(\mathbb{Q},\mathbb{Q}) \geq S(\mathbb{Q},P) \,\,\,\,\, \text{for any } P\neq \mathbb{Q},$$ with equality holding if and only if $P = \mathbb{Q}$. 

Computation is another important reason - the energy score can be computed in quadratic time in $n$ while the computation time of Wasserstein distance is cubic in $n$ \citep{Dellaporta2022}. Furthermore, as we will see shortly in Section \ref{sec:asymp}, the expected energy score is a strictly concave function of $\lambda={c}/{(c+n)}$, which guarantees that there is a unique value of $c$ that maximizes the score. In fact, we can write $\lambda$ as an {explicit} function of quantities that can be estimated from the observations, which is particularly helpful for both computation and theoretical study. See Proposition \ref{prop:concave} and Appendix \ref{sec:app_energy} for further details.

We now outline how to practically choosing $c$ given observations $\{y_1,\ldots,y_n\}$, as we do not have access to $\mathbb{Q}$ directly. For each $y_i\in\{y_1,...,y_n\}$, we define $y_{-i} = \{y_1,...,y_n\}\setminus \{y_i\}$.  Using  $y_{-i}$, we can derive the $p$ moments: $$\mu_{-i}^{(k)}=\frac{1}{n-1}\sum_{j\neq i} y^k_j \quad  \forall k\in\{1,...,p\}.$$
Subsequently, we can estimate the parameters $\theta_{-i}$ using the method of moments as before. The predictive model is then fitted as  
$$p_{-i}(y) =\frac{c}{c+n-1} f_{\theta_{-i}}(y) + \frac{n}{c+n-1}\mathbb{P}_{-i}.$$
where $\mathbb{P}_{-i}$ denotes the empirical distribution of $y_{-i}$.  To assess the predictive performance of the model for the datum $y_i$, we calculate 
$s(y_i, P_{-i}),$
where $P_{-i}$ is the CDF of $p_{-i}$. We then seek to find $c$ that maximizes
\begin{flalign}
 \widehat{S}(\mathbb{Q},P_n) =\frac{1}{n}\sum_{i=1}^n s(y_i, P_{-i}) \label{eq:loocvEnergy}.
\end{flalign}
The cumulative energy score $\sum_{i=1}^n s(y_i,P_{-i})$ in (\ref{eq:loocvEnergy}) can be interpreted in a similar fashion to  the marginal likelihood in (\ref{eq:marginalLikelihoodLOOCV}), where the log score is replaced with the energy score. Alternatively, we may also determine $c$ using a hold out estimate, prequential scoring or $K$-fold CV. We suggest to use LOOCV for small/ medium sample sizes, and a hold out estimate or $K$-fold CV for larger datasets due to their computational efficiency. We will shortly study the asymptotic properties of a hold out estimate of the optimal $c$. Finally, we highlight that $c$ is calculated only once prior to the initiation of predictive resampling and remains constant throughout the bootstrap process, thereby saving computational costs.

\subsection{Asymptotic theory} \label{sec:asymp}
 In practice, the presence of parametric model helps to regularize the BB, which will particularly helpful when sample size is small as we will see in Section 6. However, it would be desirable for our selection procedure for $c$ to be consistent. Specifically, we would like that our selection strategy yields $c=O(1)$ asymptotically under model misspecification, so the nonparametric part dominates. In this subsection, we study the asymptotics of our selection procedure for $c$ that maximizes an estimate of the expected energy score (\ref{eq:theoEnergy}).

 We begin with introducing some notation. Let {$\lambda={c}/{(c+n)}$}, and let $\mathbb{Q}$ be the true data generating mechanism with density $q(\cdot)$.  Suppose we observe $n$ i.i.d. samples $y_{1:n} \iid \mathbb{Q}$, and let $\theta^*$ denote the limit (in probability) of the method of moments estimator $\theta_n$, which exists under standard regularity conditions, e.g. the first $p$ moments of $\mathbb{Q}$ are finite and $h$ is continuous.
We define our parametric model as being misspecified if $f_{\theta^*}(\cdot)$ is not almost everywhere equal to $q(\cdot)$. Now define the optimal weighting as
\begin{align}\label{eq:opt_lam}
{\lambda_n = \argmax_{\lambda \in [0,1]} \, S\left(\mathbb{Q}, P_n\right),} 
\end{align}
where $P_n$ is our mixture predictive. In practice of course, we do not have access to $\mathbb{Q}$, so it is natural to utilize data splitting to estimate the above. Consider a partition of the $n$ observations into $n_t$ training and $n_v$ validation data points, where $n_t + n_v = n$. We denote these as $y_t$ and $y_v$ respectively where $y_{1:n} = \{y_t, y_v\}$, and {the mixture predictive $P_{n_t}$ is fitted to $y_t$ with weighting $\lambda = c/(c+n_t)$}. A natural estimate of (\ref{eq:opt_lam}) is then
\begin{align}\label{eq:est_lam}
{\widehat{\lambda}_n = \argmax_{\lambda \in [0,1]}\, S\left(\mathbb{P}_v, P_{n_t}\right),}
\end{align}
where $\mathbb{P}_v$ is the empirical distribution of $y_v$, and $P_{n_t}$ is the mixture predictive estimated from $y_t$.  Fortunately, both (\ref{eq:opt_lam}) and (\ref{eq:est_lam}) are straightforward to optimize as described below.

\begin{prop}\label{prop:concave}
    The energy scores $S\left(\mathbb{Q}, P_n\right)$ and $S\left(\mathbb{P}_v, P_{n_t}\right)$ are strictly concave in $\lambda$, and thus have unique maxima. Furthermore, both $\lambda_n$ and $\widehat{\lambda}_n$ can be expressed in closed form.
\end{prop}

In the interest of space, we provide the explicit forms of $\lambda_n$ and $\widehat{\lambda}_n$, as well as an extension of the above proposition to cross-validation, in Appendix \ref{sec:app_energy}. We highlight that for finite samples, {the maximum of $S\left(\mathbb{P}_v, P_{n_t}\right)$ may lie outside the range of $[0,1]$, hence the need to clip $\widehat{\lambda}_n$ to $0$ or $1$, which corresponds to $c = 0$ or $c = \infty$ respectively.} 
 Finally, although $\widehat{\lambda}_n$ is estimated from $n_t$ data points, since it is a fraction, we can solve $\widehat{\lambda}_n = c/(c+n)$ to obtain the optimal  $c$ for the full dataset.

We are now ready to state our main consistency result regarding the estimate $\widehat{\lambda}_n$.

\begin{theorem} \label{theo:optimal} 
Assume that $\int y^2 \,q(y) \,dy<\infty,$ the parametric density $f_{\theta}(x)$ is a continuous function of $\theta$ for each $x$, and there exists a neighborhood $U$ of $\theta^*$ such that $\sup_{\theta \in U}\mu_\theta^{(2)} < \infty$. If $f_\theta$ is misspecified, then the weight estimate (\ref{eq:est_lam}) satisfies $\widehat{\lambda}_n\overset{p}\to 0$ as $n_t, n_v \to \infty$.  
\end{theorem}

The above result shows that our selection procedure is consistent: if $f_\theta$ is misspecified, our estimated $\widehat{\lambda}_n$ will converge to 0, leaving us to rely on the BB solely, which is known to have excellent asymptotic properties \citep{Lyddon2018}. Note that we assume $n_v \to \infty$, so our above result applies to the case where $n_v \sim kn$ for $k \in (0,1)$. As $\widehat{\lambda}_n$ has an explicit form, the above result can be easily extended to CV when the validation set grows with $n$. Although the above result does not directly apply to LOOCV, we also expect it to hold under regularity conditions \citep{arlot2010survey}. In practice, we observe this to be the case, and we leave a formal verification for future work. Furthermore, from the definition of (\ref{eq:opt_lam}), the  predictive mixture $P_n$ optimally weighted with $\lambda_n$ will always dominate both the parametric and nonparametric components, that is 
\begin{align}\label{eq:dom}
   S(\mathbb{Q},P_n)\geq \max\left\{S(\mathbb{Q},F_{\theta_n}),S(\mathbb{Q},\mathbb{P}_n)\right\}.
\end{align}
While (\ref{eq:dom}) only holds for the theoretical value $\lambda_n$ and not necessarily for the estimated $\widehat{\lambda}_n$, we show in Section 6 that this seems to also hold empirically for $\widehat{\lambda}_n$ as well.

When $f_\theta$ is well-specified, the asymptotic behavior of  $\widehat{\lambda}_n$ is less significant as both the parametric and nonparametric components are well-specified, so our method will resort to evaluating which component is a better predictive fit for the particular sample. Furthermore, the parameter computed as functional of $P_n$, i.e. $\theta(P_n)$, will always be the same value irrespective of the value of $\widehat{\lambda}_n$, due to moment matching. We suspect that asymptotically, $\widehat{\lambda}_n$ will converge to a random variable in this setting; a similar study of a bootstrap consisting of a parametric and nonparametric component has been conducted by \cite{lee1994optimal}  where  $\widehat{\lambda}_n\to Z/(Z+1)$ for $Z\sim \chi^2_1$ when the model is well-specified.  We leave a detailed investigation of this for future work.

\section{Regression}\label{sec:reg}
We now consider extending the mixture predictive to the regression setting. Let the observed data be $\{y_i, x_i\}_{i=1}^n$, where $x_i\in\mathbb{R}^p$ denotes the independent variables and $y_i \in \mathcal{Y}$ denotes the scalar dependent variable.  In this section, we consider generalized linear models (GLM) as the parametric component, as this class has particularly clear connections to the method of moments in some settings. We first present our suggested mixture predictive, before showing that this choice is a natural extension to the moment matching argument.

\subsection{Predictive resampling}
Let $f_{\beta}(y \mid x)$ denote the conditional density of the GLM with link function $g$. For $\mu  = g^{-1}(x^\top \beta)$, we have $E_{\beta}[Y\mid x]  = \mu$ and $\text{Var}_{\beta}\left[Y \mid x\right] = \phi \,V(\mu)$, where $V(\mu)$ is the variance function and $\phi$ is the dispersion parameter. We then define the mixture predictive as: 
\begin{align}\label{eq:linreg}
p_i(y_{i+1} \mid x_{i+1})&=\frac{c}{c+i}\, f_{\beta_i}(y_{i+1}|x_{i+1})
 + \frac{i}{c+i}\mathbb{P}_{i}(y_{i+1} \mid x_{i+1}),
\end{align}
where $\beta_i$ is estimated from $\{y_{j}, x_{j}\}_{j = 1}^i$, and the nonparametric component is
$$\mathbb{P}_i(y_{i+1} \mid x_{i+1}) = \frac{\sum_{j = 1}^i \mathbbm{1}\left(x_j = x_{i+1}\right) \delta_{y_j}(y_{i+1})}{\sum_{j = 1}^i \mathbbm{1}\left(x_j = x_{i+1}\right)}\cdot$$
We will discuss estimation of $\beta_i$ in the next subsection.
The nonparametric conditional predictive $\mathbb{P}_{i}(y_{i+1} \mid x_{i+1})$ can be interpreted as follows: among all $x_j \in \{x_1,...,x_i\}$ that share the same value as $x_{i+1}$, we choose $y_{i+1}$ uniformly from  the values of $y_j$ which correspond to $x_j$. Interestingly, unlike the BB, we will end up with multiple values of $y$ for each unique value $x$, since there is a possibility of drawing from the parametric predictive. We demonstrate this in Section \ref{sec:logreg}, and provide related theory in Appendix \ref{sec:app_condit}.

For predictive resampling of the covariates, we will draw the new covariates $x_{i+1}$ from the BB, which is consistent with \cite{holmes2023statistical}, and its convergence is well understood. 
It turns out that this choice is necessary for a martingale condition due to the nonparametric component $\mathbb{P}_i(y_{i+1} \mid x_{i+1})$, which we discuss shortly. The predictive resampling scheme is then summarized in Algorithm \ref{alg:linreg}.

\begin{algorithm}[!h]
\caption{Predictive resampling for GLM}\label{alg:linreg}
\begin{algorithmic}[1]
\State Given observed data $\{y_i, x_i\}_{i=1}^n$ and number of posterior samples $B$
\State Compute the moment matrix $ \mu_n^x ,\,\mu_n^y $ and the OLS $\beta_n$
  \For{$i \gets n$ to $N-1$}
        \State  Sample $x_{i+1}$ by the Bayesian bootstrap 
        \State  Sample $y_{i+1}$ by (\ref{eq:linreg}) 
        \State  Update cross-moments:
        $\mu^{yx}_{i+1} \gets \frac{i}{i+1}\mu^{yx}_{i} + \frac{1}{i+1}x_{i+1}y_{i+1} $
        \State Update $\beta_{i+1}$ with $\{\mu^{y x}_{i+1}, x_{1:i+1}\}$
  \EndFor
\State Repeat Lines 3-8 $B$ times and \textbf{return} $\{\beta^{(1)}_N,...,\beta^{(B)}_N\}$
\end{algorithmic}
\end{algorithm}

\subsection{Convergence}
An elegant connection that we leverage is that if $g$ is the canonical link function, choosing $\beta_i$ as the MLE is a natural extension of the moment matching argument, ensuring compatibility with the nonparametric component. To begin, consider the empirical cross-moment term
\begin{align*}
     \mu_i^{yx} := \frac{1}{i}\sum_{j = 1}^i y_j x_j.
\end{align*}
The key fact is that if $\beta_i$ is the MLE for a GLM with a canonical link, it will satisfy 
\begin{align}\label{eq:cross_moment}
    \frac{1}{i}\sum_{j = 1}^i g^{-1}\left(x_j^\top \beta_i\right) x_j = \mu_i^{yx},
\end{align}
that is $\beta_i$ matches the cross-moments of the parametric and nonparametric component. This cross-moment matching property is exactly what guarantees that $\mu_i^{yx}$ is a martingale, as the parametric component satisfies $E_{\beta_i}[Y_{i+1}X_{i+1} \mid \mathcal{F}_i] = i^{-1}\sum_{j = 1}^i g^{-1}\left(x_j^\top \beta_i\right) x_j$. We thus opt to use (\ref{eq:cross_moment}) to estimate $\beta$ in the parametric component of our mixture predictive, even if the link function $g$ is not  canonical. We then have the following convergence result.

\begin{theorem}\label{theo:reg}
Assume that the design matrix $\mathbf{X}_n = (x_1,\ldots,x_n)^\top$ is full rank and includes the intercept, the initial estimate $\beta_n$ exists, and $\beta_i$ satisfies (\ref{eq:cross_moment}) for $i \geq n$. Assume additionally that one of the following conditions is satisfied:
\begin{enumerate}
    \item The dependent variable $Y\geq 0$ is non-negative, and the link function $g$ is strictly monotone and continuously differentiable. 
    \item The dependent variable $Y\in \mathbb{R}$ is real, the link function $g$ is the identity function, and there exist some non-negative constants $a$ and $b$ such that
    \begin{align*}
   V(\mu) \leq a \mu^2 + b. 
\end{align*}
\end{enumerate}\vspace{-2mm}
Then under predictive resampling, the cross-moment $\{\mu_i^{yx}\}_{i \geq n}$ is a martingale bounded in $L_1$, and thus $\mu_i^{yx} \to \mu_\infty^{yx}$ and $\beta_i\to \beta_\infty$ a.s. 
\end{theorem}

The condition on $V(\mu)$ is quite weak, for example being satisfied by the normal or the Laplace distribution.
Finally, we show in Appendix \ref{sec:app_condit} that empirical conditional moments $E[Y \mid x]$ also converge a.s. when $Y \geq 0$, which will be demonstrated in Section \ref{sec:logreg}.

\subsection{Practical considerations}\label{sec:reg_prac}
We now discuss a few practical considerations for the regression setting.
While we consider the exact solution to (\ref{eq:cross_moment}) when we study the convergence theory in the previous subsection, it is much more expedient to use an online update to approximate $\beta_i$ when predictive resampling when it is not available in closed form. This variant of Algorithm \ref{alg:linreg} is outlined in Appendix \ref{sec:app_regprac}, where we instead approximate the exact $\beta_i$ computation in Line 7 with an online update based on predictive recursion as in (\ref{eq:param_MP}), in a similar fashion to the update considered in \cite{fong2024asymptotics}. We show in  Appendix \ref{sec:app_regprac} that in practice, this approximation gives very accurate estimates of the solution to (\ref{eq:cross_moment}) under predictive resampling for logistic regression, and the resulting posteriors are very similar.  Finally, the selection of the hyperparameter $c$ is again crucial, and the methodology outlined in Section 4.2 can be easily extended to regression contexts, which we discuss in Appendix \ref{sec:app_regprac}.

\section{Experimental studies}
In this section,  we present two experimental studies. The first is a simulation example, which we will divide into two cases: one involving a moderate amount of data and the other involving a very small dataset. In both cases, we will further consider two scenarios, where $f_{\theta}$ is well-specified or misspecified. The second study is a logistic regression task using real-world data, namely the German Credit dataset \citep{statlog_(german_credit_data)_144}. In both experiments, we will primarily compare the mixture MP to the BB and the parametric MP.  

In the simulation example, we predictively resample $Y_{n+1:N}$ with $N=n+1500$
where $B=5000$ posterior samples are generated. In the logistic regression example, we draw of $N = n+4000$ forward samples, and generate $B=5000$ posterior samples. Note that $N$ and $B$ here align with the notation in Algorithm \ref{alg:ps}. The difference in $N$ arises from the differing numbers of predictive resamples needed for the martingale to converge. All simulations were executed on an Intel i7-13620H CPU.\footnote{Code for the experiments can be found at \url{https://github.com/w41039/Mixture-MP}.}

\subsection{Simulation example}
In all examples discussed within this subsection, the parametric model $f_\theta$ is selected as the normal distribution, {where $\theta$ consists of the mean and variance}. In the well-specified case, data are generated from a normal distribution $N(1, 25)$ for both data sizes. In the misspecified case, data are generated from a skew normal distribution characterized by parameters $\xi = 1$, $\omega = 5$, and $\alpha = 1$ for the small $n$ scenario, and $\alpha = 2$ for the moderate $n$ scenario. Consequently, in the second misspecified case, we expect that the prior hyperparameter $c$ will take on small value.

\subsubsection{Moderate sample size}

In both the well-specified and misspecified cases, the sample size is $n = 100$. We first examine the LOOCV estimate of the expected energy score of our mixture predictive. In the well-specified case (Fig. \ref{fig:2} (left)), the value of $c$ that maximizes the LOOCV energy score in this sample is $c = \infty$ (as we observe {$\widehat{\lambda}_n = 1$}), which suggests that one could drop the BB component in (\ref{eq:2}) and use the parametric part solely. In the misspecified setting (Fig. \ref{fig:2} (right)), the optimal value is $c=144$. This aligns with our expectation that the $c$ should be smaller (relative to $n$) when $f_\theta$ is misspecified.

Table \ref{table1} displays the optimal values of $c$ across the two scenarios, evaluated over 200 repeated samples of the datasets of size $n = 100$. In the well-specified case, 127 instances yielded an optimal $c$ of infinity ({$\widehat{\lambda}_n = 1$}). Conversely, only 53 instances demonstrate an optimal $c$ of infinity in the misspecified case. This outcome indicates that our approach can  differentiate the well-specified and misspecified cases, and is in line with the result of Theorem \ref{theo:optimal}.

\begin{figure}[!ht]
    \centering
    \begin{subfigure}[t]{0.46\textwidth}
        \centering
\includegraphics[width=\linewidth]{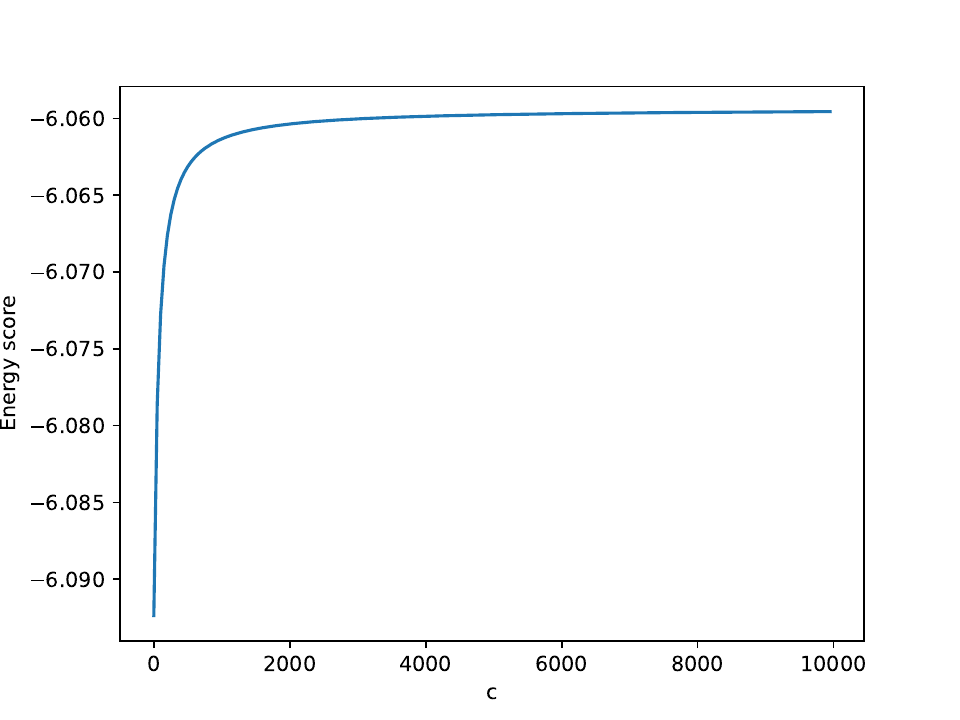}
    \end{subfigure}
    \begin{subfigure}[t]{0.46\textwidth}
        \centering
\includegraphics[width=\linewidth]{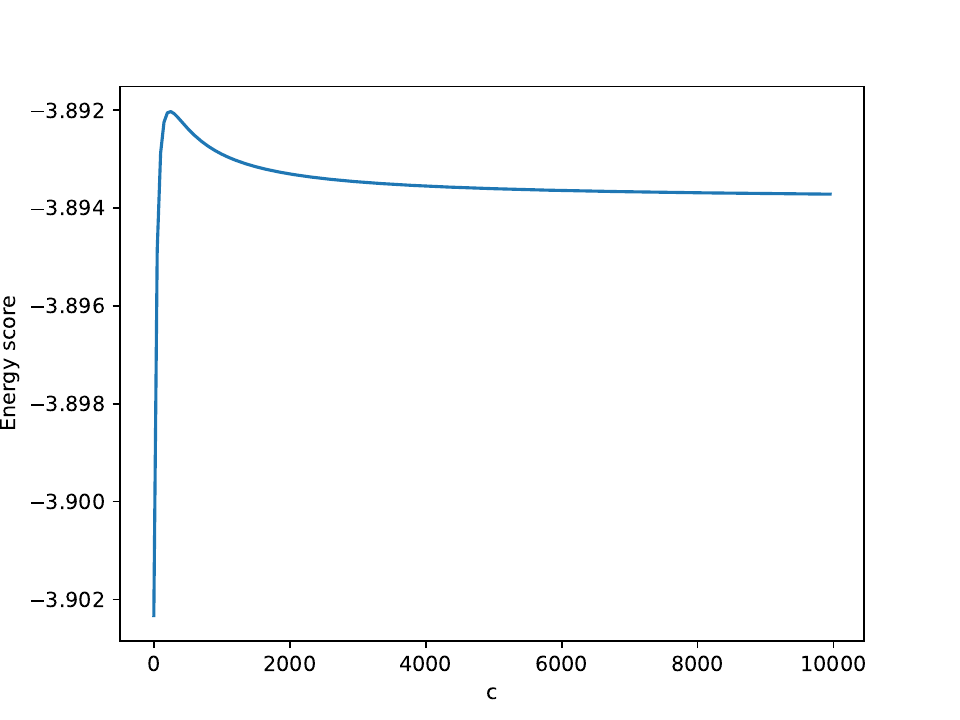}
    \end{subfigure}
    \vspace{-2mm}
    \caption{Plots of $c$ against LOOCV energy score for the (Left) well-specified and (Right) misspecified scenario. } 
    \label{fig:2}
\end{figure}

\begin{table}[!ht]
    \centering
    \scriptsize
    \begin{tabular}{c c c | c c c}
         \textbf{Moderate $n$} & \textbf{Well-specified} & \textbf{Misspecified} & \textbf{Small $n$} & \textbf{Well-specified} & \textbf{Misspecified}  \\ \hline
        \textbf{\% of $c=\infty$} & 63.5\% & 26.5\% &\textbf{\% of $c=\infty$} & 58\% & 61\% \\
    \end{tabular}
    \caption{Percentage of optimal $c$ being infinity under the well-specified scenario ($\mathbb{Q}$ is normal) and misspecified scenario ($\mathbb{Q}$ is skew normal) for (Left) moderate $n$ and (Right) small $n$. 
    }
    \label{table1}
\end{table}

Fig. \ref{fig:3} (left) and (right) 
illustrate the posterior distributions of the mean for the three different approaches. The posteriors are very similar and are centered around the correct value, irrespective of whether the model is well-specified or misspecified. This is because $\theta$ only relies on the empirical moment, which is the same in all three methods. This illustrates the lack of prior information on $\theta$ induced by our semiparametric mixture.

In contrast, Fig. \ref{fig:4} (left)
illustrates the posterior distribution of the skewness in the misspecified case, where the three approaches differ significantly. The posterior distribution derived solely from parametric MP is a point mass at 0, whereas the posterior obtained through the BB is centered around the true skewness of 0.45.  For the mixture MP, the posterior distribution of the skewness is a smooth density centered at around 0.2, indicating regularization induced by the parametric component $f_{\theta}$. This demonstrates the ability of the mixture MP to incorporate strong prior information into the BBB for higher-order moments. In this case however, the BB still performs the best due to $n$ being moderate, but we will see shortly that the regularization can help greatly when $n$ is small. Finally, we highlight that the mixture MP and BB will agree as $n \to \infty$ and the parametric component becomes negligible.

For illustrative purposes, we perform predictive resampling  with $c = 1e10$ in the well-specified case, and examine the convergence of moments through the sample paths (Fig. \ref{fig:4} (right)), where we see that the sample paths $\mu^{(3)}_{n+i}$ converge to a random $\mu_\infty^{(3)}$ as $i$ increases. 

\begin{figure}[!ht]
\centering
    \begin{subfigure}[t]{0.46\textwidth}
        \includegraphics[width=\linewidth]{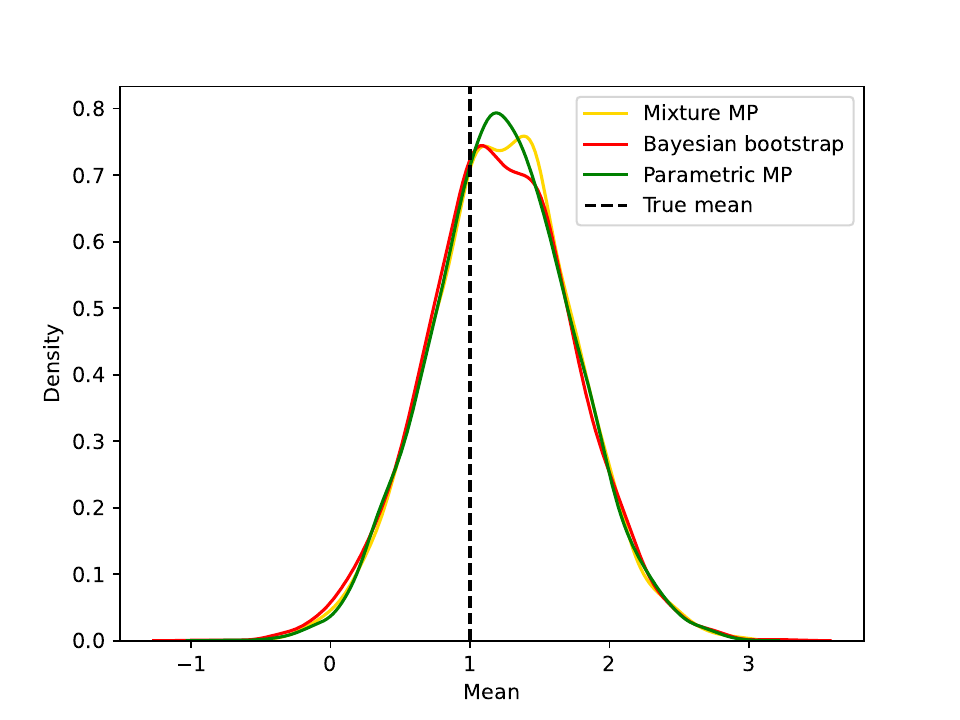}
    \end{subfigure}
    \begin{subfigure}[t]{0.46\textwidth}
\includegraphics[width=\linewidth]{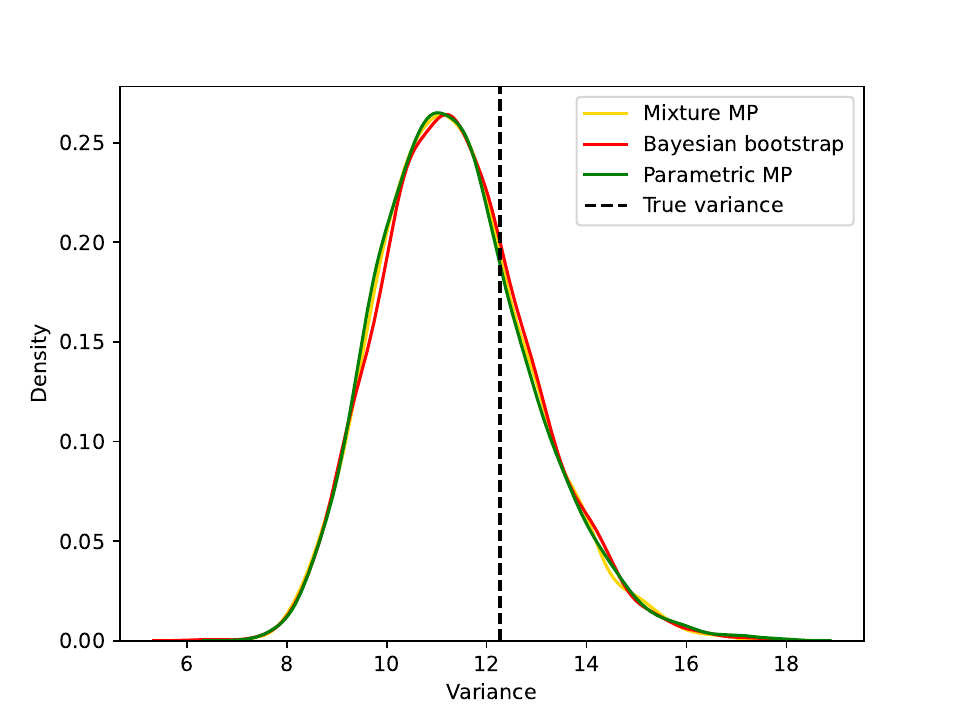}
    \end{subfigure}
    \vspace{-2mm}
    \caption{Posterior distributions of (Left) mean in the well-specified scenario and (Right) variance in the misspecified scenario.}
    \label{fig:3}
\end{figure}
\begin{figure}[!ht]
\centering
    \begin{subfigure}[t]{0.46\textwidth}
        \includegraphics[width=\linewidth]{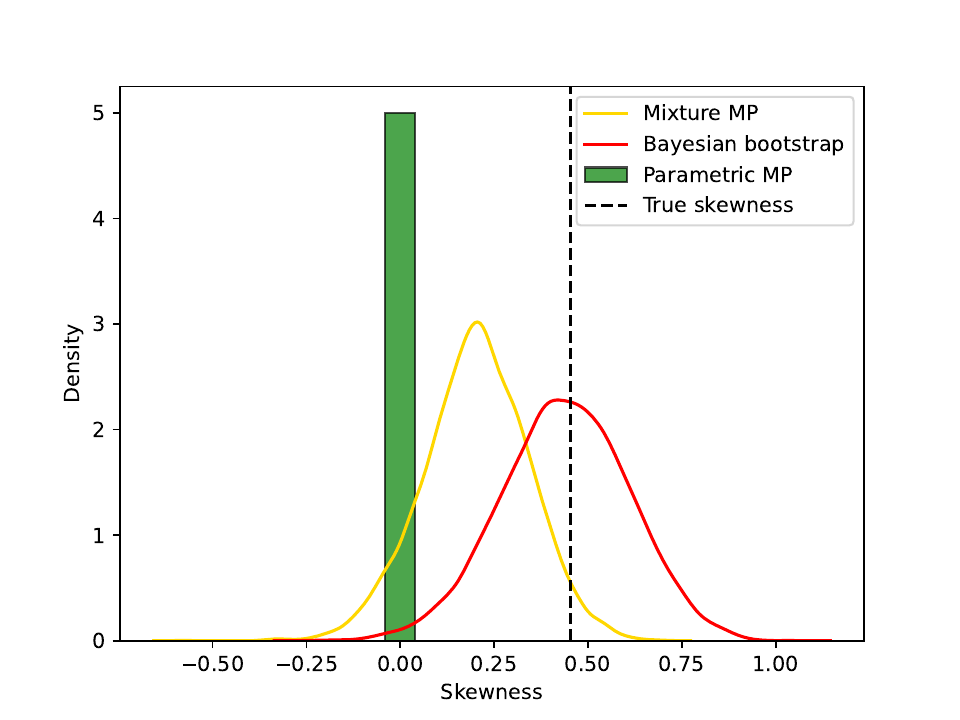}
    \end{subfigure}
    \begin{subfigure}[t]
    {0.46\textwidth}
\includegraphics[width=\linewidth]{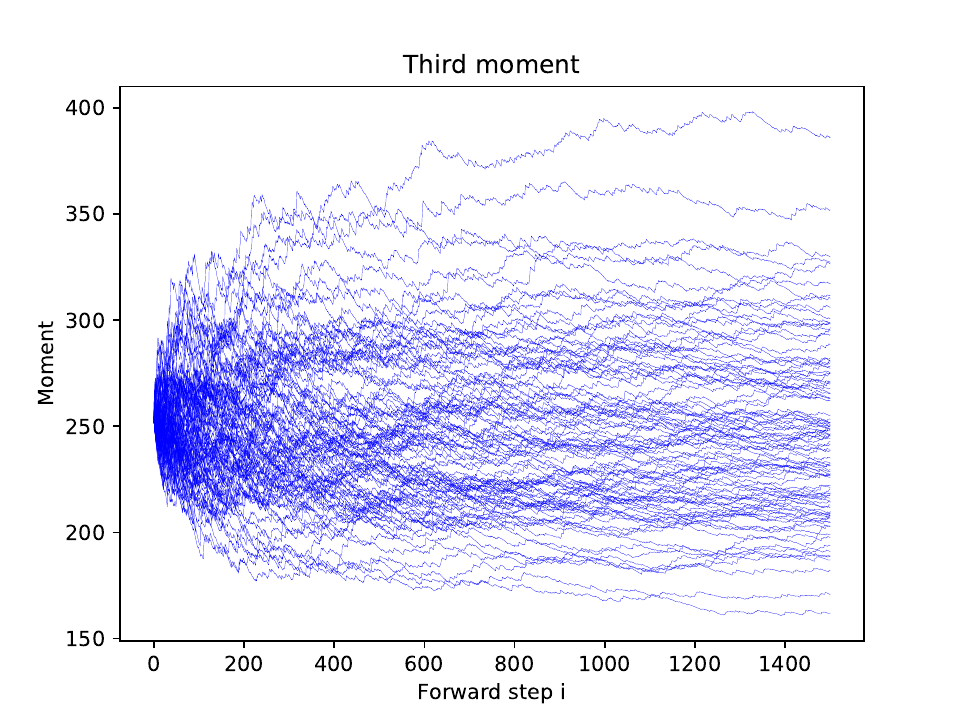}
    \end{subfigure}
    \vspace{-2mm}
    \caption{(Left) Posterior distribution of skewness in the misspecified case and (Right) sample paths of $\mu^{(3)}_{n+i}$ in the misspecified case.}
    \label{fig:4}
\end{figure}

\subsubsection{Small sample size}
Statistical inference for small datasets is  more challenging for the BB. This section reports the results of simulations conducted with a limited number of data points, demonstrating that the proposed method outperforms the BB, especially in estimating statistics that requires smoothness which the BB struggles with (e.g. Example \ref{exp1}). In both well-specified and misspecified scenarios, we consider a dataset size of $n = 7$.

Table \ref{table1} (right) displays the optimal values of $c$ across two scenarios, evaluated over 200 repeated samples of the datasets of size $n = 7$. In this case, the magnitude of $c$ in both cases are comparable as the sample size is likely too small to distinguish between the normal and the skew normal distribution.

In the misspecified scenario, we select a dataset for which $c=9.2$ is selected to demonstrate how the mixture MP can be helpful with small $n$. Fig. \ref{fig:5} (left) and (right) illustrate the posterior distributions of the 95th percentile and median respectively using the same setup as Example \ref{exp1}. We see that the mixture MP is smooth and continuous, whereas the BB posterior is concentrated on several values. This smoothness can be further observed in the histogram of posterior samples of the 95th percentile in Fig. \ref{fig:6} (left).
 Furthermore, the mixture MP intuitively lies in between the BB and the parametric MP. Fig. \ref{fig:6} (right) illustrates the posterior distribution of the skewness. Here, the BB greatly overestimates the skewness due to small sample size, while the parametric component of the mixture MP provides useful regularization towards 0. We defer the well-specified case to Appendix \ref{sec:app_sim}.
\begin{figure}[!ht]
    \centering
    \begin{subfigure}[t]{0.46\textwidth}
        \centering
        \includegraphics[width=\linewidth]{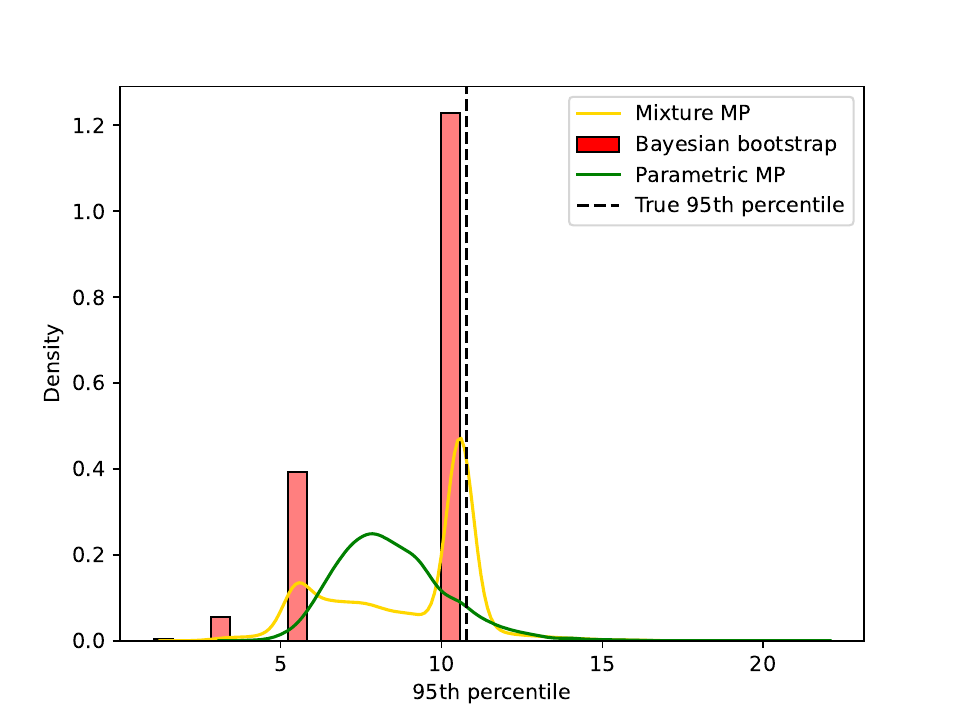} 
    \end{subfigure}
    \begin{subfigure}[t]{0.46\textwidth}
        \centering
        \includegraphics[width=\linewidth]{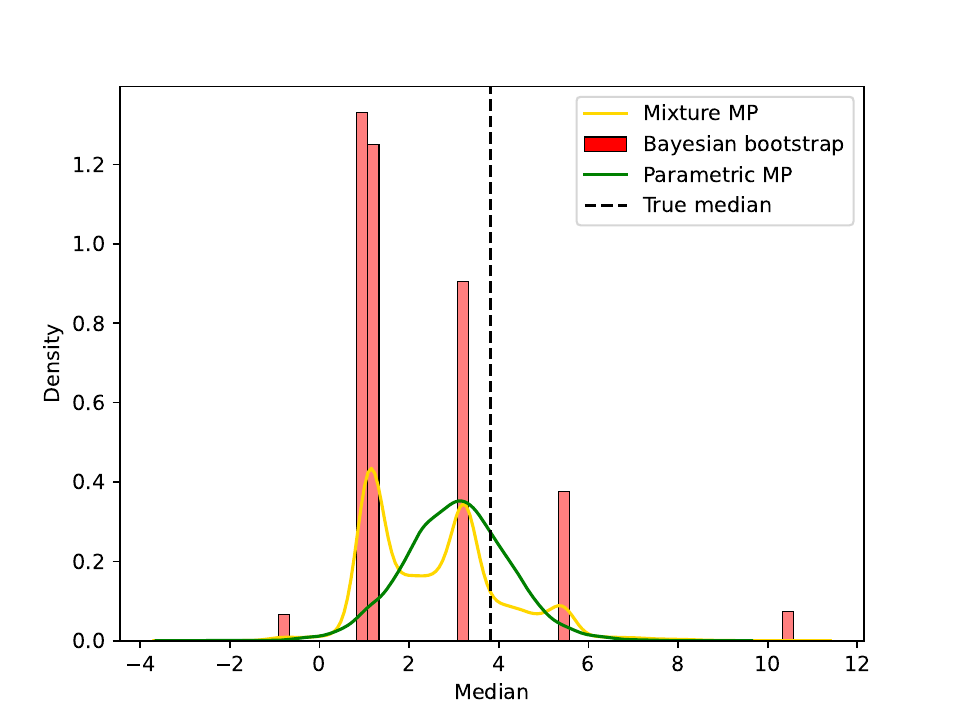} 
    \end{subfigure}\vspace{-2mm}
    \caption{Posterior distributions of the (Left) 95th percentile and (Right) median  in the misspecified scenario.}
    \label{fig:5}\vspace{-5mm}
\end{figure}
\begin{figure}[!ht]
\centering
    \begin{subfigure}[t]{0.46\textwidth}
        \centering
        \includegraphics[width=\linewidth]{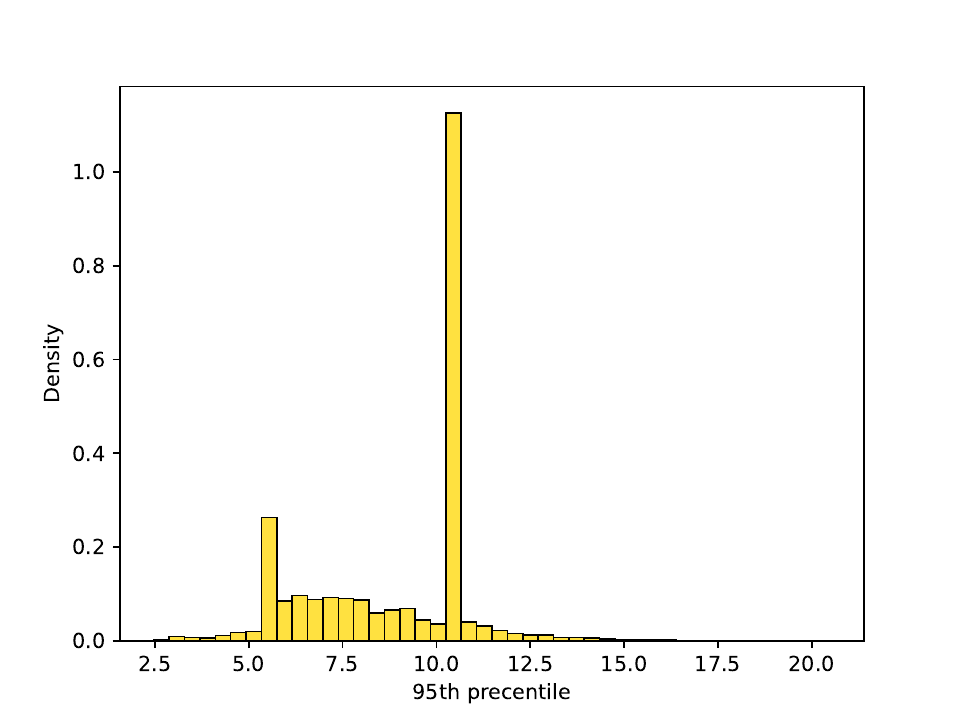} 
    \end{subfigure}
    \begin{subfigure}[t]{0.46\textwidth}
        \centering
\includegraphics[width=\linewidth]{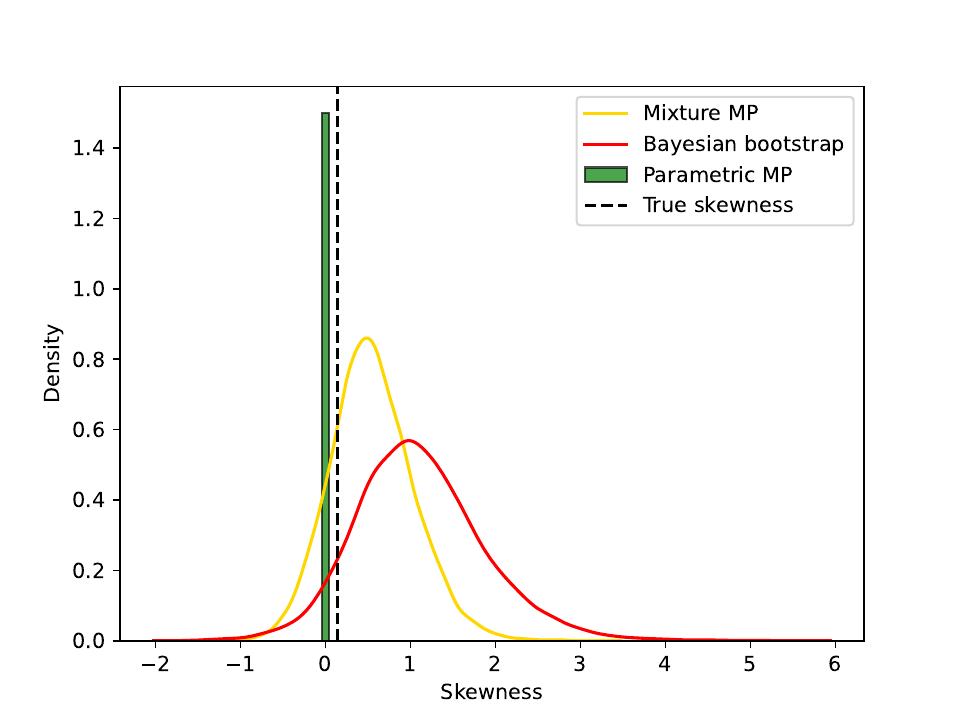} 
    \end{subfigure}\vspace{-2mm}
    \caption{(Left) Histogram of posterior samples of the 95th percentile obtained by mixture MP and (Right) posterior distribution of the skewness in the misspecified scenario.}
    \label{fig:6}
    \end{figure}

\subsection{Logistic regression} \label{sec:logreg}
The German Credit dataset \citep{statlog_(german_credit_data)_144} contains 1,000 data points of individuals applying for loans, and includes 20 features such as age, employment status, credit history, and loan amount. The binary target variable classifies applicants as ``good'' or ``bad'' credit risks. We employ this dataset in our logistic regression example. As suggested in Section \ref{sec:reg_prac}, we utilize an online approximation of the MLE, which is outlined in Appendix \ref{sec:app_regprac}.

Among the 1000 data points, we randomly choose 500 for training (including for calculating $c$) and 500 for testing. We first applied a non-Bayesian logistic regression model to the data in order to eliminate statistically insignificant features, resulting in the retention of 9 out of 20 features. Subsequently, we estimated the optimal value of $c$ using 5-fold CV on the 500 training data points, resulting in a 
selected  value of $c = 2000$ (Fig. \ref{fig:7} (left)). Note that we report the relative energy score for easier interpretation, which is simply the energy score with its value at $c = 0$ subtracted.

The posterior distributions of $\beta_3$ derived using the parametric MP and BB demonstrate some minor differences (Fig. \ref{fig:7} (right)). The mixture MP lies between the two but is more aligned with the parametric MP, which is because the weight parameter $c$ is quite large relative to the training sample size $n$. The similarities are unsurprising, as we do not include strong prior information regarding $\beta$.
We include the posteriors for other coefficients in Appendix \ref{sec:app_logreg}, which display a similar trend.

To highlight the substantial difference between the mixture MP and the BB in the regression case, we consider the posterior of the empirical \textit{conditional} mean, defined as
\begin{align*}
{\mu}_i^{y \mid x} &= \frac{i^{-1}\sum_{j = 1}^{i}\mathbf{1}(x_j = x)\,  y_j}{i^{-1}\sum_{j = 1}^i \mathbf{1}(x_j = x)}
\end{align*}
where $x \in \{x_1,\ldots,x_n\}$. The above is precisely the nonparametric estimate of $E[Y \mid X = x]$, and Fig. \ref{fig:8} (left) illustrates convergence of $\mu_i^{y \mid x}$ under predictive resampling, which we prove theoretically in Appendix \ref{sec:app_condit}. 
 Note that the convergence rate is much slower for $\mu_i^{y \mid x}$ than for the cross-moment $\mu_i^{yx}$, as $\mu_i^{y \mid x}$ is only updated when the specific atom $x$ is drawn from $\{x_1,\ldots,x_n\}$. As a result, we opt to use $N = 50000, B = 500$ specifically for Fig. \ref{fig:8}. An alternative option which would require fewer forward samples $N$ would be to consider $E[Y \mid X \in (x - \delta, x+\delta)$] for some small $\delta$.

\begin{figure}[ht]
    \centering
    \begin{subfigure}[t]{0.42\textwidth}
        \centering
\includegraphics[width=\linewidth]{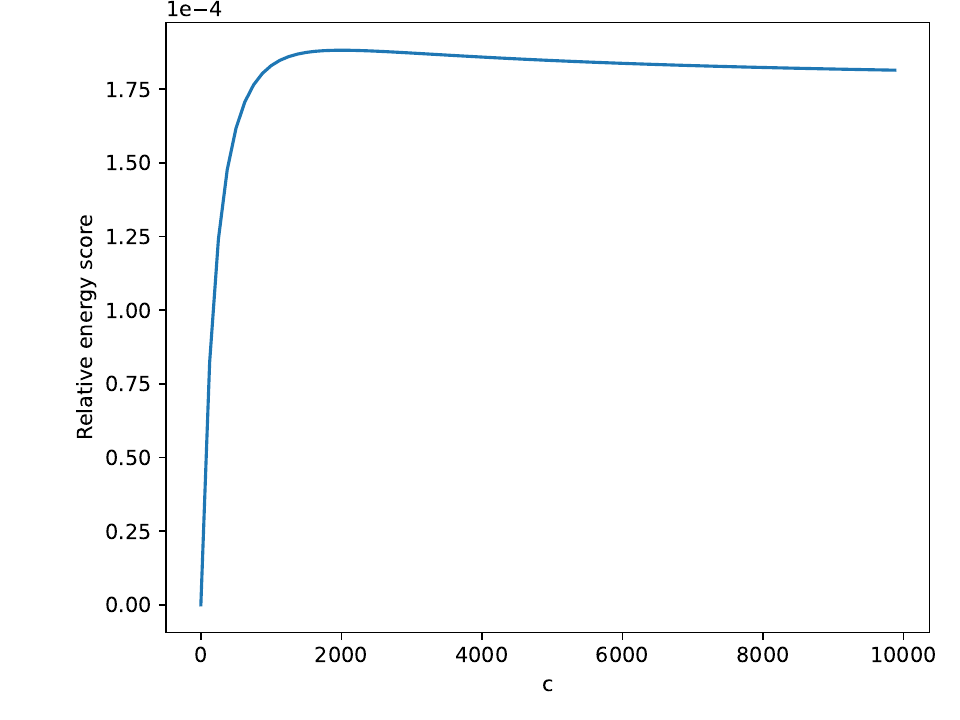} 
    \end{subfigure}
    \begin{subfigure}[t]{0.46\textwidth}
        \centering
\includegraphics[width=\linewidth]{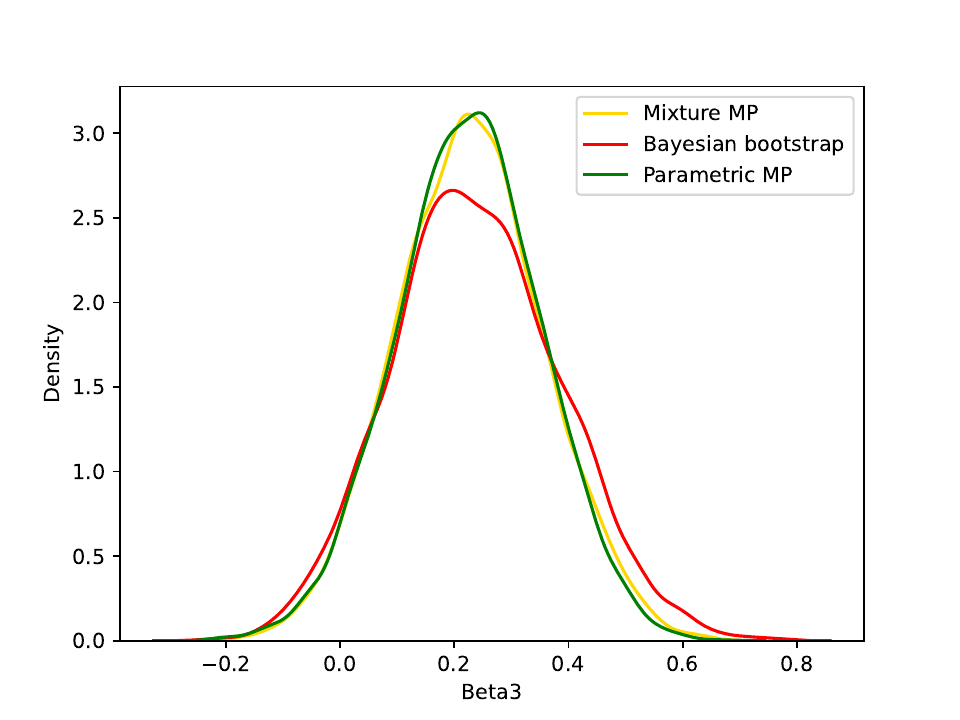} 
    \end{subfigure}\vspace{-2mm}
    \caption{(Left) Plot of $c$ against 5-fold CV relative energy score and (Right) posterior distribution of $\beta_3$; the relative energy score is translated to be equal to 0 at $c = 0$.}
    \label{fig:7}
\end{figure}

Fig \ref{fig:8} (right) illustrates the posterior distribution of $\mu_{N}^{y \mid  x_{150}}$, which concretely highlights the difference of the three approaches. The posterior distribution of $\mu_{N}^{y \mid  x_{150}}$ obtained by the BB is a point mass at $y_{150} = 0$, whereas that obtained via the parametric MP roughly has support over $(0,0.6)$ and is centered at $0.25$. In contrast, the posterior obtained via our mixture exhibits more substantial mass near 0, while also allowing for the generation of values in the range $(0,0.6)$, clearly interpolating between the parametric MP and the BB. This illustrates the key regularizing behavior of the parametric component.

\begin{figure}[!ht]
\centering
    \begin{subfigure}[t]{0.46\textwidth}
        \centering
\includegraphics[width=\linewidth]{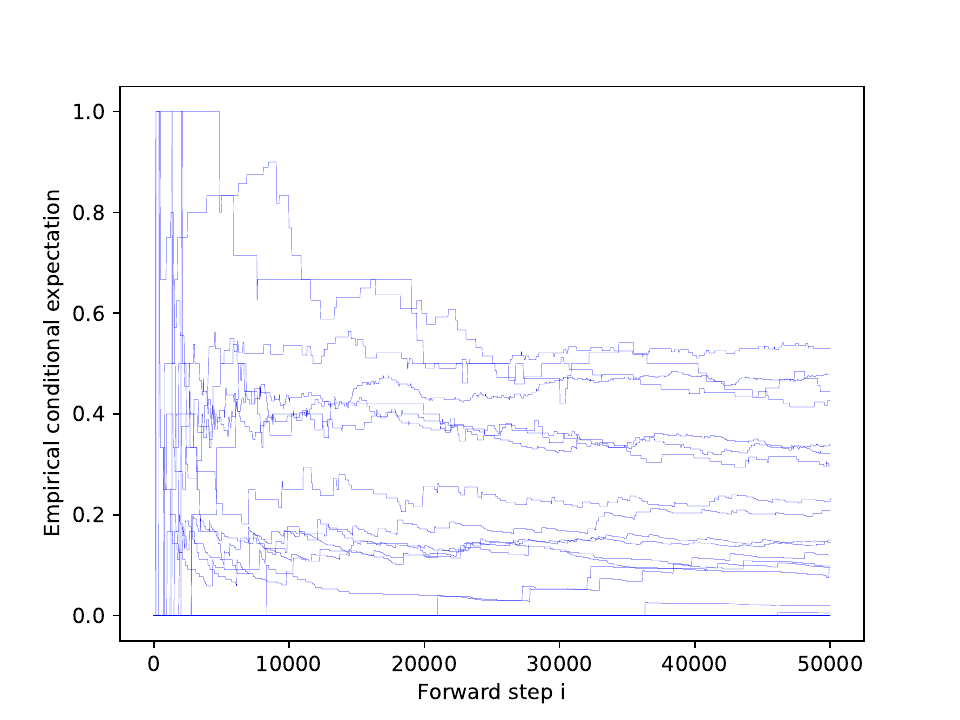} 
    \end{subfigure}
    \centering
    \begin{subfigure}[t]{0.46\textwidth}
        \centering
\includegraphics[width=\linewidth]{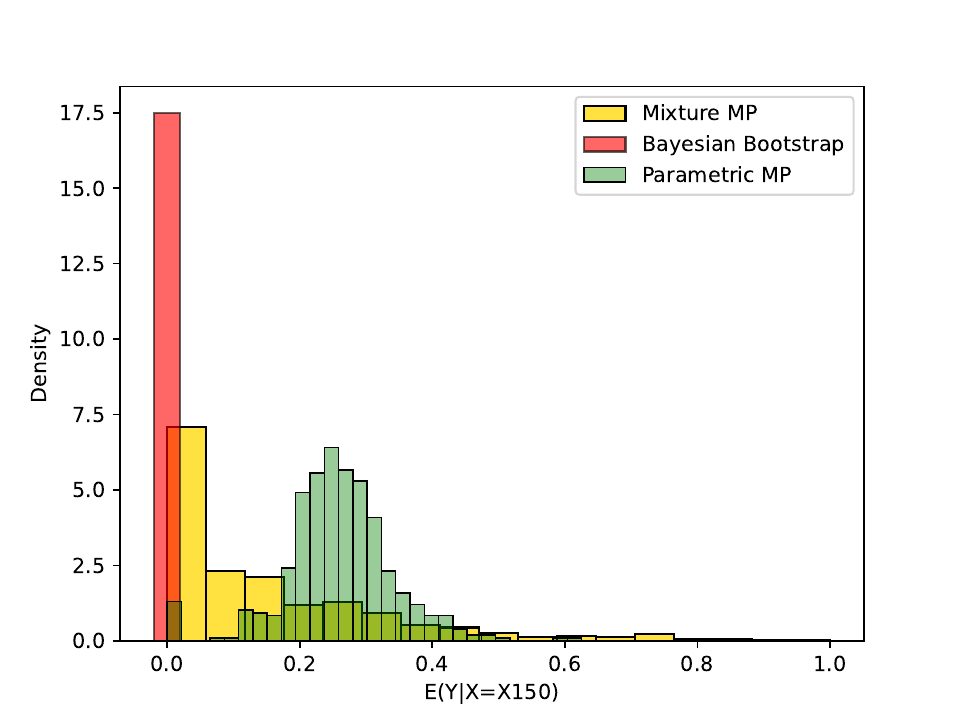} 
    \end{subfigure}
    \caption{(Left) Convergence of empirical conditional expectation $\mu_i^{y \mid x_{150}} $ 
    and (Right) posterior distribution of $\mu_N^{y \mid x_{150}}$; for this specific demonstration, we have $N = 50000$ and $B = 500$.}
    \label{fig:8}
    \end{figure}

\begin{table}[h]
\vspace{4mm}
\centering
    \scriptsize
    \begin{tabular}{c c c c c}
        \hline
          & \textbf{Mixture MP} & \textbf{Parametric MP} & \textbf{BB}\\ \hline
        \textbf{Relative energy score ($\times 10^{-4}$)} &1.622 & 1.616 & 0\\ 
        \textbf{Standard error ($\times 10^{-4}$)} & 0.041 & 0.045 & / \\ \hline
        \textbf{Run-times (s)} &6.47 & 6.30 & 3.21\\
        \hline
    \end{tabular}
    \caption{Hold out relative energy scores  (over 200 data splits) and run-times (seconds); for reference, Bayesian logistic regression with MCMC required 128 seconds. 
}
    \label{table3}
    \end{table}

To further evaluate the three methods, we compare their average predictive energy score on the hold out test data (Table \ref{table3}). Again, we report the relative energy score, which is the energy score with its value at $c = 0$ subtracted (i.e. the BB acts as the baseline). The predictive energy score of the mixture MP method is marginally better than that of the parametric MP method and significantly better than the BB. The predictive score of mixture MP and parametric MP is likely similar as the dataset is relatively linear. Our results thus empirically verifies the dominance of the predictive mixture as in (\ref{eq:dom}).

The run-times are presented in the last row of Table \ref{table3}. As a benchmark, Bayesian logistic regression with MCMC, which samples an effective sample size of 5000, has a runtime of 128s.  The methodology based on the MDP presented in \cite{Lyddon2018} requires MCMC, so its running time will exceed 128s. In contrast, our mixture MP completes the process in 6.5s. Additionally, although not demonstrated, our method is parallelizable while MCMC is not, which would drastically further enhances the speed of our approach. Finally, we highlight that the BB is still the fastest as we can skip predictive resampling by drawing the Dirichlet weights directly.

\section{Discussion}
We have introduced a novel approach for constructing the martingale posterior utilizing a semiparametric predictive distribution based on the method of moments. Simulation studies show that our approach can outperform the BB in the small sample size scenario, where the regularization of the parametric component can be useful. One potential future direction is to consider the performance under other nonparametric predictive components, such as those based on copulas. Another obvious direction is to extend the methodology to an \textit{ensemble} of predictive distributions, where it may be of interest to update the weights during imputation, in a similar fashion to \cite{shirvaikar2024general}.










\section*{Acknowledgments}
This work was supported by the Research Grants Council of Hong Kong through the General Research Fund (Grant No. 17307321) and the Early Career Scheme (Grant No. 27304424).
\bibliography{paper-ref}

\newpage
\appendix
\renewcommand{\theprop}{\thesection.\arabic{prop}} 
\renewcommand{\thecorollary}{\thesection.\arabic{corollary}} 
\renewcommand{\thetheorem}{\thesection.\arabic{theorem}} 
\renewcommand{\thelemma}{\thesection.\arabic{lemma}} 
\renewcommand{\thealgorithm}{\thesection.\arabic{algorithm}}
\renewcommand{\theequation}{\thesection.\arabic{equation}}
\renewcommand{\thefigure}{\thesection.\arabic{figure}}
\renewcommand{\thetable}{\thesection.\arabic{table}}
\renewcommand{\theexample}{\thesection.\arabic{example}}

\section*{Appendix notation}
 In the Appendices, we will use the following notation: 
\begin{center}
\begin{tabular}{ c|c } 
 Notation & Description \\ 
  \hline
 $\mathbb{Q}(y), q(y)$ & CDF/PDF of true distribution of $Y$  \\ 
 $P_n(y), p_n(y)$ & CDF/PDF of mixture predictive  \\ 
 $F_\theta, f_\theta; \, \theta \in \mathbb{R}^p$ & CDF/PDF of parametric component\\
  $\mu^{(k)}_\theta = E_{f_\theta}(Y^k)$ & $k$-th moment of $f_{\theta}$\\
  $|\mu_\theta|^{(k)} = E_{f_\theta}(|Y|^k)$ & $k$-th absolute moment of  $f_{\theta}$\\
  $\mu^{(k)}_n=n^{-1}\sum_{i=1}^n y_i^k$ & $k$-th empirical moment \\
  $|\mu_n|^{(k)}=n^{-1}\sum_{i=1}^n |y_i|^k$ & $k$-th empirical absolute moment\\
  $\theta_n = h(\mu_n^{(1)},\ldots, \mu_n^{(p)})$  &  Estimated from $Y_{1:n}\iid \mathbb{Q}$\\
    $\theta^* = \lim_{n \to \infty} \theta_n$ & Limit of $\theta_n$ estimated from $Y_{1:n} \iid \mathbb{Q}$ \\
    $\theta_i = h(\mu_i^{(1)},\ldots, \mu_i^{(p)})$ & Estimated from predictive resampling \\
  $\theta_\infty = \lim_{i \to \infty} \theta_i$ & Limit of $\theta_i$ from predictive resampling\\
       $\mathcal{F}_i=\sigma(Y_1,\ldots,Y_i)$ &  $\sigma$-algebra generated by $Y_{1:i}$ 
\end{tabular}
\end{center}\vspace{5mm}

For the regression setting, we introduce some additional notation.
\begin{center}
\begin{tabular}{ c|c } 
 Notation & Description \\ 
  \hline
       $x\in \mathbb{R}^p$ & Deterministic covariate vector\\
            $X\in \mathbb{R}^p$ & Random covariate vector\\
  $\mathbf{X}_n\in \mathbb{R}^{n\times p}$ & Design matrix \\
       $y \in \mathcal{Y}$& Deterministic response\\
       $Y \in \mathcal{Y}$ & Random response\\ 
     $\mathbf{Y}_n \in \mathcal{Y}^n$ & Column vector of responses\\
     $\beta \in \mathbb{R}^p$ & Column vector of GLM coefficients
\end{tabular}
\end{center}\vspace{5mm}

\setcounter{prop}{0}
\setcounter{corollary}{0}
\setcounter{theorem}{0}
\setcounter{lemma}{0}
\setcounter{algorithm}{0}
\setcounter{equation}{0}
\setcounter{figure}{0}
\setcounter{table}{0}
\setcounter{example}{0}
\section{Appendix: Martingale convergence}\label{sec:app_mart}
\subsection{Proof of Proposition \ref{prop:mom_mart}} 
 For any $i\in \mathbb{N}$, we have that
\begin{flalign*}
E[\mu^{(k)}_{i+1} \mid \mathcal{F}_i]&\nonumber= E\left[\frac{1}{i+1}\sum_{j=1}^{i+1} Y^k_j\mid \mathcal{F}_i\right] &&\\\nonumber
&= \frac{1}{i+1} E\left[Y^k_{i+1}\mid \mathcal{F}_i\right] + \frac{1}{i+1}\sum_{j=1}^i y_{j}^k &&\nonumber\\
&= \frac{1}{i+1} \left[\frac{c}{c+i}\mu^{(k)}_{\theta_i}+ \frac{i}{c+i}\mu_{i}^{(k)}\right]+ \frac{i}{i+1}\mu_{i}^{(k)}&&\\\nonumber
&= \frac{1}{i+1} \left[\frac{c}{c+i} \mu_{i}^{(k)} + \frac{i}{c+i}\mu_{i}^{(k)}\right]+ \frac{i}{i+1}\mu_{i}^{(k)}&&\\\nonumber
&= \mu_{i}^{(k)} &&\nonumber
\end{flalign*}
where we have used the moment matching in the 4th line. 

\subsection{Proof of Theorem \ref{thm:1}} \label{sec:app_proofth1}
We split the proof into the two cases of even and odd $p$, where $p$ is the number of moments being matched.
\subsubsection{Even $p$}
We begin by showing convergence of the $p$-th moment, which is easy as $\mu_i^{(p)}$ is a non-negative martingale.
\begin{lemma}\label{lem:positive moment} If $p$ is even, then the $p$-th moment $\mu_i^{(p)}$ converges a.s. under the moment martingale posterior.
\end{lemma}
\begin{proof} For any $i \geq n$, we have that
\begin{align*}
E\left[|\mu^{(p)}_i| \hspace{0.1cm} \mid \mathcal{F}_n\right]\nonumber= E\left[\mu^{(p)}_i \mid \mathcal{F}_n\right]  = \mu^{(p)}_n < \infty.
\end{align*}
Therefore $\sup_{i \geq n} E\left[|\mu^{(p)}_i| \hspace{0.1cm} \mid \mathcal{F}_n\right] = \mu^{(p)}_n < \infty$, so by Doob's martingale convergence theorem, we have
$\mu^{(p)}_i \to \mu^{(p)}_\infty$ a.s.
\end{proof}
For $k < p$, it is clear from the moment matching condition that $\mu_i^{(k)}$ is still a martingale. What remains is to show that the martingales are bounded in $L_1$, which now follows.
\begin{lemma}\label{lem:pImplyLower} If $p$ is even, for any $k < p$, the $k$-th moment is bounded in $L_1$ under the moment martingale posterior.
\end{lemma}
\begin{proof} 
A standard result from Jensen's and H\"{o}lder's inequalities is that 
$$
\left|E\left[X^k\right]\right| \leq E\left[\left|X\right|^k\right]\leq E\left[\left|X \right|^p\right]^{k/p}
$$
for $p > k$. In our setting then, we have
\begin{align*}
    \left|\mu_i^{(k)} \right| \leq \left(\left|\mu_i\right|^{(p)}\right)^{k/p} = \left(\mu_i^{(p)}\right)^{k/p}  
\end{align*}
where the equality follows from the fact that $p$ is even. Taking expectations, we have 
\begin{align*}
    E\left[\left|\mu_i^{(k)} \right|\mid \mathcal{F}_n\right] &\leq E\left[\left(\mu_i^{(p)}\right)^{k/p} \mid \mathcal{F}_n \right]  \\
    &\leq E\left[\mu_i^{(p)} \mid \mathcal{F}_n \right]^{k/p} \end{align*}
    where the second inequality follows from Jensen's inequality.
From Lemma \ref{lem:positive moment}, we have that $\sup_{i \geq n}E\left[\mu_i^{(p)} \mid \mathcal{F}_n \right] < \infty$, which implies the desired result
$\sup_{i \geq n} E\left[\left|\mu_i^{(k)} \right|\mid \mathcal{F}_n\right] \leq \infty$.
\end{proof}
Putting the above together, since $\mu_i^{(k)}$ is a martingale bounded in $L_1$, we have the desired a.s. convergence.

\subsubsection{Odd $p$}
When $p$ is odd, we require the additional assumption given in Theorem \ref{thm:1}, which allows us to control the absolute value of the $p$-th moment. Once again, we verify that the moment martingales are bounded in $L_1$.
\begin{lemma}\label{lem:oddp} If $p$ is an odd integer,  then the $p$-th moment $\mu_i^{(p)}$ is bounded in $L_1$ if there exist some non-negative $a$ and $b$ such that
    $$|\mu_\theta|^{(p)} \leq a|\mu^{(p)}_\theta |+ b.$$
\end{lemma}
\begin{proof}
First, we analogously define the empirical absolute value moment as
$$
|\mu_n|^{(p)} = \frac{1}{n}\sum_{i = 1}^n |Y_i|^p.
$$
It is immediate that $|\mu_i^{(p)}| \leq |\mu_i|^{(p)}$, so we study the latter. We note that
\begin{align*}
        E\left[|\mu_{i+1}|^{(p)}\Big|\mathcal{F}_{i}\right] 
        &= \frac{i}{i+1} |\mu_i|^{(p)}+ \frac{1}{i+1} E\left[|Y_{i+1}|^{p} \mid \mathcal{F}_{i}\right] 
        \end{align*}
        where
        \begin{align*}
    E\left[|Y_{i+1}|^{p} \mid \mathcal{F}_{i}\right] &=\frac{c}{c+i}|\mu_{\theta_i}|^{(p)} + \frac{i}{c+i}|\mu_i|^{(p)} 
    \end{align*}
    Putting this together, we have
    \begin{align}\label{eq:mom_rec}
        E\left[|\mu_{i+1}|^{(p)}\Big|\mathcal{F}_{i}\right] &= |\mu_i|^{(p)} + \frac{c}{(i+1)(c+i)}\left[|\mu_{\theta_i}|^{(p)} - |\mu_i|^{(p)}\right] \nonumber\\
        &\leq |\mu_i|^{(p)} + \frac{c}{(i+1)(c+i)}|\mu_{\theta_i}|^{(p)} 
    \end{align}
Next, we apply the assumption 
\begin{align*}
    |\mu_\theta|^{(p)} &\leq a|\mu^{(p)}_\theta| + b  \\
    &=a|\mu^{(p)}_i|  + b\\
    &\leq a|\mu_i|^{(p)} + b
\end{align*}
where the second line follows from moment matching. Plugging this back into (\ref{eq:mom_rec}) gives
\begin{align*}
      E\left[|\mu_{i+1}|^{(p)}\Big|\mathcal{F}_{i}\right] 
        &\leq |\mu_i|^{(p)} \left[ 1 + \frac{ac}{(i+1)(c+i)}\right] + \frac{bc}{(i+1)(c+i)}\\
         &\leq |\mu_i|^{(p)} \left[ 1 + O(i^{-2})\right] + O(i^{-2})
\end{align*}
Iterating the expectation, we can show that
\begin{align*}
    E\left[|\mu_{i+1}|^{p}\mid \mathcal{F}_n\right] \leq |\mu_n|^{(p)}\prod_{j = n}^i \left[1 + O(j^{-2})\right] + \sum_{j = n}^i O(j^{-2})\prod_{k = j+1}^{i} \left[1 + O(k^{-2})\right]
    \end{align*}
Since $|\mu_n|^{(p)} < \infty$, standard arguments give $\sup_i E\left[|\mu_i|^{(p)} \mid \mathcal{F}_n\right] < \infty$, which gives the desired $\sup_i E\left[|\mu_i^{(p)}|\right] < \infty$.
\end{proof}
Since $p-1$ is even, we can then apply Lemmas \ref{lem:positive moment} and \ref{lem:pImplyLower} to show convergence of all moments $k < p$ in the odd case.
 
In practice, we may not want to work with the absolute moment, as it may involve manually computing an integral. Instead, an easier assumption to check is based on the following.
\begin{lemma}\label{lem:oddp2} Suppose $p$ is an odd integer. Let us assume that the parametric $(p+1)$-th moment exists and there exist some non-negative constants $a'$ and $b'$ such that
$$
\mu_\theta^{(p+1)} \leq  a' |\mu_\theta^{(p)}|^{\frac{p+1}{p}} + b'.
$$
Then the condition of Lemma \ref{lem:oddp} holds, i.e. there exist some non-negative $a$ and $b$ such that
    $$|\mu_\theta|^{(p)} \leq a|\mu^{(p)}_\theta |+ b.$$
\end{lemma}
\begin{proof}
The assumption in Theorem \ref{thm:1} is an upper bound on $|\mu_\theta|^{(p)}$. From H\"{o}lder's inequality, we have
\begin{align*}
|\mu_\theta|^{(p)} &\leq \left(|\mu_\theta|^{(p+1)}\right)^{\frac{p}{p+1}} = \left(\mu_\theta^{(p+1)}\right)^{\frac{p}{p+1}}
\end{align*}
where the equality follows as $p+1$ is even. Plugging in the assumption, we have
$$
|\mu_\theta|^{(p)} \leq \left(a' |\mu_\theta^{(p)}|^{\frac{p+1}{p}} + b'\right)^{\frac{p}{p+1}}.
$$
One can thus find $a$ and $b$ such that
$$
|\mu_\theta|^{(p)} \leq a |\mu_\theta^{(p)}| + b
$$
as required.
\end{proof}
The above assumption, while stronger, may be nicer to check in practice, as it does not depend on any absolute moments.

\subsubsection{Example}
\begin{example}
Consider a three parameter gamma distribution with $\alpha,\beta,\theta$ representing shape, scale and location respectively. Here, as $p = 3$, we verify the condition of Lemma \ref{lem:oddp2}. As we have $$\mu_\theta^{(3)}=\theta^3+3\theta^2\beta\alpha+3\theta\beta^2\alpha(\alpha+1)+\beta^3\alpha(\alpha+1)(\alpha+2)$$
and 
$$\mu_\theta^{(4)}=\theta^4+4\theta^2\beta\alpha+6\theta^2\beta^2\alpha(\alpha+1)+4\theta\beta^3\alpha(\alpha+1)(\alpha+2)+\beta^4\alpha(\alpha+1)(\alpha+2)(\alpha+3),$$
we obviously have $$\mu_\theta^{(4)}=O\left(\left(|\mu_\theta^{(3)}|\right)^{\frac{4}{3}}\right).$$
\end{example}

\subsection{Proof of Proposition \ref{prop:exist}}
\subsubsection{Almost supermartingale}
Before we prove Proposition \ref{prop:exist}, we present a specialization of the almost supermartingale convergence theorem of \citet{robbins1971convergence}, which will be used in the proof of Proposition \ref{prop:exist} and \ref{prop:higher}, as well as later results. The original result is more general than the version we provide below.

\begin{lemma} [\citet{robbins1971convergence}]
    A non-negative adapted sequence of random variables $X_1,X_2,\ldots$ is an almost supermartingale if it satisfies
    $$E[X_{i+1}|\mathcal{F}_{i}]\leq (1+a_i)X_{i}+b_i$$
    where $a_i, b_i$ are non-negative adapted random variables. 
    Suppose all the following are satisfied:
    \begin{enumerate}
        \item $\sum_{i=1}^\infty a_i<\infty$ a.s., and 
        \item $\sum_{i=1}^\infty b_i<\infty$ a.s.
    \end{enumerate}
    Then we have $X_i{\to} X_\infty$ a.s., where $ X_\infty$ is finite a.s.
    \label{lemma:Convergence theorem of almost-supermartingale}
\end{lemma}

\subsubsection{Weak convergence}
For Proposition \ref{prop:exist}, it suffices to prove that $\int g(y)\,dP_i(y)$ converges almost surely for all bounded and continuous functions $g$, as we can then apply  \citet[Theorem 2.2]{berti2006almost}. 
    Note that 
    $$\int g(y)\,dP_i(y)=\frac{c}{c+i}\int g(y)\,dF_{\theta_i}(y)+\frac{i}{c+i}\mathbb{P}_i[g]$$
    where $\mathbb{P}_i[g] = n^{-1}\sum_{j = 1}^i g(y_j)$. Since $g$ is bounded, there exists a finite scalar $B>0$ such that
\begin{align*}
    \int g(y)\, dF_{\theta_i}(y) \leq B, \quad \mathbb{P}_i[g] \leq B.
\end{align*}
This then implies
\begin{align*}
\left|\int g(y)\,dP_i(y) - \mathbb{P}_i[g] \right|&=\frac{c}{c+i}\left|\int g(y)\, dF_{\theta_i}(y) - \mathbbm{P}_i[g]\right|\\
&\leq \frac{2Bc}{c+i}\to 0 
\end{align*}
so it suffices to show almost sure convergence of $\mathbb{P}_i[g]$. 

Consider the following decomposition:  $g^+(y)=g(y) \mathbbm{1}\{g(y)\geq0\}$ and $g^-(y)=|g(y)| \mathbbm{1}\{g(y)<0\}$, which gives
$$
\mathbb{P}_i[g] = \mathbb{P}_{i}[g^+] - \mathbb{P}_{i}[g^-]
$$
 We aim to prove that both $\mathbb{P}_{i}[g^+]$ and  $\mathbb{P}_{i}[g^-]$ converge a.s., which implies their difference also converges. 
    To begin, we prove the convergence of $\mathbb{P}_{i}[g^+]$ by showing it is a almost supermartingale. We have that
    \begin{flalign} E\left[\mathbb{P}_{i+1}\left[g^+\right]\Big|\mathcal{F}_{i}\right] =  \frac{i}{i+1} \mathbb{P}_{i}\left[g^+\right]+ \frac{1}{i+1} E\left[g^+(Y_{i+1})\Big|\mathcal{F}_{i}\right] \label{eq:empirical moment of g(y)}
    \end{flalign}
    where
    $$E\left[g^+(Y_{i+1})\Big|\mathcal{F}_{i}\right]=\frac{c}{c+i} \int g^+(y)\,dF_{\theta_i}(y) + \frac{i}{c+i}\mathbb{P}_{i}\left[g^+\right]. $$
    The boundedness of $g$ immediately gives
 \begin{align*}
    E\left[g^+(Y_{i+1})\Big|\mathcal{F}_{i}\right]\leq \frac{Bc}{c+i}  + \frac{i}{c+i}\mathbb{P}_i[g^+]
\end{align*}
Plugging this into the (\ref{eq:empirical moment of g(y)}) gives
\begin{align*}
    E\left[ \mathbb{P}_{i+1}\left[g^{+}\right] \mid \mathcal{F}_i\right] \leq \left[\frac{i}{i+1}+ \frac{i}{(c+i)(i+1)}\right]\mathbb{P}_i\left[g^+\right] + \frac{Bc}{(c+i)(i+1)}
\end{align*}
Since $N/(c+i) \leq 1$, we have
\begin{align*}
        E\left[ \mathbb{P}_{i+1}\left[g^{+}\right] \mid \mathcal{F}_i\right] \leq \mathbb{P}_i\left[g^+\right] + \frac{Bc}{(c+i)(i+1)}
\end{align*}
so $\mathbb{P}_i[g^+]$ is an almost supermartingale. It is clear that
$$
\sum_{i = n}^\infty \frac{Bc}{(c+i)(i+1)} < \infty,
$$
so we can apply Lemma \ref{lemma:Convergence theorem of almost-supermartingale} to guarantee a.s. convergence of $\mathbb{P}_i[g^+]$.

The same proof can be used to show that $\mathbb{P}_{i}\left[g^{-}\right]$ converges a.s. Together we have that $\mathbb{P}_i[g]$ converges a.s. for each bounded and continuous $g$.

\subsection{Proof of Proposition \ref{prop:higher}}\label{sec:app_proofhigher}
    The goal is to show that $\mu_i^{(k)}$ converges almost surely, where $k \geq p$.
    \subsubsection{Even $k$}
    We begin with the case where $k$ is even, as $\mu_i^{(k)}$ is non-negative. To begin, we can carry out the same computation as (\ref{eq:mom_rec}) in  Lemma \ref{lem:oddp} to obtain:
    \begin{align*}
        E\left[\mu_{i+1}^{(k)}\Big|\mathcal{F}_{i}\right] &= \mu_i^{(k)} + \frac{c}{(i+1)(c+i)}\left[\mu_{\theta_i}^{(k)} - \mu_i^{(k)}\right]\\
        &\leq \mu_i^{(k)} + \frac{c}{(i+1)(c+i)}\mu_{\theta_i}^{(k)} 
    \end{align*}
It is clear that $\mu_{i+1}^{(k)}$ is an almost supermartingale, so from Lemma \ref{lemma:Convergence theorem of almost-supermartingale}, it converges if
\begin{align*}
    \sum_{j = n}^\infty \frac{c}{(j+1)(c+j)} \, \mu_{\theta_j}^{(k)} < \infty \quad \text{a.s.}
\end{align*}
By assumption, $\mu_{\theta}^{(k)}$ is bounded by a continuous function $g(\theta)$. Since $\theta_i \to \theta_\infty$ a.s., we have  $g(\theta) \to g(\theta_\infty) < \infty$ a.s. 
Finally, this implies $\mu_{\theta_i}^{(k)}$ is bounded a.s., so the above term is finite a.s.

\subsubsection{Odd $k$}
When $k$ is odd, the story is slightly more challenging as $\mu_i^{(k)}$ is not guaranteed to be non-negative. Following the proof of Proposition \ref{prop:mom_mart}, we can decompose the moment into
\begin{align*}
    \mu_i^{(k)} = \mu_i^{(k)+} - \mu_i^{(k)-}
\end{align*}
where
\begin{align*}
    \mu_i^{(k)+} &= \frac{1}{i}\sum_{j = 1}^i \mathbbm{1}\left(y_j^k \geq 0\right) y_j^k, \\
    \mu_i^{(k)-} &= \frac{1}{i}\sum_{j = 1}^i \mathbbm{1}\left(y_j^k < 0\right) \left|y_j^k\right|.
\end{align*}
We can compute
\begin{align*}
    E\left[\mu_i^{(k)+} \mid \mathcal{F}_i\right] &= \frac{i}{i+1} \mu_{i}^{(k)+} + \frac{1}{i+1}E\left[\mathbbm{1}\left(Y_{i+1} \geq 0\right) Y_{i+1}^k \mid \mathcal{F}_i\right]\\
    &=   \mu_i^{(k)+} + \frac{c}{(i+1)(c+i)}\left[\mu_{\theta_{i}}^{(k)+} - \mu_i^{(k)+}\right]\\
    &\leq    \mu_i^{(k)+} + \frac{c}{(i+1)(c+i)}\mu_{\theta_{i}}^{(k)+}
\end{align*}
where the definition of $\mu_{\theta}^{(k)+}$ is
\begin{align*}
    \mu_{\theta}^{(k)+} = \int \mathbbm{1}\left(y^k \geq 0\right)y^k\,  f_{\theta}(y) \, dy.
\end{align*}
Again, this is an almost supermartingale, which converges if
\begin{align*}
    \sum_{j = n}^i \frac{c}{(j+1)(c+j)}\, \mu_{\theta_j}^{(k)+} < \infty \quad \text {a.s.}
\end{align*}
It is clear that
\begin{align*}
    \mu_{\theta}^{(k)+} \leq \left|\mu_{\theta}\right|^{(k)}
\end{align*}
so $\mu_{\theta_i}^{(k)+}$ is bounded a.s. if $\left|\mu_{\theta_i}\right|^{(k)}$ converges a.s. 
By assumption, $|\mu_{\theta}|^{(k)}$ is bounded by $g(\theta)$ which is continuous in $\theta$, so $\left|\mu_{\theta_i}\right|^{(k)}$ is bounded a.s. again. The same proof can then be applied to show a.s. convergence of $\mu_{i}^{(k)-}$, so $\mu_i^{(k)} \to \mu_{\infty}^{(k)}$ a.s.

\subsubsection{Example}
 We now demonstrate the verification of the assumptions in Proposition \ref{prop:higher} in the below examples.
\begin{example}
For $f_\theta = \mathcal{N}(\mu,\sigma^2)$ where $\theta = (\mu,\sigma^2)$, the 4-th non-central moment is
$$
\mu_\theta^{(4)} = \mu^4 + 6\mu^2 \sigma^2 + 15 \mu \sigma^4.
$$
This is clearly a continuous function of $(\mu,\sigma^2)$ and is bounded for each $(\mu,\sigma^2)$. Furthermore, $k$ is even so $|\mu_\theta|^{(4)} = \mu_\theta^{(4)}$ giving us the required upper bound.

For $k = 3$, we can avoid calculating the absolute moment by using H\"{o}lder's inequality:
$$
|\mu_{\theta}|^{(3)} \leq \left(\mu_\theta^{(4)}\right)^{\frac{3}{4}}.
$$
We can then directly use $g(\theta) = (\mu_\theta^{(4)})^{\frac{3}{4}}$, which is again continuous and bounded for each $\theta$.
\label{exp3}
\end{example}\vspace{10mm}

\setcounter{prop}{0}
\setcounter{corollary}{0}
\setcounter{theorem}{0}
\setcounter{lemma}{0}
\setcounter{algorithm}{0}
\setcounter{equation}{0}
\setcounter{figure}{0}
\setcounter{table}{0}
\setcounter{example}{0}
\section{Appendix: Energy score} \label{sec:app_energy}
\subsection{Energy score and distance}\label{sec:app_lam}
A standard result is that the squared energy distance between $\mathbb{Q}$ and $P_n$ can be written as
\begin{align*}
D^2(\mathbb{Q},P_n) &= - S(\mathbb{Q},P_n)-E|Y-Y'|\\
    &= 2E[|Y-X|] - E[|X-X'|]  -  E[|Y-Y'|] 
\end{align*}
where $Y,Y'\sim \mathbb{Q}$ and $X,X'\sim P_n$. Since the term $E|Y-Y'|$ does not contain $c$, maximizing the expected energy score is the same as minimizing the energy distance. The proofs presented in  Appendix \ref{sec:app_energy} will thus consider minimizing the energy distance as this simplifies the proof.

To avoid cumbersome notation in the proofs, we deviate from the notation of the main paper. Let our hold out test set be of size $v$ and be denoted as $\{{y}_i\}_{i = 1}^v$, and the training set of size $n$ will be written as $\{x_i\}_{i = 1}^n$. We denote the empirical distribution of the test and training set as $\mathbb{P}_v$ and $\mathbb{P}_n$ respectively.
It is clear that in practice, the training and test set will be partitions of the $(n+v)$ total observations.

First, consider the following theoretical limit terms:
\begin{align*}
    A&=\iint |y'-y|\, q(y')\, q(y) \,dy'\,dy\\
    B&=\iint |x-y|\, f_{\theta^*}(x)\,q(y) \,dx \, dy\\
    C&=\iint |x-x'|\,f_{\theta^*}(x)\,f_{\theta^*}(x')\, dx\,dx'
    \end{align*}
where $\theta^*$ is the limit of $\theta_n$ (in probability) and $q$ is the density of $\mathbb{Q}$. We will also consider various empirical estimates of the above terms:
    \begin{align*}
    \Psi_{A,v}=\frac{1}{n v}\sum_{i=1}^n \sum_{j = 1}^v |y_j-x_i|, &\quad 
     \Psi_{A,n} =\frac{1}{n^2}\sum_{i =1}^{n} \sum_{j = 1} ^{n}|x_i-x_j|,\\
    \Psi_{B,v}=\frac{1}{v}\sum_{j = 1}^v\int |y_j-x|f_{\theta_n}(x)\, dx, &\quad 
        \Psi_{B,n}=\frac{1}{n}\sum_{i=1}^{n} \int |x'-x_i| f_{\theta_{n}}(x') dx',\\
    \Psi_{C,n} =\iint |x-x'| &f_{\theta_{n}}(x) f_{\theta_{n}}(x')\, dx \, dx'.
\end{align*}
The subscripts $\{A,B,C\}$ indicates which theoretical quantity the term is approximating, whereas $\{v, n\}$ indicate whether the validation set is involved or not.

\subsection{Proof of Proposition \ref{prop:concave}}

\subsubsection{Hold out energy distance} The following proposition immediately implies Proposition \ref{prop:concave} for the hold-out setting, that is for the score function given by (\ref{eq:est_lam}) and the optimal $\widehat{\lambda}_n$.
\begin{prop}\label{prop:d2_quad}
 The hold out squared energy distance $D^2(\mathbb{P}_v,P_n)$ can be written as a strictly convex quadratic function of $\lambda=c/(c+n)$ as follows:
\begin{flalign}
    D^2(\mathbb{P}_{v},P_n) &= (2\Psi_{B,n}-\Psi_{A,n}-\Psi_{C,n}) \, \lambda^2 - 2(\Psi_{A,v}-\Psi_{A,n}-\Psi_{B,v}+\Psi_{B,n})\, \lambda+K\label{eq:DLambda1}
\end{flalign}
where $K$ contains terms that do not depend on $\lambda$, and the coefficient of $\lambda^2$ is strictly positive. Furthermore, for $\lambda\in [0,1]$, the above is minimized at
\begin{align}
\widehat{\lambda}_n=\min\left\{1,\max\left\{0,\dfrac{\Psi_{A,v} - \Psi_{A,n} - \Psi_{B,v} + \Psi_{B,n}}{2\Psi_{B,n} - \Psi_{A,n}- \Psi_{C,n}}\right\}\right\}
\cdot \label{eq:app_opt_lam}
\end{align}
\label{lem:optimise lambda}
\end{prop}
\begin{proof}
Expanding out the predictive mixture form, the energy distance takes the form
\begin{align*}
     D^2(\mathbb{P}_{v},P_n)&= 2\left[\lambda\Psi_{B,v} + (1-\lambda)\Psi_{A,v}\right] -\left[\lambda^2 \Psi_{C,n} +2\lambda(1-\lambda)\Psi_{B,n}+\left(1-\lambda\right)^2 \Psi_{A,n}\right]\nonumber\\
     &- \iint |y-y'| \, d\mathbb{P}_{v}(y) \, d\mathbb{P}_v(y')\\
    &= (2\Psi_{B,n}-\Psi_{A,n} -\Psi_{C,n})\lambda^2 - 2(\Psi_{A,v}-\Psi_{A,n}-\Psi_{B,v}+\Psi_{B,n})\lambda + K 
    \end{align*}
Completing the square then gives
\begin{align*}
   D^2(\mathbb{P}_{v},P_n) &= (2\Psi_{B,n} -\Psi_{A,n}-\Psi_{C,n}) \left(\lambda-\frac{\Psi_{A,v} -\Psi_{A,n}-\Psi_{B,v}+\Psi_{B,n}}{2\Psi_{B,n}-\Psi_{A,n}-\Psi_{C,n}}\right)^2 \nonumber + K'
\end{align*}
where again $K'$ does not depend on $\lambda$.  We note the key fact that
$$2\Psi_{B,n}-\Psi_{A,n}-\Psi_{C,n}=D^2(F_{\theta_n},\mathbb{P}_n) > 0$$
for each $n$. Note that the inequality is strict, as $F_{\theta}$ is a continuous distribution and $\mathbb{P}_n$ is discrete, and the energy distance is a proper metric. This positivity then guarantees strict convexity. As a result, to minimize $D^2(\mathbb{P}_v, P_n)$, we require 
$$
 \left( \lambda -\frac{\Psi_{A,v}-\Psi_{A,n}-\Psi_{B,v}+\Psi_{B,n}}{2\Psi_{B,n}-\Psi_{A,n}-\Psi_{C,n}}\right)^2  = 0
$$
which gives the result if the solution of the above lies in the range $[0,1]$. If the solution lies outside $[0,1]$, since the energy distance (\ref{eq:DLambda1}) is quadratic in $\lambda$, the minimum on $\lambda \in [0,1]$ will be the closest boundary point to the global minimum, which is attained by clipping appropriately as in (\ref{eq:app_opt_lam}).
\end{proof}

\subsubsection{Theoretical and cross-validated energy distance}
 Proposition \ref{prop:d2_quad} can be straightforwardly extended to the case where $D^2(\mathbb{P}_v,P_n)$ is replaced with the theoretical quantity $D^2(\mathbb{Q},P_n)$, which would then imply the the second part of Proposition \ref{prop:concave} regarding the theoretical score (\ref{eq:opt_lam}) and the optimal $\lambda_n$. For completion, we also consider the cross-validated estimate
 \begin{align*}
   \widehat{D}^2(\mathbb{Q},P_n) = \frac{1}{J}\sum_{j = 1}^J D^2\left(\mathbb{P}_{v}^{(j)},P_{n}^{(j)} \right)
 \end{align*}
 where we consider $J$ splits of the $(n+v)$ data points into $v$ validation and $n$ training points (sticking to the Appendix notation). In the above $\mathbb{P}_{v}^{(j)}$ corresponds to the empirical distribution of the $j$-th validation set, and $P_{n}^{(j)}$ is the mixture predictive fitted to the $j$-th training set. 

 We state the above two extensions formally below.
 \begin{corollary}\label{cor:lambda_cv}
Proposition \ref{prop:d2_quad} holds if  $D^2(\mathbb{P}_v,P_n)$ is replaced with $D^2(\mathbb{Q},P_n)$, and $\Psi_{A,v}$ and $\Psi_{B,v}$ are replaced with $\Psi_A$ and $\Psi_B$ respectively, where
\begin{align*}
    \Psi_{A} &= \frac{1}{n}\sum_{i = 1}^n \int |y - x_i| \, q(y)\, dy,\\
    \Psi_B &= \iint |y - x| \, q(y)\, f_{\theta_n}(x)\, dy\, dx.
\end{align*}
Similarly, Proposition \ref{prop:d2_quad} holds if  $D^2(\mathbb{P}_v,P_n)$ is replaced with $ \widehat{D}^2(\mathbb{Q},P_n)$, and all quantities $\Psi_{\cdot,\cdot}$ are replaced with their respective averages over $J$ data splits.
 \end{corollary}
 \begin{proof}
     The first part holds trivially by exchanging $\mathbb{P}_v$ with $\mathbb{Q}$. For the second part, we note that the objective $D^2(\mathbb{P}_v,P_n)$ is linear in $\Psi_{\cdot, \cdot}$ so averaging over $J$ folds just results in a quadratic equation in $\lambda$ with averaged coefficients.
 \end{proof}

\subsection{Proof of Theorem \ref{theo:optimal}}
To begin, we restate the term of interest:
$$ \widehat{\lambda}_n=\min\left\{1,\max\left\{0,\dfrac{\Psi_{A,v} -\Psi_{A,n}-\Psi_{B,v}+ \Psi_{B,n}}{2\Psi_{B,n}-\Psi_{A,n}-\Psi_{C,n}}\right\}\right\}$$
which we show converges to 0 under model misspecification. The proof technique is relatively simple, although we will require a few technical steps. The main idea is that under model misspecification, the denominator of the inner term satisfies
\begin{align*}
 2\Psi_{B,n}-\Psi_{A,n}-\Psi_{C,n} \overset{p}{\to} 2B  -A -C = D^2\left(\mathbb{Q}, F_{\theta^*}\right) > 0
\end{align*}
where we again use the fact that the energy distance is a proper metric, and $\mathbb{Q} \neq F_{\theta^*}$ by assumption.
On the other hand, the numerator of the inner term satisfies
\begin{align*}
    \Psi_{A,v} -\Psi_{A,n}-\Psi_{B,v}+ \Psi_{B,n} \overset{p}{\to} A - A - B + B = 0.
\end{align*}
Since the clipping function is continuous (from the continuity of $\min$ and $\max$), we have $\widehat{\lambda}_n \overset{p}{\to} 0$ given the above. The required assumptions are exactly those needed to obtain the above convergence.

\subsubsection{Technical lemmas}
\begin{lemma} 
Assuming $E_{\mathbb{Q}}[Y^2]<\infty$ where $Y \sim \mathbb{Q}$, we have $\Psi_{A,n},\Psi_{A,v}\overset{p}\to A$.\label{lem:psi2,5}
\end{lemma}
\begin{proof}

Firstly, when $E_{\mathbb{Q}}(Y^2)<\infty$, we have $$A=E_{Y,Y'\iid \mathbb{Q}}\left[|Y-Y'|\right]\leq 2E_{\mathbb{Q}}\left[|Y|\right]<\infty,$$  
so the limit $A$ is finite by assumption.

\paragraph*{Term 1} For $\Psi_{A,n}$,  the double sum includes terms where \( i = j \), so we can write
  \[
  \sum_{i=1}^n \sum_{j=1}^n |x_i - x_j| =  \sum_{i \neq j} |x_i - x_j| = 2 \sum_{i < j} |x_i - x_j|.
  \]
Thus,  we have
  \[
  \Psi_{A,n}= \frac{2}{n^2} \sum_{i < j} |x_i - x_j| .
  \]
For the proof, we relate \( \Psi_{A,n} \) to a U-statistic. Define 
$$U = \frac{2}{n(n-1)} \sum_{i<j} |x_i - x_j|, $$
where 
$$E_{\mathbb{Q}}[U] = E_{Y,Y'\iid \mathbb{Q}}[|Y - Y'|] = A.$$ Then,  
\[
\Psi_{A,n}  = \frac{n-1}{n} U,
\]
so we only need to study the consistency of the U-statistic $U$, which we include for completion. For a U-statistic with kernel \( g(x_i, x_j) = |x_i - x_j| \) (degree 2), the variance is:  
\[
\text{Var}_{\mathbb{Q}}(U) = \frac{4(n-2)}{n(n-1)} \zeta_1 + \frac{2}{n(n-1)} \zeta_2,
\]
where 
\begin{align*}
\zeta_1 &=
 \text{Var}_{Y_1\sim \mathbb{Q}}(E_{Y'\sim \mathbb{Q}}[|Y - Y'| | Y]) = \text{Var}_{\mathbb{Q}}(h(Y)),\\
 \zeta_2 &= 
 \text{Var}_{Y,Y'\iid \mathbb{Q}}(|Y - Y'|).  
\end{align*}
Thus,  
\begin{align*}
\text{Var}_{\mathbb{Q}}(U) = \frac{4(n-2)}{n(n-1)} \text{Var}_{\mathbb{Q}}(h(Y)) + \frac{2}{n(n-1)} \text{Var}_{Y,Y'\iid \mathbb{Q}}(|Y - Y'|)
\end{align*}
and
$$
h(y)=\int |y'-y|\,q(y')\, dy'.
$$
With the assumption that $E_{\mathbb{Q}}[Y^2]<\infty$, we have $$\text{Var}_{Y,Y'\iid \mathbb{Q}}(|Y-Y'|)\leq2 E_{\mathbb{Q}}\left[Y^2\right]<\infty.$$
Also, note that $h(Y)=E_{Y'\sim \mathbb{Q}}[|Y'-Y|\,|Y]$, 
\begin{flalign*}
    E_{Y\sim \mathbb{Q}}\left[h(Y)^2\right] &= E_{Y\sim \mathbb{Q}}\left[E_{Y'\sim \mathbb{Q}}[|Y'-Y|\,|Y]^2\right]\\
    &\leq E_{Y,Y'\iid \mathbb{Q}}\left[(Y'-Y)^2\right]\leq 2E_{\mathbb{Q}}[Y^2] <\infty
\end{flalign*}
so $\text{Var}_{\mathbb{Q}}(h(Y)) < \infty$. Therefore, we have $\text{Var}_{\mathbb{Q}}[U]\to 0$, so by Markov's inequality, we have $U\overset{p}\to A$, which in turn gives $\Psi_{A,n} \overset{p}\to A$.

\paragraph*{Term 2} For $\Psi_{A,v}$, since \( y_i \) and \(x_i\) are independent, we have that $\Psi_{A,v}$ is a 2-sample U-statistic with degree 1 in each sample. The mean is 
  \[
  E_{\mathbb{Q}}[\Psi_{A,v}] = \frac{1}{nv} \sum_{i = 1}^n \sum_{j = 1}^v E_{Y_j,X_j \iid\mathbb{Q}}\left[|Y_j - X_i|\right] = A.
  \]
  The variance of $\Psi_{A,v}$ in this case (e.g. \citet[Theorem 12.6]{van2000}) is 
\begin{align*}
\text{Var}_{\mathbb{Q}}(\Psi_{A,v}) &= 
\frac{1}{nv} \sum_{c = 0}^1 \sum_{d = 0}^1 {{n-1}\choose {1 - c}}  {{v - 1} \choose {1 - d}} \zeta_{cd}
\end{align*}
where
\begin{align*}
    \zeta_{00} =0, &\quad  \zeta_{11}= \text{Var}_{X,Y \iid \mathbb{Q}}\left(|X - Y|\right), \quad \zeta_{01} = \zeta_{10} = \text{Cov}_{{X,X',Y \iid \mathbb{Q}}}\left(|X - Y|, |X' - Y|\right).
\end{align*}
Plugging this in gives
\begin{align*}
    \text{Var}_{\mathbb{Q}}(\Psi_{A,v}) &= \frac{1}{nv}\left[\text{Var}_{X,Y \iid \mathbb{Q}}\left(|X - Y|\right) + (n + v-2)  \text{Cov}_{{X,X',Y \iid \mathbb{Q}}}\left(|X - Y|, |X' - Y|\right)\right]
\end{align*}
We can show that
\begin{align*}
    \text{Cov}_{{X,X',Y \iid \mathbb{Q}}}\left(|X - Y|, |X' - Y|\right) &= E_{{X,X',Y \iid \mathbb{Q}}}\left[|X - Y||X' - Y|\right] - E_{X,Y \iid \mathbb{Q}}\left[|X - Y|\right]^2\\
    &\leq E_{{X,Y \iid \mathbb{Q}}}\left[(X - Y)^2\right]- E_{X,Y \iid \mathbb{Q}}\left[|X - Y|\right]^2\\
    &= \text{Var}_{X,Y \iid \mathbb{Q}}\left(|X - Y|\right)
\end{align*}
where the second line follows from Cauchy-Schwarz. This gives
\begin{align*}
    \text{Var}_{\mathbb{Q}}(\Psi_{A,v}) &\leq \frac{n+v-1}{nv}\, \text{Var}_{X,Y \iid \mathbb{Q}}\left(|X - Y|\right) 
\end{align*}
From the previous section, we have
$$
\text{Var}_{X,Y \iid \mathbb{Q}}\left(|X - Y|\right) \leq 2E_{\mathbb{Q}}[Y^2]< \infty
$$ 
which gives  $\text{Var}_{\mathbb{Q}}(\Psi_{A,v})\to0$ as $n,v \to 0$. By Markov's inequality, we then have $\Psi_{A,v}\overset{p}{\to}A$.

\end{proof}
\begin{lemma}
Assume that there exists a neighborhood $U$ of $\theta^*$ such that $\sup_{\theta \in U}\mu_\theta^{(2)} < \infty$. Then we have
$\Psi_{C,n}\overset{p}{\to} C.$ \label{lem:psi3}
\end{lemma}
\begin{proof}
First, we note that $C \leq 2E_{f_{\theta^*}}[|X|]<\infty$ from the bounded variance assumption. Next, if we can show that for any deterministic sequence $\theta_n$ converging to $\theta^*$, it follows that $\Psi_{C,n}(\theta_n)\to\Psi_{C,n}(\theta^*)=C$, then $\Psi_{C,n}(\theta)$ is continuous at the point $\theta^*$. We can thus apply continuous mapping theorem to show that $\Psi_{C,n}\overset{p}\to C$. 

Now we consider a sequence $\theta_n\to\theta^*$ deterministically to show continuity of $\Psi_{c,n}(\theta)$. Consider a sequence $X_n \sim f_{\theta_n}$ for $n \geq 1$ and $X \sim f_{\theta}^*$. 
We begin by showing a few intermediate results that will be useful in the proof. Note that the assumption implies that $\{X_n\}_{n \geq 1}$ is uniformly integrable as $X_n$ is bounded in $L_2$, and $E_{f_{\theta}^*}[|X|]<\infty$. As a result, for any $\epsilon > 0$, we can choose $K$ sufficiently large so that
\begin{align}
    \int_{|x|> K}|x|\,|f_{\theta_n}(x)-f_{\theta^*}(x)|dx\leq \int_{|x|> K}|x|\,f_{\theta_n}(x)dx+\int_{|x|> K}|x|\,f_{\theta^*}(x)dx <\frac{\epsilon}{2}\cdot \label{eq:ui1}
\end{align}
Next, choose $n$ sufficiently large so that
\begin{align}
    \left\|f_{\theta_n}-f_{\theta^*}\right\|_{L_1} \leq \frac{\epsilon}{2K} \label{eq:L1}
\end{align}
which is possible as $f_{\theta_n}(x) \to f_{\theta^*}(x)$ for each $x$ (from continuity), and Scheff\'{e}'s lemma gives
$$
\|f_{\theta_n}- f_{\theta^*}\|_{L_1} \to 0.
$$
Now we prove that $\Psi_{C,n}\to C$ deterministically by considering the absolute difference 
$$|\Psi_{C,n} -C|=\Big|\iint|x-x'|\left[f_{\theta_n}(x)f_{\theta_n}(x')-f_{\theta^*}(x)f_{\theta^*}(x')\right]dx\,dx'\Big|.$$
We rewrite the term 
$$f_{\theta_n}(x)f_{\theta_n}(x')-f_{\theta^*}(x)f_{\theta^*}(x')=f_{\theta_n}(x)\left[f_{\theta_n}(x')-f_{\theta^*}(x')\right]+f_{\theta^*}(x')\left[f_{\theta_n}(x)-f_{\theta^*}(x)\right].$$
Thus, 
\begin{align*}
|\Psi_{C,n}-C|&\leq \underbrace{\iint|x-x'|f_{\theta_n}(x)\left|f_{\theta_n}(x')-f_{\theta^*}(x')\right|dx\,dx'}_{T_1} \\
&+ \underbrace{\iint|x-x'|f_{\theta^*}(x')\left|f_{\theta_n}(x)-f_{\theta^*}(x)\right|dx\,dx'}_{T_2}.
\end{align*}
For $T_1$, we have 
$$T_1=\int\left(\int|x-x'|f_{\theta_n}(x)dx\right)\left|f_{\theta_n}(x')-f_{\theta^*}(x')\right|dx'.$$
The inner integral is $E_{f_{\theta_n}}[|X_n-x'|]$. Note that $|X_n-x'|\leq |X_n|+|x'|$, so $E_{f_{\theta_n}}[|X_n-x'|]\leq E_{f_{\theta_n}}[|X_n|]+|x'|$. Thus, we have 
\begin{flalign*}T_1&\leq \int\left(E_{f_{\theta_n}}[|X_n|]+|x'|\right)\left|f_{\theta_n}(x')-f_{\theta^*}(x')\right|dx'\\
&=E_{f_{\theta_n}}[|X_n|]\cdot \|f_{\theta_n}-f_{\theta^*}\|_{L_1} + \int|x'|\left|f_{\theta_n}(x')-f_{\theta^*}(x')\right|dx'.
\end{flalign*}
The first term converges to 0 as $\sup_{n} E_{f_{\theta_n}}|X_n| < \infty$ from uniform integrability, and $\|f_{\theta_n}-f_{\theta^*}\|_{L_1} \to 0$ from (\ref{eq:L1}). We will handle the second term shortly.

Similarly, we have 
\begin{flalign*}T_2&\leq E_{f_{\theta^*}}|X|\cdot\|f_{\theta_n}-f_{\theta^*}\|_{L_1} + \int|x|\left|f_{\theta_n}(x)-f_{\theta^*}(x)\right|dx.
\end{flalign*}
Again the first term converges to 0.

Now consider the common second term in $T_1$ and $T_2$, which we can decompose as
$$\int_{|x|\leq K}|x|\,|f_{\theta_n}(x)-f_{\theta^*}(x)|dx+\int_{|x|> K}|x|\,|f_{\theta_n}(x)-f_{\theta^*}(x)|dx.$$
For the first integral we have 
$$\int_{|x|\leq K}|x|\,|f_{\theta_n}(x)-f_{\theta^*}(x)|dx\leq K \int_{|x|\leq K}\,|f_{\theta_n}(x)-f_{\theta^*}(x)|dx \overset{\textnormal{(\ref{eq:L1})}}\leq K  \frac{\epsilon}{2K}=\frac{\epsilon}{2},$$
while for the second integral it is bounded by $\epsilon/2$ by (\ref{eq:ui1}), this means there exist a sufficiently large $n$ such that 
$$\int|x|\,|f_{\theta_n}(x)-f_{\theta^*}(x)|dx \leq \epsilon.$$

As a result, both $T_1$ and $T_2$ converge to 0, which is sufficient to imply $\Psi_{C,n}\to C$ for any deterministic sequence of $\theta_n\to \theta^*$.

\end{proof}
\begin{lemma}
Assume $E_{\mathbb{Q}}[Y^2]<\infty$ and there exists a neighborhood $U$ of $\theta^*$ such that $\sup_{\theta \in U} \mu_\theta^{(2)}< \infty$. Then we have that $\Psi_{B,n},\Psi_{B,v}\overset{p}{\to} B$.\label{lem:psi1,4}
\end{lemma}
\begin{proof}
Recall that  
\begin{align*}
    \Psi_{B,v} = \frac{1}{v}\sum_{j = 1}^v \int  |y_i-x|  f_{\theta_n}(x) \,dx,\quad \Psi_{B,n} = \frac{1}{n}\sum_{i = 1}^n \int \left| x' - x_i\right|f_{\theta_n}(x')\,dx'.
\end{align*}
and $B = \iint \left|x - y \right| q(y) f_{\theta^*}(x) \, dx\,dy$.
With the assumptions, we have $$B\leq E_{\mathbb{Q}}|Y|+E_{f_{\theta^*}}|X|<\infty.$$

\paragraph*{Term 1} First, we show that $\Psi_{B,n}\overset{p}{\to}B$. Consider the alternative term
\begin{align*}
    \Psi^*_{B,n} = \frac{1}{n}\sum_{i = 1}^n\underbrace{\int \left|x' - x_i \right|f_{\theta^*}(x')\, dx'}_{k(x_i)}.
\end{align*}
 Since  \( x_i \) are i.i.d., and \( E_{\mathbb{Q}}[k(Y)] = B \), we have that
  \[
  E_{\mathbb{Q}}[\Psi^*_{B,n}] = E_{\mathbb{Q}}\left[ \frac{1}{n} \sum_{i=1}^n k(x_i) \right]  =   E_{\mathbb{Q}}\left[k(Y)\right] = B .
  \]
  Similarly, the variance of $\Psi^*_{B,n}$ is
\[
\text{Var}_{\mathbb{Q}}(\Psi^*_{B,n}) = \text{Var}_{\mathbb{Q}}\left( \frac{1}{n} \sum_{i=1}^n k(x_i) \right) = \frac{1}{n} \text{Var}_{\mathbb{Q}}(k(Y)).
\]

Note that $k(Y)=E_{X \sim f_{\theta^*}}[|X-Y|\mid Y]$, so we have 
\begin{flalign*} E_{\mathbb{Q}}[k(Y)^2]&=E_{Y \sim \mathbb{Q}}\left[E_{X \sim f_{\theta^*}}[|X-Y|\mid Y]^2\right]\\
&\leq E_{Y\sim \mathbb{Q}}\left[E_{X\sim f_{\theta^*}}[(X-Y)^2\mid Y]\right]=E_{Y\sim \mathbb{Q}, X\sim f_{\theta^*}}[(X-Y)^2]<\infty.
\end{flalign*}
The last inequality holds under the assumptions $E_{\mathbb{Q}}(Y^2),E_{f_{\theta^*}}(X^2)<\infty$, as we have
$$
E_{Y \sim \mathbb{Q}, X\sim f_{\theta^*}}[(X-Y)^2] \leq E_{f_{\theta^*}}[X^2] + E_{\mathbb{Q}}[Y^2] + 2\sqrt{E_{f_{\theta^*}}[X^2]\, E_{\mathbb{Q}}[Y^2]}
$$
Therefore,  $\text{Var}_{\mathbb{Q}}(k(Y))<\infty$ and $\text{Var}_{\mathbb{Q}}(\Psi^*_{B,n})\to0$. By Markov's inequality, we have $\Psi^*_{B,n}\overset{p}{\to}B$. Now we compute
\begin{align*}
    \left|\Psi_{B,n} - \Psi^*_{B,n}\right|  &= \left|\frac{1}{n}\sum_{i = 1}^n \int \left|x'- x_i\right|\left(f_{\theta^*}(x') - f_{\theta_n}(x')\right)\, dx'\right|\\
    &\leq \frac{1}{n}\sum_{i = 1}^n \int \left|x'- x_i\right|\left|f_{\theta^*}(x') - f_{\theta_n}(x')\right|\, dx'\\
    &\leq \int \left|x\right|\left|f_{\theta^*}(x) - f_{\theta_n}(x)\right|\, dx + \left[\frac{1}{n}\sum_{i = 1}^n \left|x_i\right|\right] \int \left|f_{\theta^*}(x) - f_{\theta_n}(x) \right|\,dx
\end{align*}
The first term converges to 0 in probability, as shown in Lemma \ref{lem:psi3}. The second term converges to 0 in probability as the empirical sum satisfies $n^{-1}\sum_{i = 1}^n |x_i| \overset{p}{\to} E_{\mathbb{Q}}\left[|Y|\right]$ and $\left\|f_{\theta^*} -  f_{\theta_n}\right\|_{L_1} \overset{p}{\to}0$, where the latter result is also shown in Lemma \ref{lem:psi3}.

\paragraph*{Term 2} Next, we show that $\Psi_{B,v} \overset{p}{\to} B$. Fortunately, we can simply repeat the same proof as for $\Psi_{B,n}$, where we replace $n$ with $v$, and $x_i$ with $y_i$. 
\end{proof}
 
\setcounter{prop}{0}
\setcounter{corollary}{0}
\setcounter{theorem}{0}
\setcounter{lemma}{0}
\setcounter{algorithm}{0}
\setcounter{equation}{0}
\setcounter{figure}{0}
\setcounter{table}{0}
\setcounter{example}{0}
\section{Appendix: Regression}\label{sec:app_reg}
            \subsection{Proof of Theorem \ref{theo:reg}}\label{sec:app_proofreg}
            \subsubsection{Martingale condition}
First, we show that the cross moment $\mu_i^{yx}$ is a martingale:
\begin{align*}
E[\mu^{yx}_{i+1} \mid \mathcal{F}_i]
= \frac{i}{i+1}\mu_{i}^{yx} + \frac{1}{i+1}E\left[ Y_{i+1} X_{i+1} \mid \mathcal{F}_i\right]
\end{align*}
where
\begin{align*}
    E[ Y_{i+1} X_{i+1}\mid \mathcal{F}_i]&\nonumber= E[X_{i+1} \, E[ Y_{i+1} \mid \mathcal{F}_i, X_{i+1}]\mid \mathcal{F}_i].
\end{align*}
The inner expectation can be split into the parametric and nonparametric components:
\begin{align*}
    E[ Y_{i+1} \mid \mathcal{F}_i, X_{i+1}] = \frac{c}{c+i}\, g^{-1}\left(X^\top _{i+1}\beta_i\right) + \frac{i}{c+i} \frac{\sum_{j = 1}^i \mathbf{1}\left(x_j = X_{i+1} \right) y_j}{\sum_{j = 1}^i \mathbf{1}\left(x_j = X_{i+1}\right)} 
\end{align*}
For the parametric component, we note that
\begin{align*}
    E\left[X_{i+1} \,g^{-1}\left(X_{i+1}^\top \beta_i\right) \mathcal{F}_i\right]&=  \frac{1}{i}\sum_{j= 1}^i g^{-1}\left(x_j^\top \beta_i \right) \, x_j = \mu_i^{yx}
\end{align*}
where the first equality follows from the BB, and the second follows exactly from the cross-moment matching property.

The nonparametric component requires more care:
\begin{align*}
    E\left[X_{i+1}\frac{\sum_{j = 1}^i \mathbf{1}\left(x_j = X_{i+1} \right) y_j}{\sum_{l = 1}^i \mathbf{1}\left(x_l = X_{i+1}\right)}\mid \mathcal{F}_i \right] &= \frac{1}{i}\sum_{k = 1}^i \sum_{j = 1}^i  \frac{\mathbf{1}\left(x_j = x_k \right) y_j \, x_k}{\sum_{l = 1}^i \mathbf{1}\left(x_l =x_k\right)}\\
    &= \frac{1}{i}\sum_{k = 1}^i \sum_{j = 1}^i \frac{ \mathbf{1}\left(x_j = x_k \right) y_j \, x_j}{\sum_{l = 1}^i \mathbf{1}\left(x_l =x_j\right)}\\
    &= \frac{1}{i}\sum_{j = 1}^i y_j \, x_j \frac{ \sum_{k = 1}^i\mathbf{1}\left(x_j = x_k \right) }{\sum_{l = 1}^i \mathbf{1}\left(x_l =x_j\right)}\\
    &=  \mu_i^{yx}
\end{align*}
where the second line relies on the identity
$$\frac{\mathbf{1}\left(x_j = x_k \right) y_j  \, x_k}{\sum_{l = 1}^i \mathbf{1}(x_l = x_k)} = \frac{\mathbf{1}\left(x_j = x_k \right) y_j \, x_j}{\sum_{l = 1}^i \mathbf{1}(x_l = x_j)}$$
which follows from the indicator function allowing us to swap indices around. We highlight that the nonparametric component also depends closely on the fact that $X_{i+1}$ is drawn from the BB. 

As a result, we have
$
E[Y_{i+1} X_{i+1} \mid \mathcal{F}_i] = \mu_i^{yx},
$
so $\mu_i^{yx}$ is indeed a martingale.

\paragraph*{Bounded in $L_1$ ($Y \geq 0$)}
We now want to verify that $\mu_i^{yx}$ is bounded in $L_1$ for convergence if $Y \geq 0$. We can upper bound
$$
\left\|\mu_{i}^{yx}\right\|_1 \leq  \frac{1}{i}\sum_{j = 1}^i \|x_j y_j\|_1 \leq  \frac{1}{i}\sum_{j = 1}^i |y_j| \|x_j \|_1 =  \frac{1}{i}\sum_{j = 1}^i y_j \|x_j \|_1
$$
where we have used the positivity of $y_j$.
Next, since $x_{1:i}$ only contains unique values $x_{1:n}$, there exists a finite constant $0 < K < \infty$ such that $\sup_j \|x_j \|_1 < K$, we have 
$$
\left\|\mu_{i}^{yx}\right\|_1 \leq K \mu_i^{y}  
$$
where 
$$
\mu_i^{y} = \frac{1}{i}\sum_{j = 1}^i y_j.
$$
However, we note that by assumption, $x_j$ contains an intercept term, so we have
\begin{align*}
    \frac{1}{i}\sum_{j = 1}^i g^{-1}\left(x_j^\top \beta_i\right) = \frac{1}{i}\sum_{j = 1}^i y_i.
\end{align*}
Following the previous section, $\mu_i^y$ is a non-negative martingale, so $E\left[\mu_i \mid \mathcal{F}_n\right] = \mu_n$ and we immediately have that $\mu_i^y$ is bounded in $L_1$, which gives $\mu_i^{yx}$ bounded in $L_1$.

\paragraph*{Bounded in $L_2$ ($Y \in \mathbb{R}$)}
For the second case with a real dependent variable, it is easier to show $\mu_i^{yx}$ is bounded in $L_2$, which implies it is bounded in $L_1$. First, Cauchy-Schwarz gives us 
\begin{align*}
    \|\mu_{i}^{yx}\|^2 &=  \left|\left|\frac{1}{i}\sum_{j = 1}^i x_j y_j\right|\right|^2 \\
    &\leq \mu_i^{xx}\, \mu_i^{yy}
\end{align*}
where
$$
\mu_i^{xx} = \frac{1}{i}\sum_{j = 1}^i \|x_j\|^2, \quad \mu_i^{yy} = \frac{1}{i}\sum_{j = 1}^i y_j^2.
$$
Next, since $x_{1:i}$ only contains unique values $x_{1:n}$, there exists a finite constant $0 < K < \infty$ such that $\sup_j \|x_j\|^2 < K$, which gives
\begin{align*}
    \|\mu_i^{yx}\|^2 \leq K \mu_i^{yy}
\end{align*}
It thus suffices to show that
$$
\sup_i E\left[ \mu_i^{yy} \mid \mathcal{F}_n \right] < \infty.
$$

Again, the standard recursive approach gives
\begin{align*}
    E\left[\mu_{i+1}^{yy} \mid \mathcal{F}_{i}\right] = \frac{i}{i+1} \mu_{i}^{yy} + \frac{1}{i+1}E\left[Y_{i+1}^2 \mid \mathcal{F}_i\right],
\end{align*}
where
\begin{align*}
    E\left[Y_{i+1}^2 \mid \mathcal{F}_i\right] = E\left[E\left[Y_{i+1}^2 \mid \mathcal{F}_i,X_{i+1}\right] \mid \mathcal{F}_i\right]
\end{align*}
and
\begin{align*}
    E\left[Y_{i+1}^2 \mid \mathcal{F}_i,X_{i+1}\right] = \frac{c}{c+i} E_{\beta_i}\left[Y_{i+1}^2 \mid X_{i+1}\right] +\frac{i}{c+i} \frac{\sum_{j = 1}^i \mathbf{1}\left(x_j = X_{i+1} \right) y_j^2}{\sum_{j = 1}^i \mathbf{1}\left(x_j = X_{i+1}\right)} 
\end{align*}
The same argument as before can be applied to the nonparametric component to show that
$$
E\left[\frac{\sum_{j = 1}^i \mathbf{1}\left(x_j = X_{i+1} \right) y_j^2}{\sum_{j = 1}^i \mathbf{1}\left(x_j = X_{i+1}\right)}\mid \mathcal{F}_i\right] = \mu_i^{yy}. 
 $$
The difficulty then lies in upper bounding the parametric component
\begin{align*}
E\left[E_{\beta_i}\left[Y_{i+1}^2 \mid X_{i+1}\right] \mid \mathcal{F}_i\right] = \frac{1}{i}\sum_{j = 1}^i E_{\beta_i}\left[Y^2 \mid  x_j\right]
\end{align*}
As $g$ is the identity, and from the assumption on $V(\mu)$, the conditional second moment in this case satisfies
    \begin{align*}
    E_{\beta}[Y^2 \mid x] &= \phi V(x^\top \beta) + (x^\top \beta)^2\\
    &\leq a'(x^\top \beta)^2 + b'.
    \end{align*}
     This then gives
    \begin{align*}
         n^{-1}\sum_{j = 1}^n E_{\beta_n}\left[Y^2 \mid  x_j\right] \leq  a'n^{-1} \beta_n^\top \mathbf{X}_n^\top \mathbf{X}_n\beta_n + b'
    \end{align*}
    where $a',b'$ are non-negative constants. Again as $g$ is the identity function, (\ref{eq:cross_moment}) gives the least square estimate 
    $\beta_n = \left(\mathbf{X}_n^\top \mathbf{X}_n\right)^{-1}\mathbf{X}_n^\top \mathbf{Y}_n$, so we have
    \begin{align*}
        n^{-1} \beta_n^\top \mathbf{X}_n^\top \mathbf{X}_n\beta_n = n^{-1}\mathbf{Y}_n^\top \mathbf{H}_n \mathbf{Y}_n
    \end{align*}
    where $\mathbf{H}_n = \mathbf{X}_n \left(\mathbf{X}_n^\top \mathbf{X}_n\right)^{-1}\mathbf{X}_n^\top$. Since $\mathbf{H}_n$ is a projection matrix, its eigenvalues are either 0 or 1, so a Rayleigh quotient argument gives
    \begin{align*}
        n^{-1}\mathbf{Y}_n^\top \mathbf{H}_n \mathbf{Y}_n \leq  n^{-1}\mathbf{Y}_n^\top \mathbf{Y}_n  = \mu_n^{yy}.
    \end{align*}

Given the above, we would have
\begin{align*}
     E\left[\mu_{i+1}^{yy} \mid \mathcal{F}_{i}\right] &\leq \frac{i}{i+1} \mu_{i}^{yy} + \frac{1}{i+1}\left[\frac{c}{c+i}\left(a'\mu_i^{yy} + b'\right) + \frac{i}{c+i} \mu_i^{yy}\right]\\
     &\leq \left[1 + O(i^{-2})\right]\mu_i^{yy} + O(i^{-2})
\end{align*}
which would give $\sup_i E[\mu_{i}^{yy} \mid \mathcal{F}_n] < \infty$ as desired.

\subsubsection{Continuous mapping theorem}
The remaining requirement is that $\beta_i$ is a continuous function of $\mu_i^{yx}$, which should follow from standard implicit function arguments. More precisely, we shall show that for any $N\geq n$ and any arbitrary $\mu^{yx}_i$ and $w_{1:n}^i$ with $w^i_i>0$, the regression coefficient $\beta_i$ is a continuous function of $\mu^{yx}_i$ and $w_{1:n}^i$.  As a reminder, we assume that $g$ is strictly monotone and continuously differentiable (which is also satisfied for $g$ being the identity). 

Note that $\beta_i$ satisfies the implicit function
\begin{align*}
      \frac{1}{i}\sum_{j = 1}^i g^{-1}(x_j^\top \beta_i) \,x_j = \mu_i^{yx}.
\end{align*}
Since we know $x_i$ only takes on values in $x_{1:n}$, we can rewrite the above as
\begin{align*}
    \sum_{j = 1}^n w_j^i \, g^{-1}(x_j^\top \beta_i)\, x_j = \mu_i^{yx},\quad w_j^i = \frac{1}{i}\sum_{k = 1}^i \mathbf{1}(x_k = x_j).
\end{align*}
We can just treat $x_{1:n}$ as constants as we condition on their values.  Let $$F(\beta, w_{1:n},\mu^{yx})=\sum_{j = 1}^n w_j \, g^{-1}(x_j^\top \beta)\, x_j - \mu^{yx}$$ be a function defined on an open set $ S:=\{\beta, w_{1:n}, \mu^{yx}\}\subseteq \mathbb{R}^{p+(n+1)}$  with values in $\mathbb{R}^p$. By  \citet[Theorem 13.7]{apostol1958mathematical}, if the following two conditions are satisfied
\begin{enumerate}
    \item $F$ have continuous first partial derivatives in $S$; and 
    \item The Jacobian with respect to $\beta$ is invertible at the point that solves $F(\beta,w_{1:n},\mu^{yx})=0$, 
\end{enumerate}
 then the implicit function theorem can be applied. For the first condition, note that
 \begin{flalign*}      
\frac{\partial F}{\partial w_j} = g^{-1}(x_j^\top \beta) \, x_j,  \quad \frac{\partial F}{\partial \beta} = \sum_{j=1}^n w_j \, \frac{x_j \, x_j^\top}{g'\left(g^{-1}(x_j^\top \beta)\right)}, \quad
\frac{\partial F}{\partial \mu^{yx}} &= -I_p,
    \end{flalign*}
    so the partial derivatives are obviously continuous for all $\beta, \mu^{yx}$ and $w_{1:n}$ which satisfies $w_i > 0$. Secondly, the Jacobian matrix is  $$\frac{\partial F}{\partial \beta} = \sum_{j=1}^n w_j \, \frac{x_j \, x_j^\top}{g'\left(g^{-1}(x_j^\top \beta)\right)}.$$
     With the assumption that $\mathbf{X}_n$ is full rank, the Jacobian matrix is invertible for any arbitrary $\beta$, $\mu^{yx}$ and $w_{1:n}$ that satisfies $w_i>0$.

To see that the Jacobian is invertible, we can write the Jacobian matrix as $$ \frac{\partial F}{\partial \beta} = \mathbf{X}_n^\top J \mathbf{X}_n$$
where $J$ is  an $n \times n$ diagonal matrix with  entries $J_j = w_j \left(g'\left(g^{-1}(x_j^\top \beta)\right)\right)^{-1}$. 
Due to strict monotonicity of $g$, we also have $J_j > 0$ or $J_j < 0$ a.s. for increasing/decreasing $g$ respectively. Without loss of generality, we assume that $g$ is increasing, so $J_j > 0$ a.s. 
To check that $ \mathbf{X}_n^\top J \mathbf{X}_n$ is positive definite, we need to verify that for any nonzero vector $v\in\mathbb{R}^p$, we have $v^\top  \mathbf{X}_n^\top J\mathbf{X}_n v>0$. We write 
$$v^\top  \mathbf{X}_n^\top J \mathbf{X}_n v = (\mathbf{X}_n v)^\top J (\mathbf{X}_n v)=: z^\top J z=\sum_{j=1}^nJ_jz_j^2.$$
If $\mathbf{X}_n$ is full rank, $\mathbf{X}_nv\neq 0$ for all non zero $v$, so $z\neq 0$ if $v\neq 0$. This then implies that for any $v\neq 0$, the sum $\sum_{j=1}^nJ_jz_j^2>0$.

As the conditions of implicit function theorem are satisfied for any $\beta$, $\mu^{yx}$ and $w_{1:n}$ with $w_i>0$, then for any $(\beta_0,w^0_{1:n},\mu^{yx}_0)$ that solves the function $$F(\beta_0,w^0_{1:n},\mu^{yx}_0)=0,$$ $\beta_0$ is a continuous function of $w_{1:n}$ and $\mu^{yx}$.  Note that for every $i\geq n$, we have $$F(\beta_i, w^i_{1:n}, \mu^{yx}_i)=0$$ by the score function condition, so $\beta_i$ is a continuous function of $w^i_{1:n}$ and $\mu^{yx}_i$ for any arbitrary $\mu^{yx}_i$ and $w^i_{1:n}$ with $w^i_i>0$. Since $\mu_i^{yx} \to \mu_\infty^{yx}$ a.s., and $w_{1:n}^i \to w^\infty_{1:n}\sim \text{Dir}(1,\ldots,1)$ a.s. with $w^\infty_i>0$ a.s., we have $\beta_i \to \beta_\infty$ a.s. by the continuous mapping theorem.

\subsection{Convergence of conditional mean ($Y \geq 0$)}\label{sec:app_condit}
In this section, we present a convergence result of the empirical conditional expectation when $Y$ is non-negative (e.g. for logistic regression),  defined as 
\begin{align*}
{\mu}_i^{y \mid x} &= \frac{i^{-1}\sum_{j = 1}^{i}\mathbf{1}(x_j = x)\,  y_j}{i^{-1}\sum_{j = 1}^i \mathbf{1}(x_j = x)} = \frac{m_i^{yx}}{w_i^x}
\end{align*}
where $x \in \{x_1,\ldots,x_n\}$
and 
\begin{align*}
    {m}_i^{yx} &= \frac{1}{i}\sum_{j = 1}^i \mathbf{1}(x_j = x) \,y_j,
         \quad w_i^x = \frac{1}{i}\sum_{j = 1}^i \mathbf{1}(x_j = x).
\end{align*}
\begin{prop}\label{prop:cond_mean}
Assume the conditions of Theorem \ref{theo:reg} hold, and $Y$ is non-negative. Then the empirical conditional expectation satisfies $\mu_i^{y \mid x} \to \mu_\infty^{y \mid x}$ a.s.
\end{prop}
\begin{proof}
We note that the denominator $w_i^x$ converges to the Dirichlet weight corresponding to $x$ a.s., so we only need to study the numerator $m_i^{yx}$.
We then have
\begin{align*}
E\left[{m}_{i+1}^{yx}\mid  \mathcal{F}_i\right] = \frac{i}{i+1}  {m}_i^{yx} + \frac{1}{i+1} E\left[\mathbf{1}(X_{i+1} = x)\, Y_{i+1}\,\mid \mathcal{F}_i\right].
\end{align*}
We can then compute
\begin{align*}
    E\left[\mathbf{1}(X_{i+1} = x)\, Y_{i+1}\,\mid \mathcal{F}_i\right] &= E\left[Y_{i+1} \mid X_{i+1} = x, \mathcal{F}_i\right]\, E\left[\mathbf{1}(X_{i+1} = x)\,\mid \mathcal{F}_i\right]\\
    &= 
    \left[\frac{c}{c+i}\, g^{-1}(x^\top  \beta_i) + \frac{i}{c+i} {\mu}_i^{y \mid x}\right]\, w_i^x\\
    &= \frac{c}{c+i}\, g^{-1}(x^\top  \beta_i)\, w_i^x + \frac{i}{c+i} {m}_i^{yx}.
\end{align*}
Together, this gives
\begin{align*}
    E\left[{m}_{i+1}^{yx}\mid  \mathcal{F}_i\right] &= \frac{i}{i+1}  {m}_i^{yx} + \frac{1}{i+1}\left[\frac{i}{c+i}{m}_i^{yx} + \frac{c}{c+i}\, g^{-1}(x^\top \beta_i) w_i^x\right] \\
    &\leq {m}_i^{yx} + \frac{c}{(i+1)(c+i)}\, g^{-1}(x^\top \beta_i)\, w_i^x
\end{align*}
where we use the fact that $m_i^{yx}$ is non-negative since $Y\in \{0,1\}$. As both $w_i^x$ and $\beta_i$ converge a.s., since $g^{-1}$ is continuous, we have that $g^{-1}(x^\top \beta_i)\, w_i^x$ is bounded a.s. which gives the following:
\begin{align*}
    \sum_{i = n}^\infty  \frac{c}{(i+1)(c+i)}\, g^{-1}(x^\top \beta_i)\, w_i^x < \infty \quad \text{a.s.}
\end{align*}
Crucially, ${m}_i^{yx}$ is non-negative so we have that ${m}_i^{yx}$ is an almost supermartingale and thus converges a.s. As a result, we have that $\mu_i^{y \mid x}$ converges a.s.
\end{proof}

\subsection{Practical details}\label{sec:app_regprac}
\subsubsection{Selecting $c$ for regression}\label{sec:app_regc}

In the regression setting, computing the energy score on a test datum requires a bit more care due to the nonparametric component $\mathbbm{P}_i(y \mid x)$, which is undefined if $x$ is not one of the observed atoms $x_{1:i}$. As a result, we opt to consider the energy score of the \textit{joint} model on $(x,y)$ instead. We must however to take care that the contribution of $x$ does not dominate that of $y$ if $p$ is large. Specifically, we consider the joint predictive
\begin{align*}
p_i(y_{i+1},x_{i+1})&= p_i(y_{i+1} \mid x_{i+1}) \, p(x_{i+1}),\\
p(x_{i+1}) &= \frac{1}{i}\sum_{j = 1}^i \delta_{x_j}(x_{i+1})
\end{align*}
where $p_i(y_{i+1} \mid x_{i+1})$ is (\ref{eq:linreg}). Note that this corresponds to the joint predictive we are predictive resampling with in Algorithm \ref{alg:linreg}. Let $\mathbb{P}_{v}$ correspond to the empirical distribution of a held-out test set $\{y_{v},x_{v}\}$. We assume that the covariates are normalized to have mean 0 and variance 1, and the covariates are of dimension $p$. Then our score to maximize is
\begin{align*}
    S(\mathbb{P}_{v}, P_n) &= \frac{1}{v}\sum_{j = 1}^v s(y_{v,j},x_{v,j}; P_n),\\
    s(y,x; P) &= - 2 E_{\{Y,X\} \sim P}\left[\sqrt{(4p)^{-1}\|x - X\|_2^2+ (y - Y)^2} \right] \\
    &+ E_{\{Y,X\}, \{Y',X'\} \iid P}\left[\sqrt{(4p)^{-1}\|X - X'\|_2^2+ (Y - Y')^2}\right]
\end{align*}
where $\|\cdot\|_2$ is the $L_2$ norm. Note that we have divided the contribution of the covariates by ${4p}$ to ensure they do not dominate the contribution to the energy score. The factor $4p$ is chosen as the maximum variance of the binary $Y$ is $1/4$, and we assume $X$ is standardized.
This amounts to essentially tuning the hyperparameters of the kernel in the MMD. The extensions to cross-validation are then obvious, which is what we implement in practice.

\subsubsection{Approximating $\beta_i$}\label{sec:app_mle}

In this section, we consider a variant of Algorithm \ref{alg:linreg} where we utilize an online  update based on Newton's method for computational expediency, that is
$$
\beta_{i+1} \gets \beta_{i} + (i+1)^{-1}\mathcal{I}_n(\beta_n)^{-1}s(y_{i+1}, \beta_{i}; x_{i+1})
$$
where we only require the new datum $\{y_{i+1},x_{i+1}\}$, avoiding a costly reprocessing of the full dataset. Additionally, we use the Fisher information at the initial $\beta_n$ to avoid having to invert the Fisher information matrix during predictive resampling. Note that this corresponds to a single step of Newton's method initialized at $\beta_i$, where we assume $\beta_i$ is the MLE computed from $\{y_j,x_j\}_{j = 1}^i$.

Although $\beta_i$ is no longer exactly the solution to (\ref{eq:cross_moment}), we will shortly demonstrate empirically for a logistic regression example that it is generally close, so we expect the martingale to be almost preserved in practice. The intuition for this is as follows: since $Y_{i+1}\mid X_{i+1}$ is drawn from $f_{\beta_i}(\cdot \mid X_{i+1})$, the optimal $\beta_{i+1}$ cannot deviate too far from $\beta_i$ from which the sample is drawn. As a result, one step of Newton's method is sufficient.

\begin{algorithm}[!h]\label{alg:logreg}
\caption{Predictive resampling for logistic regression }\label{alg:logreg}
\begin{algorithmic}[1]
\State Given the observed data $\{y_i, x_i\}_{i=1}^n$ and the number of posterior samples $B$
\State Compute the initial estimate $\beta_n$ and the Fisher information $\mathcal{I}_n(\beta_n)$
  \For{$i \gets n$ to $N-1$}
        \State  Sample $x_{i+1}$ by the Bayesian bootstrap 
        \State  Sample $y_{i+1}$ by Equation (\ref{eq:linreg}) 
        \State  Update coefficients:
        $\beta_{i+1} \gets \beta_{i} + (i+1)^{-1}\mathcal{I}_n(\beta_n)^{-1}s(y_{i+1}, \beta_{i}; x_{i+1})$
  \EndFor 
  \State {Repeat} Lines 3-7 $B$ times and \textbf{return} $\{\beta^{(1)}_N,...,\beta^{(B)}_N\}$
\end{algorithmic}
\end{algorithm}\vspace{-2mm}
\paragraph*{Empirical comparison} We now compare Algorithms \ref{alg:linreg} and \ref{alg:logreg} in the logistic regression example. We consider the same same setup (with the same data split) as Section \ref{sec:logreg}, but only consider $p = 2$ covariates ($\beta_2$ and $\beta_5$) for computational expediency. For all examples, we run predictive resampling for $N = 4000$ forward steps and generate $B = 2000$ samples, and set $c = 1600$.
For Algorithm \ref{alg:linreg}, we implement the MLE with 3 steps of Newton-Raphson with a full pass of the dataset each time.

Fig. \ref{fig:joint_post} illustrates the joint posterior plots  for $\beta_2, \beta_5$ using Algorithms \ref{alg:linreg} and \ref{alg:logreg}. We see that both the marginal and joint posteriors are practically indistinguishable. However, the exact case (Algorithm \ref{alg:linreg}) required 20 minutes to run, compared to 1.8 seconds for the approximate case (Algorithm \ref{alg:logreg}). This large difference is precisely due to the full pass of the dataset required for the exact case.

To further solidify the validity of Algorithm \ref{alg:logreg}, we consider evaluating the actual values of $\beta_i$ along predictive resampling trajectories. Specifically, we consider two setups:
\paragraph*{Setup 1:} Carry out exact predictive resampling with Algorithm \ref{alg:linreg}, and simultaneously compute the approximate MLE:
   $${\beta}^{\text{Online}}_{i+1} \gets {\beta}_{i}^{\text{Online}} + (i+1)^{-1}\mathcal{I}_n(\beta_n)^{-1}s(y_{i+1}, {\beta}_{i}^{\text{Online}}; x_{i+1})
$$
\paragraph*{Setup 2:} Carry out approximate predictive resampling with Algorithm \ref{alg:logreg}, and simultaneously compute the exact MLE:
   $${\beta}^{\text{MLE}}_{i+1} \gets \text{MLE}(y_{1:i+1}, x_{1:i+1})
$$
Along trajectories in both setups, we then monitor $\|\beta_i^{^{\text{Online}}}\|_2$, $\|{\beta}_i^{\text{MLE}}\|_2$ and $\|\beta_i^{\text{MLE}} - {\beta}_i^{\text{Online}}\|_2$. This setup is required as we want to monitor a.s. differences between the two optimization algorithms (and not just distributional differences).

Fig. \ref{fig:setup1} and \ref{fig:setup2} plots these quantities for Setups 1 and 2 respectively for two illustrative predictive resampling trajectories. We see in both settings that $\|\beta_i^{\text{MLE}} - {\beta}_i^{\text{Online}}\|_2$ is very small relative to $\|\beta_i^{\text{MLE}}\|_2$ and $\|\beta_i^{\text{Online}}\|$, indicating that Algorithms \ref{alg:logreg} is close to Algorithm \ref{alg:linreg} not only in distribution, but on a per-sample a.s. basis as well. The key property to notice is that $\|\beta_i^{\text{MLE}} - {\beta}_i^{\text{Online}}\|_2$ stays small even as the number of forward steps increases, which solidifies our earlier intuition that predictive resampling prevents $\beta_i^{\text{MLE}}$ and $\beta_i^{\text{Online}}$ from deviating too much.
\begin{figure}[!ht]
        \centering
        \includegraphics[width=0.52\linewidth]{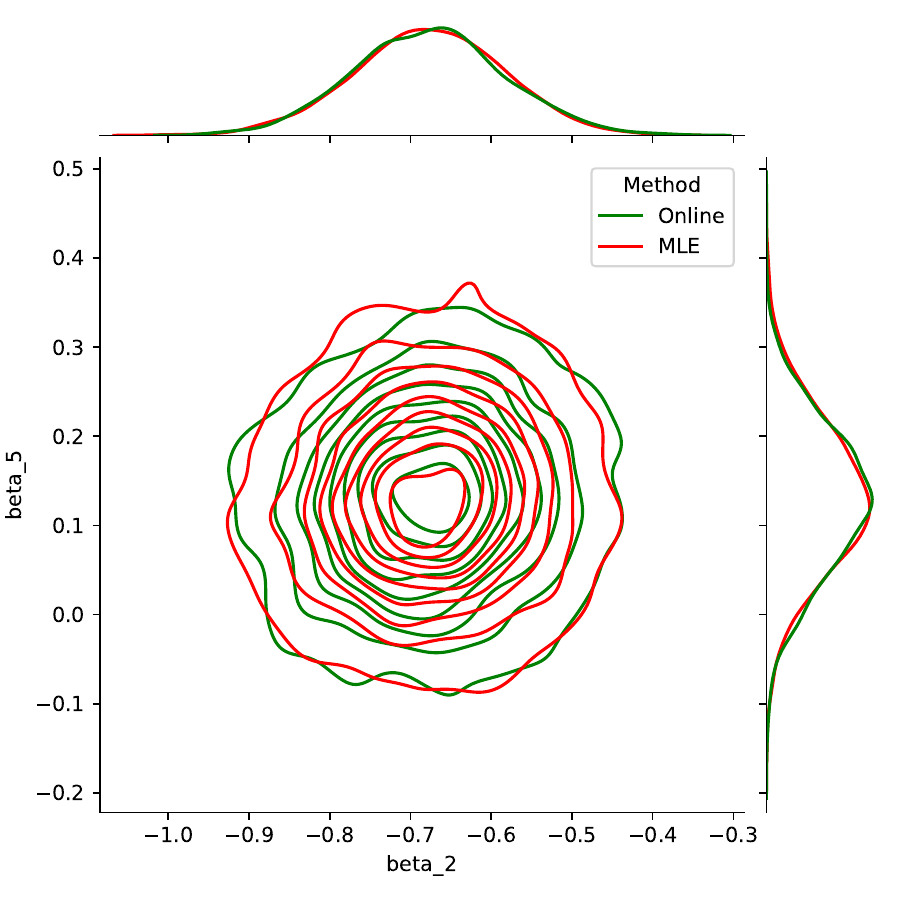} 
        \caption{Joint posterior of $\beta_2, \beta_5$ under Algorithms \ref{alg:linreg} and \ref{alg:logreg}}
        \label{fig:joint_post}
\end{figure}
\begin{figure}[!ht]
    \centering
    \begin{subfigure}[t]{0.49\textwidth}
        \centering
        \includegraphics[width=\linewidth]{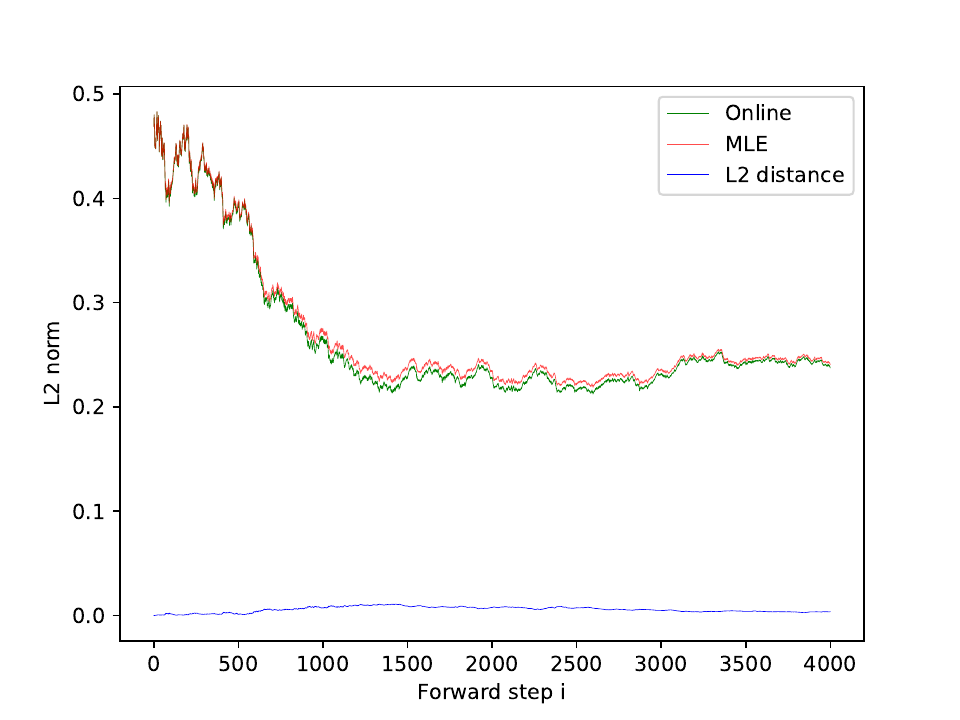} 
    \end{subfigure}
    \begin{subfigure}[t]{0.49\textwidth}
        \centering
        \includegraphics[width=\linewidth]{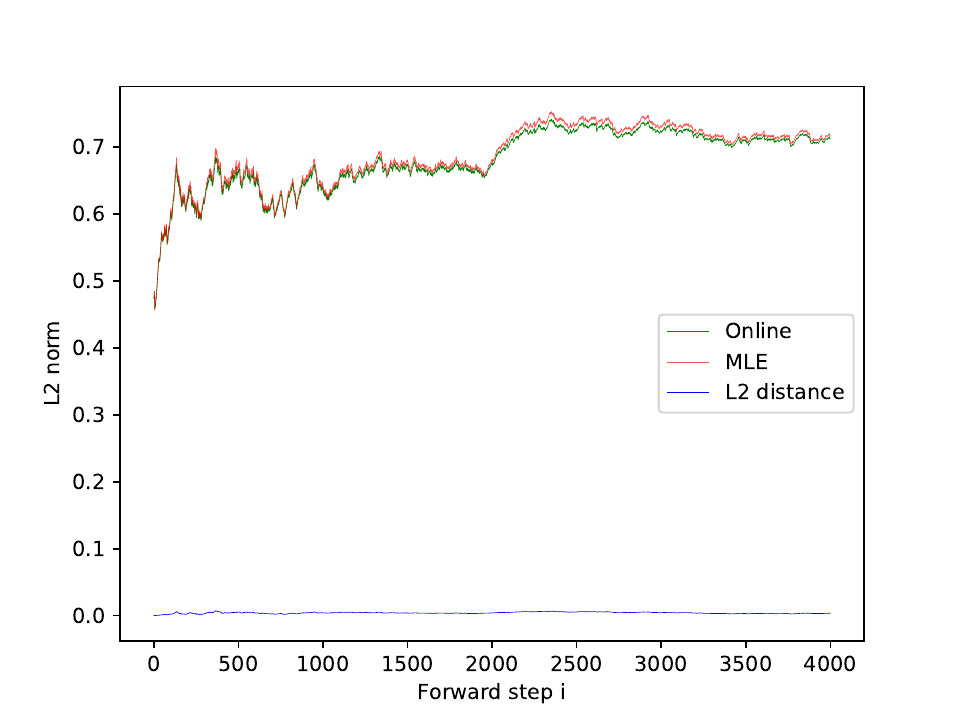} 
    \end{subfigure}
    \caption{Plots of $\|\beta_i^{^{\text{Online}}}\|_2$, $\|{\beta}_i^{\text{MLE}}\|_2$ and $\|\beta_i^{\text{MLE}} - {\beta}_i^{\text{Online}}\|_2$ for Setup 1 for two posterior samples}
    \label{fig:setup1}
\end{figure}\vspace{-5mm}
\begin{figure}[!ht]
    \centering
    \begin{subfigure}[t]{0.49\textwidth}
        \centering
        \includegraphics[width=\linewidth]{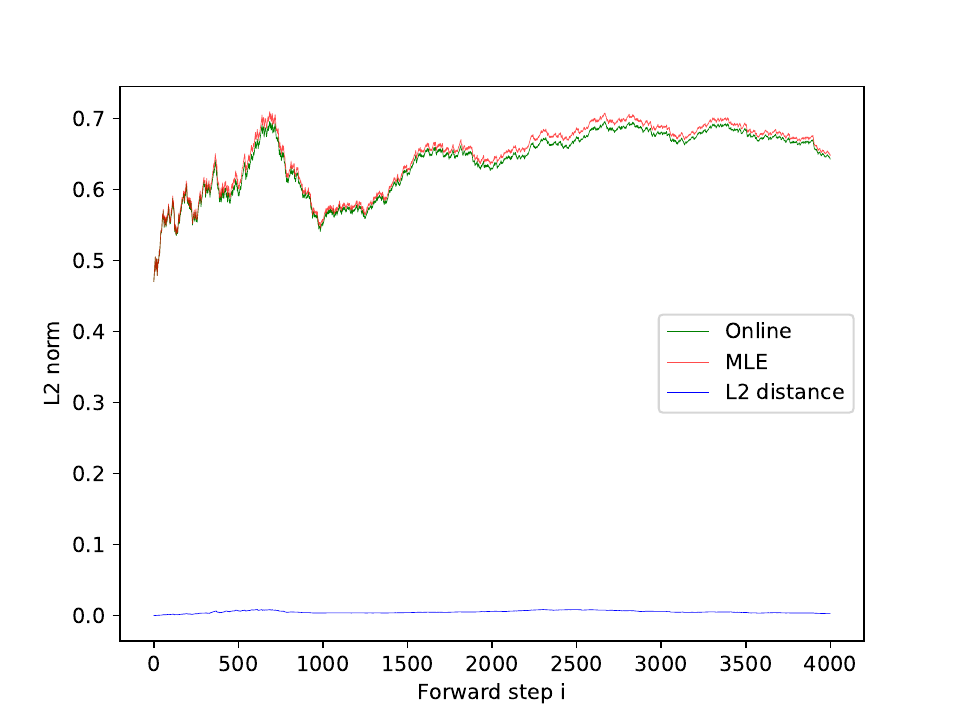} 
    \end{subfigure}
    \begin{subfigure}[t]{0.49\textwidth}
        \centering
        \includegraphics[width=\linewidth]{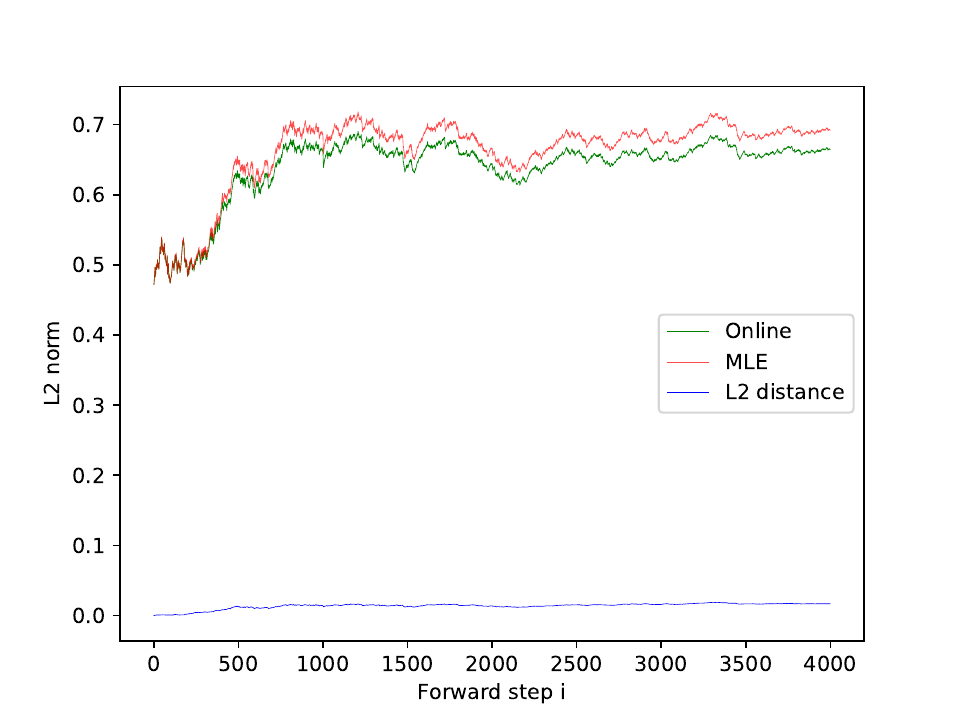} 
    \end{subfigure}
    \caption{Plots of $\|\beta_i^{^{\text{Online}}}\|_2$, $\|{\beta}_i^{\text{MLE}}\|_2$ and $\|\beta_i^{\text{MLE}} - {\beta}_i^{\text{Online}}\|_2$ for Setup 2 for two posterior samples}
    \label{fig:setup2}
\end{figure}

\setcounter{prop}{0}
\setcounter{corollary}{0}
\setcounter{theorem}{0}
\setcounter{lemma}{0}
\setcounter{algorithm}{0}
\setcounter{equation}{0}
\setcounter{figure}{0}
\setcounter{table}{0}
\setcounter{example}{0}
\section{Appendix: Additional experiments}\label{sec:app_exp}
\subsection{Simulation}\label{sec:app_sim}
We present the simulation studies of the well-specified case with small sample size. We select a dataset in which the optimal $c$ is effectively infinity in the well-specified case (Fig. \ref{fig:6} (left)) and another dataset where $c$ remains small (Fig. \ref{fig:6} (right)). 
\begin{figure}[ht]
    \centering
    \begin{subfigure}[t]{0.44\textwidth}
        \centering
        \includegraphics[width=\linewidth]{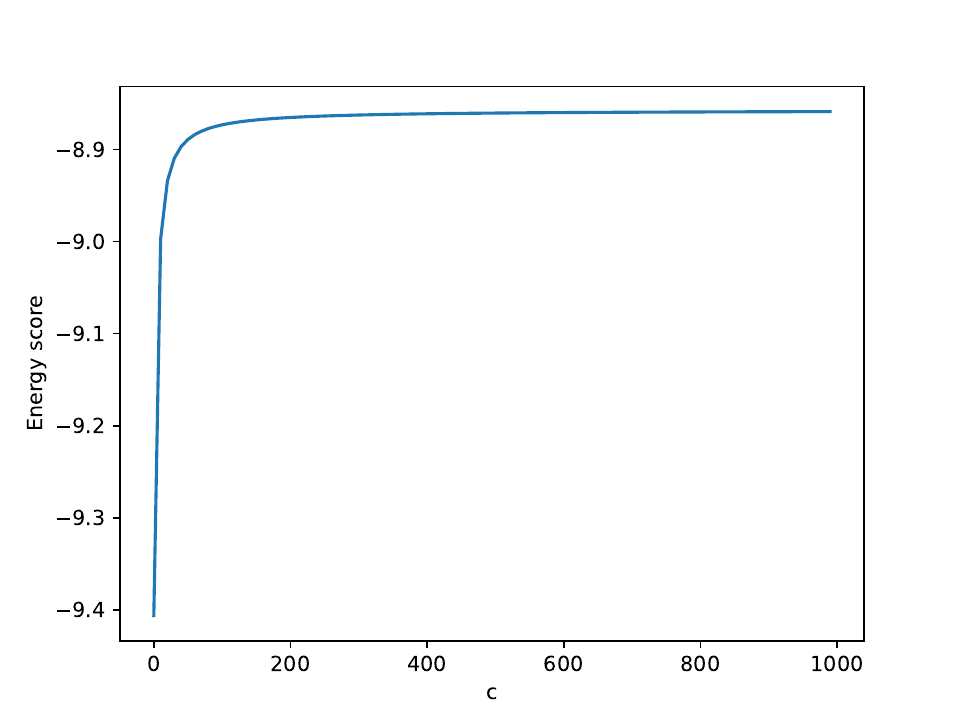}
        \label{fig:graph8}
    \end{subfigure}
        \begin{subfigure}[t]{0.44\textwidth}
        \centering
        \includegraphics[width=\linewidth]{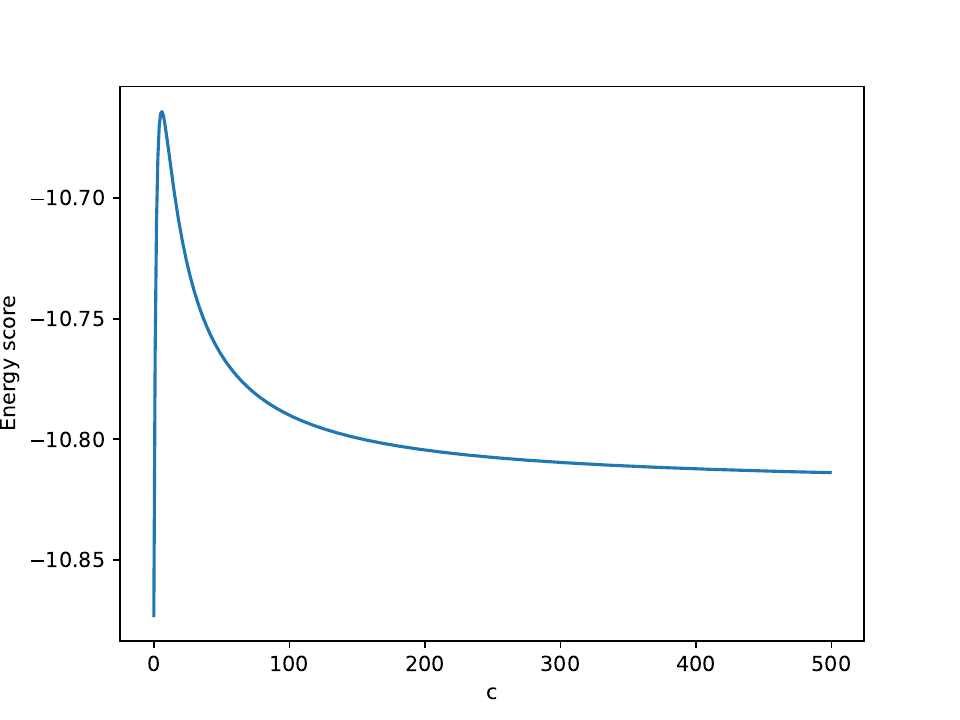} 
        \label{fig:graph9}
    \end{subfigure}\vspace{-5mm}
    \caption{Plots of $c$ against LOOCV energy score.}
    \label{fig:6}
\end{figure}

The posterior distribution obtained by mixture MP demonstrates significant variation from that obtain by BB (Fig. \ref{fig:10}, \ref{fig:13}) irrespective of the value of $c$. In particular, the posterior distribution of the quantile obtained from the mixture MP is continuous and smooth, whereas the BB yields a discrete distribution with posterior masses at several points. 
\begin{figure}[!ht]
      \begin{subfigure}[t]{0.44\textwidth}
        \centering
        \includegraphics[width=\linewidth]{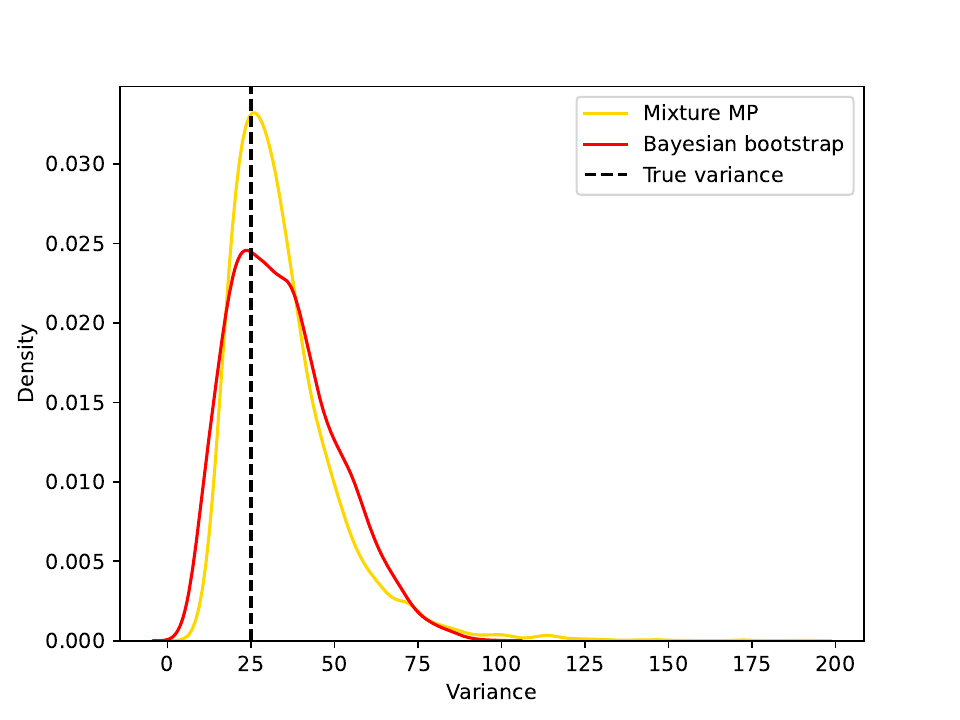} 
        \label{fig:graph28}
    \end{subfigure}
    \centering
      \begin{subfigure}[t]{0.44\textwidth}
        \centering
        \includegraphics[width=\linewidth]{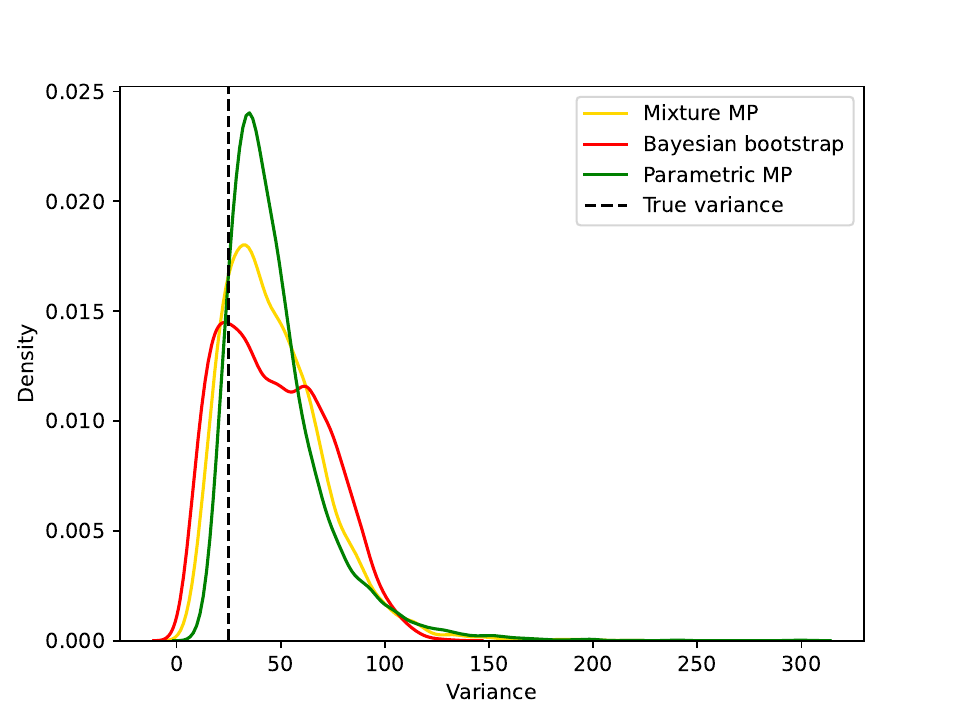} 
        \label{fig:graph29}
    \end{subfigure}\vspace{-5mm}
    \caption{Posterior distribution of variance for $c=\infty$ (Left) and $c=6$ (Right).}
    \label{fig:10}
\end{figure}\vspace{-5mm}
\begin{figure}[!ht]
    \begin{subfigure}[t]{0.44\textwidth}
        \centering
        \includegraphics[width=\linewidth]{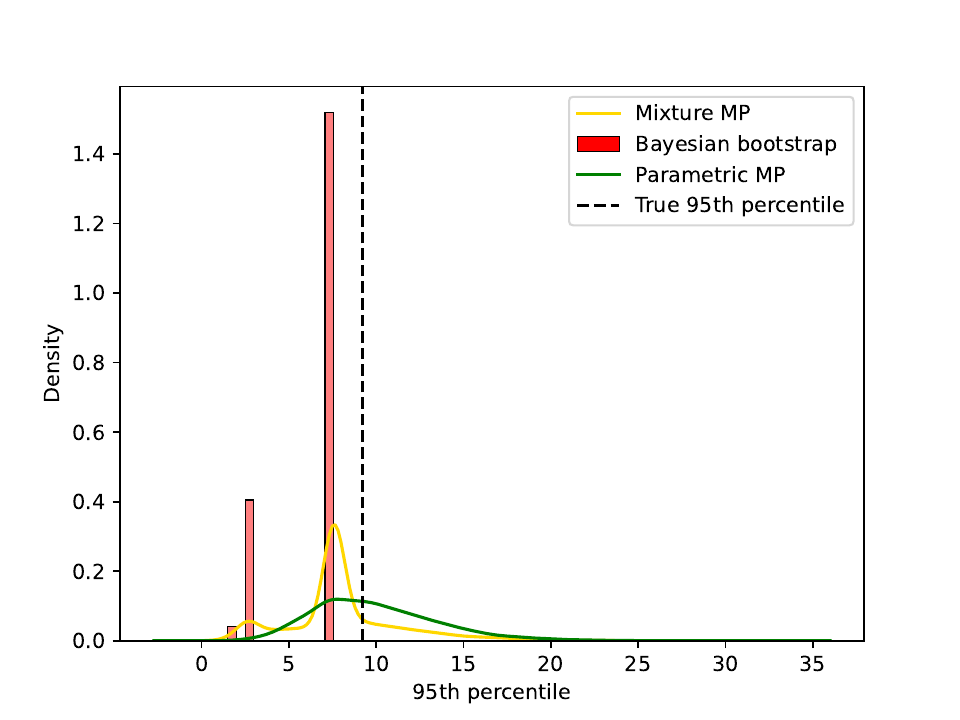} 
        \label{fig:graph30}
    \end{subfigure}
    \centering
    \begin{subfigure}[t]{0.44\textwidth}
        \centering
        \includegraphics[width=\linewidth]{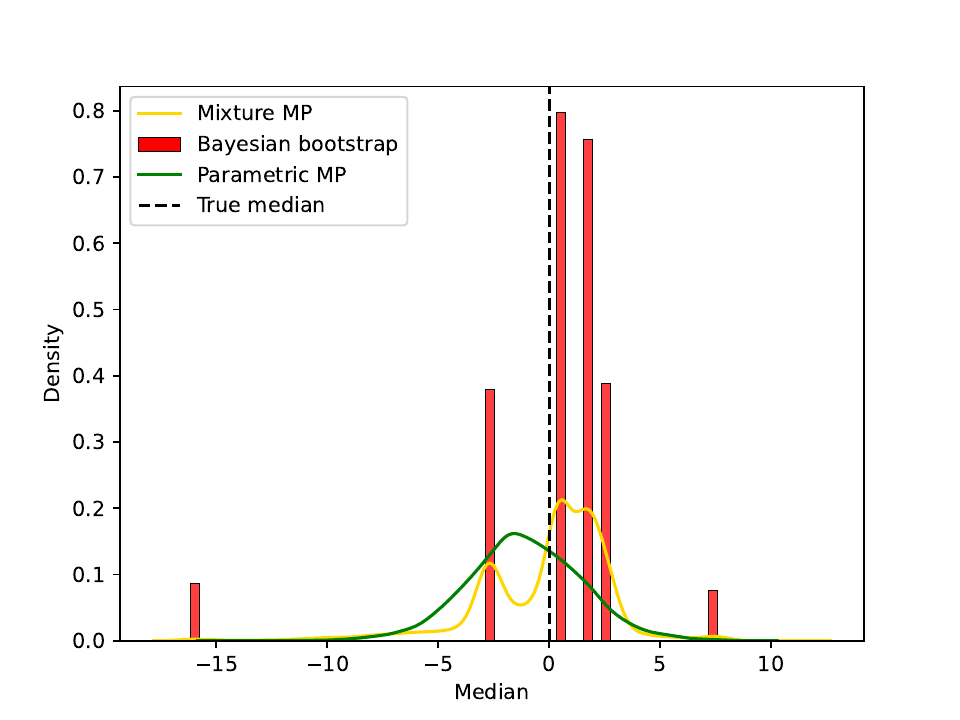} 
        \label{fig:graph31}
    \end{subfigure}\vspace{-2mm}
    \caption{Posterior distribution of (Left) 95th percentile and (Right) median; $c=6$.}
    \label{fig:13}
\end{figure}

The run-times of all the simulation studies, including those presented in the main article, are summarized in the following table.

\begin{table}[h]
\centering
    \scriptsize
    \begin{tabular}{c c c c }
        \hline
         \textbf{Run-times (seconds)} & \textbf{Mixture MP} & \textbf{Parametric MP} & \textbf{BB}\\ \hline
        \textbf{Moderate $n$, well-specified} &0.215 & 0.169 & 0.006\\ 
        \textbf{Moderate $n$, misspecified} & 1.027 & 1.297 & 0.009 \\
        \textbf{Small $n$, well-specified} &0.477 & 0.505 & 0.002\\
        \textbf{Small $n$, misspecifed} &0.477 & 0.505 & 0.002\\ \hline
    \end{tabular}
    \caption{Run-times of all simulations.}
\end{table}

\subsection{Logistic regression}\label{sec:app_logreg}
As stated in the main article, we include posterior plots of some of the remaining coefficients for reference. Consistent with previous observations, the posterior distributions do not exhibit too much variation between the methods, as none of the three methods incorporate prior information regarding the regression coefficients. The mixture MP nonetheless tends to be closer to the parametric MP due to the large value of $c = 2000$.
\begin{figure}[ht]
    \begin{subfigure}[t]{0.45\textwidth}
        \centering
        \includegraphics[width=\linewidth]{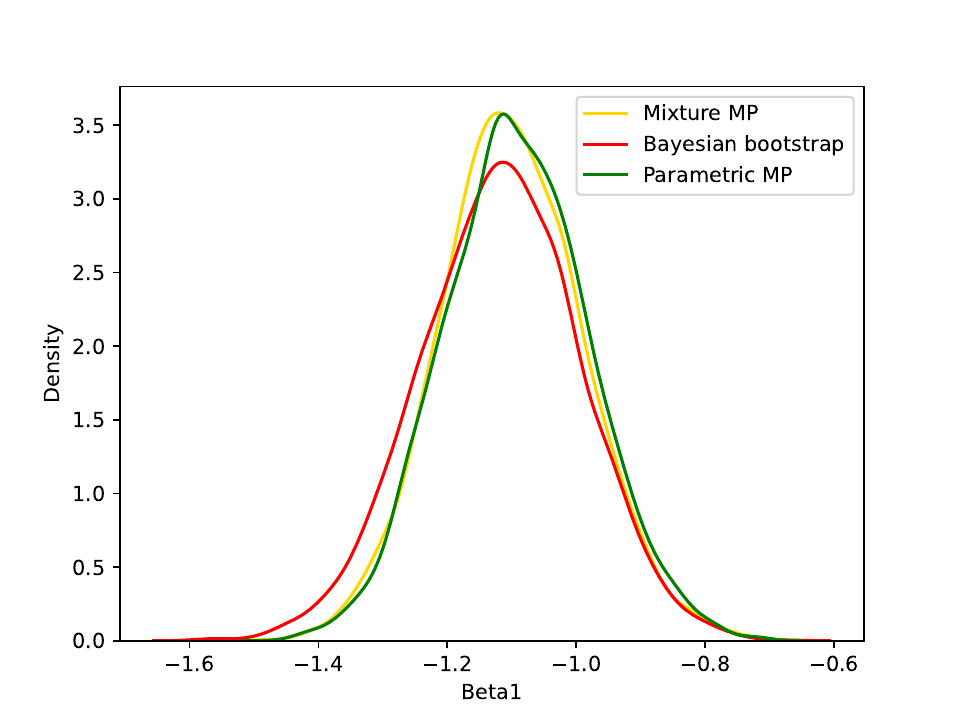} 
    \end{subfigure}
    \begin{subfigure}[t]{0.45\textwidth}
        \centering
        \includegraphics[width=\linewidth]{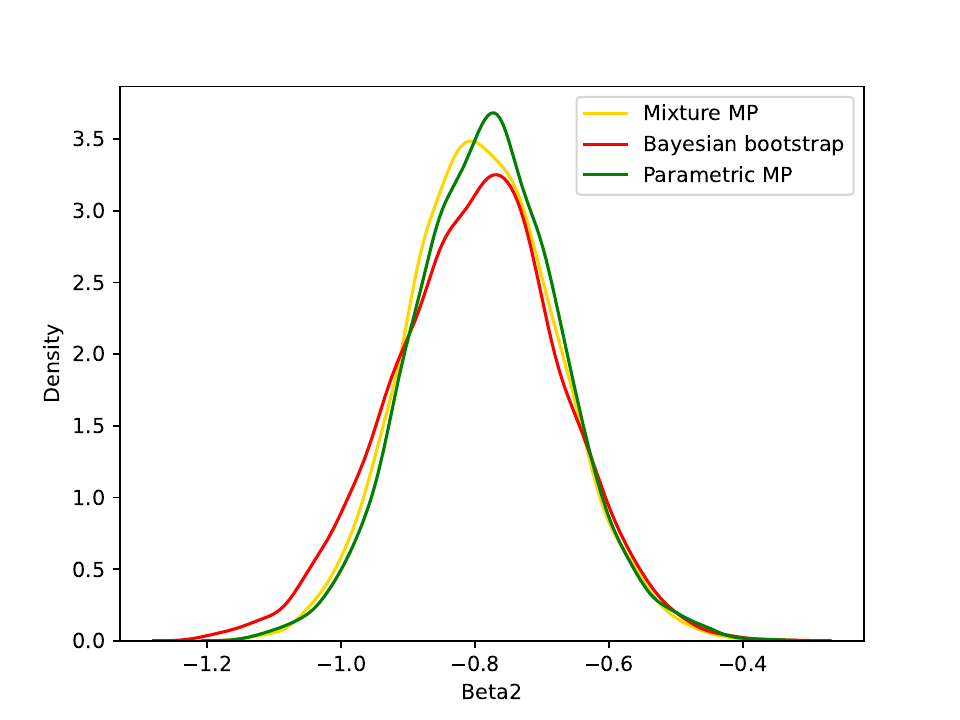} 
    \end{subfigure}\\
    \centering
    \begin{subfigure}[t]{0.45\textwidth}
        \centering
        \includegraphics[width=\linewidth]{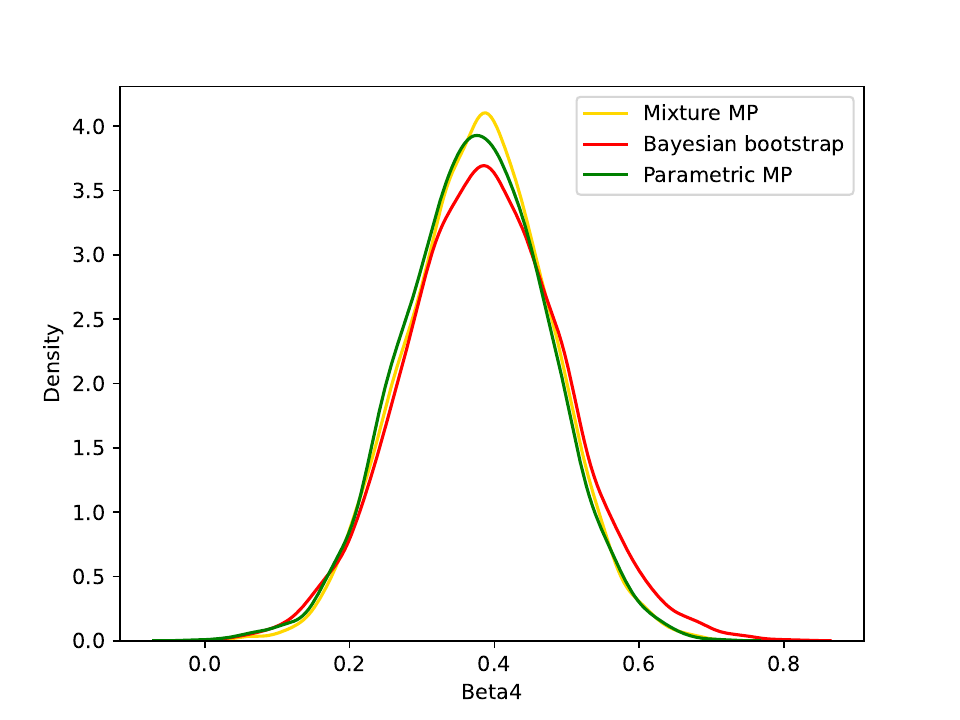} 
    \end{subfigure}
    \begin{subfigure}[t]{0.45\textwidth}
        \centering
        \includegraphics[width=\linewidth]{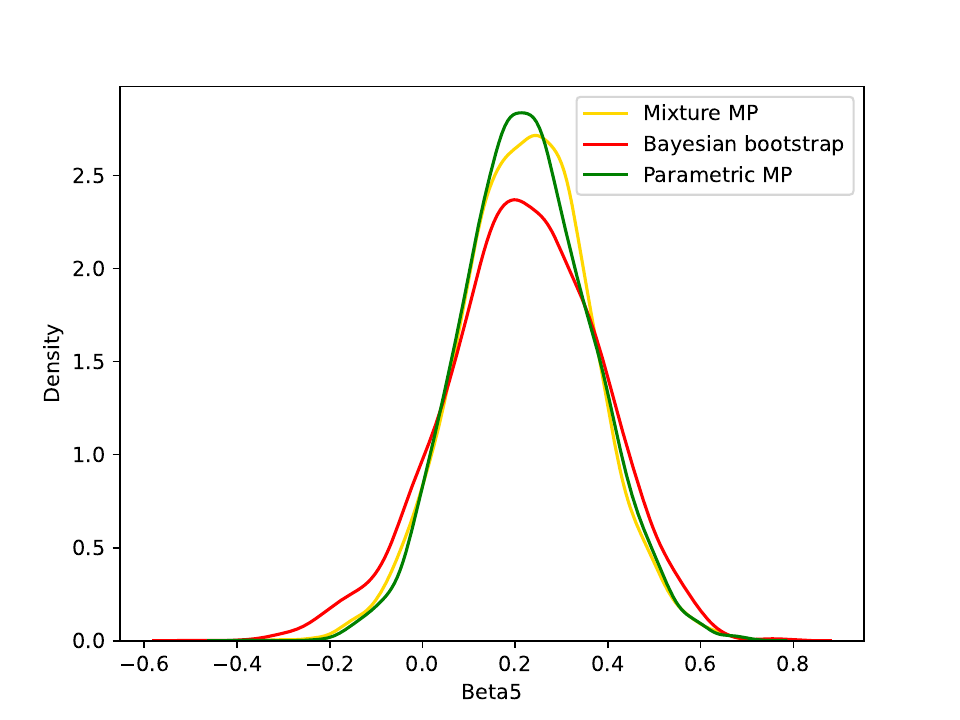} 
    \end{subfigure}
    \caption{Posterior distribution of logistic regression coefficients ($\beta_1,\beta_2,\beta_4,\beta_5$).}
    \end{figure}

\end{document}